\documentclass[twocolumn,superscriptaddress,pra,aps,10pt,longbibliography]{revtex4-2}
\usepackage{amssymb}
\usepackage{color}
\usepackage{amsmath}
\usepackage{mathrsfs}
\usepackage{amsbsy}
\usepackage{amsthm}
\usepackage{graphicx}
\usepackage{lipsum} 
\usepackage[OT1]{fontenc}
\usepackage{algorithm}
\usepackage{algpseudocode}

\usepackage{tikz}
\usetikzlibrary{quantikz}

\usepackage{bbm}
\usepackage{bm}
\usepackage{epsfig}
\usepackage{xfrac}
\usepackage{xcolor}
\usepackage{enumerate}
\usepackage[shortlabels]{enumitem}

\usepackage{thmtools}
\usepackage{thm-restate}

\definecolor{myurlcolor}{rgb}{0,0,0.7}
\definecolor{myrefcolor}{rgb}{0.8,0,0}
\usepackage[unicode=true,pdfusetitle, bookmarks=true,bookmarksnumbered=false,bookmarksopen=false, breaklinks=false,pdfborder={0 0 0},backref=false,colorlinks=true, linkcolor=myrefcolor,citecolor=myurlcolor,urlcolor=myurlcolor]{hyperref}

\newcommand{\cross}{\mathbin{\tikz [x=1.6ex,y=1.6ex,line width=.15ex] \draw (0,-0.2) -- (1,0.8) (0,0.8) -- (1,-0.2);}}%

\usepackage{braket}
\renewcommand{\v}[1]{\ensuremath{\mathbf{#1}}} 
\newcommand{\gv}[1]{\ensuremath{\text{\boldmath$ #1 $}}}
\newcommand{\abs}[1]{\left| #1 \right|} 
\newcommand{\norm}[1]{\left\| #1 \right\|} 
\newcommand{\trace}{\mathrm{Tr}}

\newcommand{\expr}{\mathrm{expr}} 

\newcommand{\tb}{{\textsc{b}}}

\newcommand{\mS}{{\mathcal{S}}}

\newcommand{\mH}{{\mathcal{H}}}

\newcommand{\mM}{{\mathcal{M}}}

\newcommand{\mj}{{\mathrm{mj}}}

\newcommand{\frakn}{{\mathfrak{n}}}
\newcommand{\frakS}{{\mathfrak{S}}}
\newcommand{\frakK}{{\mathfrak{K}}}
\newcommand{\frakB}{{\mathfrak{B}}}

\newcommand{\frakm}{{\mathfrak{m}}}

\newcommand{\mD}{{\mathcal{D}}}
\newcommand{\mE}{{\mathcal{E}}}

\newcommand{\mT}{{\mathcal{T}}}
\newcommand{\id}{{\mathbbm{1}}}

\newcommand{\vx}{{\gv{x}}}
\newcommand{\vt}{{\gv{t}}}
\newcommand{\vm}{{\gv{m}}}
\newcommand{\vlambda}{{\gv{\lambda}_\theta}}
\newcommand{\vy}{{\gv{y}}}

\newcommand{\va}{{\gv{a}}}
\newcommand{\vb}{{\gv{b}}}
\newcommand{\vc}{{\gv{c}}}

\newcommand{\bR}{{\mathbb{R}}}
\newcommand{\sto}{\mathbb{S}_{d,D}}
\newcommand{\dsto}{\mathbb{S}_{D,D}^{\rm db}}
\newcommand{\bE}{{\mathbb{E}}}
\newcommand{\bC}{{\mathbb{C}}}

\newcommand{\bU}{{\mathbb{U}}}
\newcommand{\bP}{{\mathbb{P}}}

\renewcommand{\Re}{{\mathrm{Re}}}

\newcommand{\upp}{{\mathrm{upp}}}
\newcommand{\prob}{{\mathrm{Pr}}}

\newcommand{\congorequal}{\cong}
\newcommand{\appropto}{\mathrel{\vcenter{
  \offinterlineskip\halign{\hfil$##$\cr
    \propto\cr\noalign{\kern2pt}\sim\cr\noalign{\kern-2pt}}}}}
\newcommand{\dket}[1]{\vert {#1} \rangle \! \rangle} 
\newcommand{\dbra}[1]{\langle \! \langle {#1} \vert} 

\let\baraccent=\= 
\renewcommand{\=}[1]{\stackrel{#1}\cong} 

\newcommand{\thmref}[1]{\hyperref[#1]{Theorem~\ref{#1}}}
\newcommand{\propref}[1]{\hyperref[#1]{Proposition~\ref{#1}}}
\newcommand{\lemmaref}[1]{\hyperref[#1]{Lemma~\ref{#1}}}
\newcommand{\figref}[1]{\hyperref[#1]{Fig.~\ref{#1}}}
\newcommand{\tableref}[1]{\hyperref[#1]{Table~\ref{#1}}}
\newcommand{\figaref}[1]{\hyperref[#1]{Fig.~\ref{#1}a}}
\newcommand{\figbref}[1]{\hyperref[#1]{Fig.~\ref{#1}b}}
\newcommand{\figcref}[1]{\hyperref[#1]{Fig.~\ref{#1}c}}
\newcommand{\figdref}[1]{\hyperref[#1]{Fig.~\ref{#1}d}}
\newcommand{\figeref}[1]{\hyperref[#1]{Fig.~\ref{#1}e}}
\renewcommand{\eqref}[1]{\hyperref[#1]{Eq.~(\ref{#1})}}
\newcommand{\secref}[1]{\hyperref[#1]{Sec.~\ref{#1}}}
\newcommand{\secsref}[2]{\hyperref[#1]{Sec.~\ref{#1}-\ref{#2}}}
\newcommand{\eqsref}[2]{\hyperref[#1]{Eqs.~(\ref{#1})-(\ref{#2})}}
\newcommand{\appref}[1]{\hyperref[#1]{Appx.~\ref{#1}}}

\usepackage{dsfont}

\newtheorem{theorem}{Theorem}
\newtheorem{proposition}[theorem]{Proposition}
\newtheorem{lemma}[theorem]{Lemma}

\definecolor{ginger}{rgb}{0.69, 0.4, 0.0}
\definecolor{blue-pigment}{rgb}{0.2, 0.2, 0.6}
\definecolor{blue-violet}{rgb}{0.54, 0.17, 0.89}

\begin{document}

\title{Optimal protocols for quantum metrology with noisy measurements}

\author{Sisi Zhou}\email{sisi.zhou26@gmail.com}
\affiliation{Institute for Quantum Information and Matter, California Institute of Technology, Pasadena, CA 91125, USA}
\affiliation{Perimeter Institute for Theoretical Physics, Waterloo, Ontario N2L 2Y5, Canada}

\author{Spyridon Michalakis}\email{spiros@caltech.edu}
\affiliation{Institute for Quantum Information and Matter, California Institute of Technology, Pasadena, CA 91125, USA}

\author{Tuvia Gefen}\email{tgefen@caltech.edu}
\affiliation{Institute for Quantum Information and Matter, California Institute of Technology, Pasadena, CA 91125, USA}

\date{\today}

\begin{abstract}

Measurement noise is a major source of noise in quantum metrology. Here, we explore preprocessing protocols that apply quantum controls to the quantum sensor state prior to the final noisy measurement (but after the unknown parameter has been imparted), aiming to maximize the estimation precision. We define the quantum preprocessing-optimized Fisher information, which determines the ultimate precision limit for quantum sensors under measurement noise, and conduct a thorough investigation into optimal preprocessing protocols. First, we formulate the preprocessing optimization problem as a biconvex optimization using the error observable formalism, based on which we prove that unitary controls are optimal for pure states and derive analytical solutions of the optimal controls in several practically relevant cases. Then we prove that for classically mixed states (whose eigenvalues encode the unknown parameter) under commuting-operator measurements, coarse-graining controls are optimal, while unitary controls are suboptimal in certain cases. Finally, we demonstrate that in multi-probe systems where noisy measurements act independently on each probe, the noiseless precision limit can be asymptotically recovered using global controls for a wide range of quantum states and measurements. Applications to noisy Ramsey interferometry and thermometry are presented, as well as explicit circuit constructions of optimal controls.

\end{abstract}

\maketitle

\section{Introduction} 
\label{sec:intro}

Quantum metrology is one of the pillars of quantum science and technology~\cite{degen2017quantum, pezze2018quantum, pirandola2018advances, giovannetti2011advances,giovannetti2006quantum}. This field deals with fundamental precision limits of parameter estimation imposed by quantum physics. Notably, it seeks to use non-classical effects to enhance the estimation precision of unknown parameters in quantum systems, which has led to the development of improved sensing protocols in various experimental platforms~\cite{caves1981quantum,leibfried2004toward,mitchell2004super,tsang2016quantum,kaubruegger2021quantum,marciniak2022optimal}.
To characterize the metrological limit of quantum sensors, the quantum Cram\'{e}r--Rao bound (QCRB)~\cite{holevo1982probabilistic,helstrom1976quantum}, which is saturable for large number of experiments, is conventionally used. It is defined using the quantum Fisher information (QFI)~\cite{braunstein1994statistical,liu2019quantum,petz1996geometries}, which is one of the most useful and celebrated tools in quantum metrology, with a considerable amount of research focused on developing better ways to calculate and bound it~\cite{fujiwara2008fibre,escher2011general,demkowicz2012elusive,yuan2017fidelity,sone2021generalized,altherr2021quantum}. 

Although the QCRB and the QFI apply extensively in quantum sensing, they are defined assuming that arbitrary quantum measurements can be applied on quantum states to extract information about the unknown parameter. However, in actual experimental platforms, such as nitrogen-vacancy centers~\cite{dutt2007quantum,hanson2008coherent,taylor2008high,jiang2009repetitive,doherty2013nitrogen,unden2016quantum,schmitt2021optimal}, superconducting qubits~\cite{krantz2019quantum}, trapped ions~\cite{myerson2008high,bruzewicz2019trapped}, and more, measurements are often noisy and time-expensive, rendering the sensitivity of practical quantum devices far from the theoretical limits given by the QCRB. 
In particular, measurement noise remains a significant source of noise in quantum sensing experiments. Other sources of noise, such as system evolution and state preparation, have been studied extensively, with methods developed to mitigate their effect~\cite{fujiwara2008fibre,viola1999dynamical,
kessler2014quantum,dur2014improved,arrad2014increasing,liu2017quantum,yamamoto2021error,albarelli2019restoring,verstraete2009quantum,jiang2009preparation,johnsson2020geometric,zhou2018achieving,zhou2021asymptotic,liu2023optimal}. 

To tackle the effect of measurement noise on quantum metrology, interaction-based readouts were proposed~\cite{davis2016approaching,frowis2016detecting,macri2016loschmidt,nolan2017optimal,haine2018using,koppenhofer2023squeezed} and demonstrated experimentally~\cite{hosten2016quantum,linnemann2016quantum,li2023improving}, where bespoke inter-particle interactions that enhance phase estimation precision in spin ensembles are applied before the noisy measurement step and after the probing step. 
The idea of employing unitary controls in a preprocessing manner, i.e. after the unknown parameter has been imparted but prior to the final measurement, was later formulated as the imperfect (or noisy) QFI problem~\cite{haine2018using,len2021quantum}, where the preprocessing is optimized over all unitary operations. 
Classical post-processing methods, such as measurement error mitigation~\cite{maciejewski2020mitigation,geller2020rigorous,bravyi2021mitigating}, can then work in complement to the quantum preprocessing method for parameter estimation under noisy measurements.

Apart from a few specific cases, such as 
qubit sensors with lossy photon detection~\cite{len2021quantum}, setting the metrological limit under measurement noise by computing imperfect QFI has been difficult, limiting its practical application. In this work, we propose a more general measurement optimization scheme, where arbitrary quantum controls (i.e., general quantum channels that can be implemented utilizing unitary gates and ancillas) are applied before the noisy measurement. The goal is to identify the FI optimized over all quantum preprocessing channels for general quantum states and measurements, that we call the quantum preprocessing-optimized FI (QPFI) and quantifies the ultimate power of quantum sensors with measurement noise, and to obtain the corresponding optimal controls, that can be applied to achieve the optimal sensitivity in practical experiments. 

We systematically study the QPFI, along with the corresponding optimal preprocessing controls in this work. 
In \secref{sec:def}, we first define the QPFI and review related concepts. We then introduce the concept of {\it error observables} in \secref{sec:opt}, and use it to demonstrate that the QPFI problem can be cast as a biconvex optimization problem~\cite{gorski2007biconvex}. In turn, this allows us to find analytical conditions for optimality, and to identify optimal controls saturating the QPFI in the setting of commuting measurements applied to pure states (see \secref{sec:pure}). The case of classically mixed states (i.e., states for which the unknown parameter is encoded in the eigenvalues) is studied in \secref{sec:classical}. Besides analytical solutions, we also manage to prove that unitary controls are optimal for pure states under general measurements, and that coarse-graining controls are optimal for classically mixed states under commuting-operator measurements, with a counterexample illustrating the non-optimality of unitary controls. 
For general mixed states, we further prove useful bounds on the QPFI in \secref{sec:general}. In terms of the asymptotic behavior of identical local measurements acting on multi-probe systems, 
in \secref{sec:asymptotic}, we identify a sufficient condition for the convergence of the QPFI to the QFI using an optimal encoding protocol based on the Holevo--Schumacher--Westmoreland (HSW) theorem~\cite{holevo1998capacity,schumacher1997sending}. We show that the relevant condition is satisfied by a generic class of quantum states, including low-rank states, permutation-invariant states, and Gibbs states (with an unknown temperature), while previously only the pure state case was proven~\cite{len2021quantum}. 

Our results provide a theoretically-accessible precision bound for quantum metrology under noisy measurements, along with a roadmap towards preprocessing optimization in sensing experiments.

\section{Definitions}
\label{sec:def}

Given a quantum state $\rho_\theta$ as a function of an unknown parameter $\theta$, the procedure to estimate $\theta$ goes as follows (see \figaref{fig:measurement}): (1) Perform a quantum measurement $\{M_i\}$ on $\rho_\theta$, which gives a measurement outcome $i$ with probability $p_{i,\theta} = \trace(\rho_\theta M_i)$; (2) Infer the value of $\theta$ using an estimator $\hat{\theta}$, which is a function of the measurement outcome $i$; (3) Repeat the above two steps multiple times and use the average of $\hat{\theta}$ over many trials as the final estimate of $\theta$.  
Here, the quantum measurement $\{M_i\}$ is mathematically formulated as a positive operator-valued measure (POVM)~\cite{nielsen2002quantum} that satisfies $M_i \geq 0$ and $\sum_i M_i = \id$ (we use $A \geq 0$ to indicate an operator $A$ that is positive semidefinite). 
We also assume in this work that $\rho_\theta$ and $M_i$ lie in finite-dimensional Hilbert spaces, with measurement outcomes contained in a finite set.

\begin{figure}[tb]
    \centering
    \includegraphics[width=0.45\textwidth]{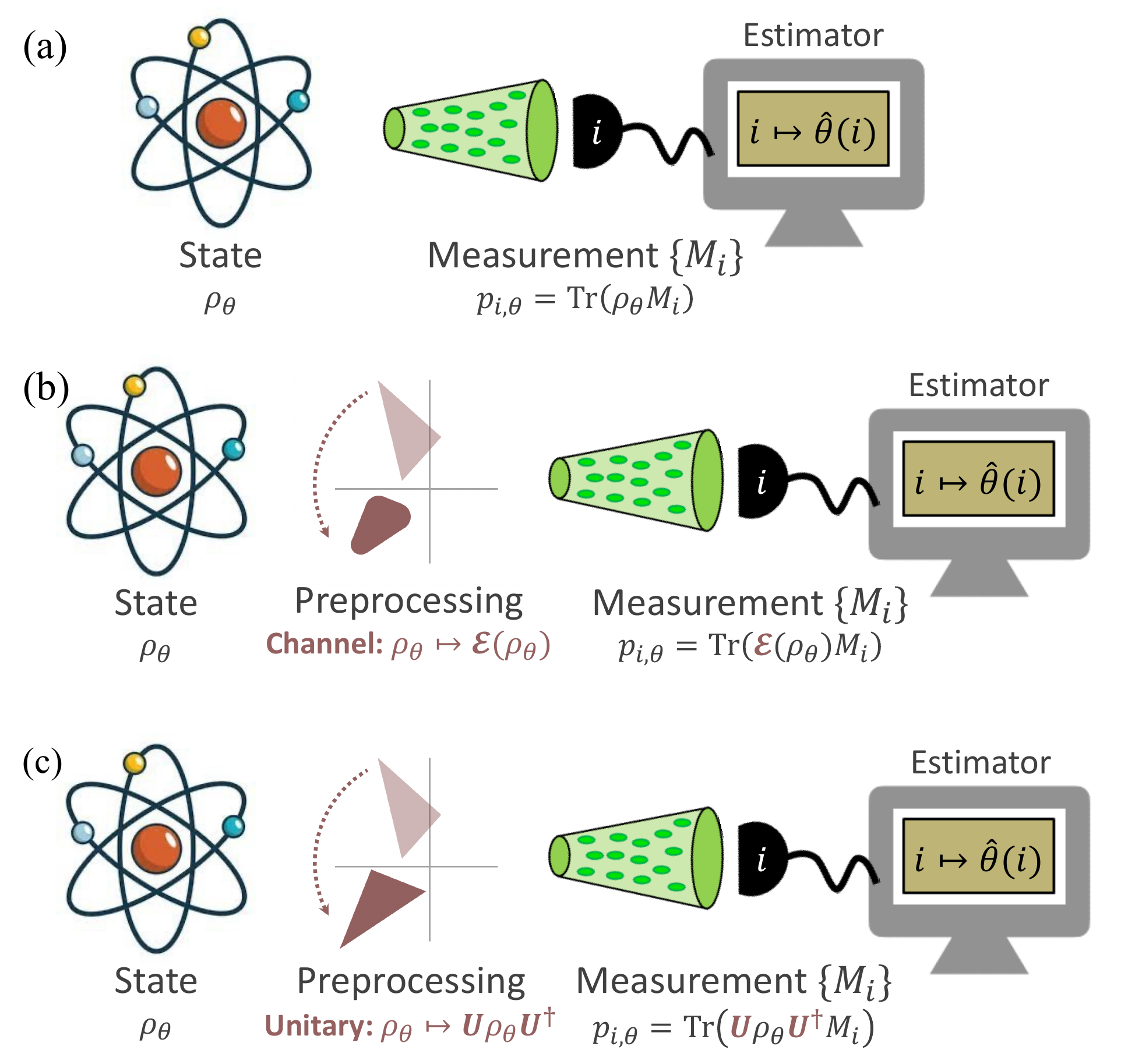}
    \caption{(a)~Standard parameter estimation procedure of a quantum state $\rho_\theta$ using a quantum measurement $\{M_i\}$. The estimation of $\theta$ is through an unbiased estimator $\hat{\theta}$ as a function of measurement outcomes $i$. The CRB states $\Delta \hat \theta \geq 1/\sqrt{N_\expr F(\rho_\theta,\{M_i\})}$. 
    (b)~Preprocessing protocols where the measurement device is fixed, and the quantum control acting before the measurement is optimized over all quantum channels. The CRB states $\Delta \hat \theta \geq 1/\sqrt{N_\expr  F^\bP(\rho_\theta,\{M_i\})}$. (c)~Preprocessing protocols where the measurement device is fixed, and the quantum control acting before the measurement is optimized over all unitary channels. The CRB states $\Delta \hat \theta \geq 1/\sqrt{N_\expr  F^\bU(\rho_\theta,\{M_i\})}$.  
    Different types of FIs discussed in this work satisfy $F(\rho_\theta,\{M_i\}) \leq F^\bU(\rho_\theta,\{M_i\}) \leq F^\bP(\rho_\theta,\{M_i\}) \leq  J(\rho_\theta)$ (and each inequality can be strict). 
    }
    \label{fig:measurement}
\end{figure}

In estimation theory, the Cram\'er--Rao bound (CRB)~\cite{kay1993fundamentals,lehmann2006theory,kobayashi2011probability} provides a lower bound on the estimation error for any locally unbiased estimator $\hat{\theta}$ at a local point $\theta_0$ where $\rho_\theta$ is differentiable, satisfying 
\begin{equation}
\label{eq:locally-unbiased}
\bE[\hat{\theta}|\theta_0] = \theta_0,\quad \text{and} \quad 
\frac{\partial}{\partial \theta}\bE[\hat{\theta}|\theta] \bigg|_{\theta=\theta_0} = 1,
\end{equation}
where we use $\bE[\cdot|\theta]$ to denote the conditional expectation over the probability distribution $\{p_{i,\theta}\}$. The above condition indicates that locally unbiased estimators $\hat\theta$ provide an unbiased estimation of $\theta$ at the point $\theta_0$, which is also precise up to first order in its neighborhood. Note that in the following we will implicitly use $\bE[\cdot]$ to represent $\bE[\cdot|\theta]$ and consider locally unbiased estimators at a local point $\theta$.
The CRB states that the estimation error $\Delta\hat\theta$ (i.e., the standard deviation of the estimator $\hat{\theta}$) has the following lower bound: 
\begin{equation}
\label{eq:CRB}
    \Delta\hat\theta := (\bE[(\hat{\theta}-\theta)^2])^{\frac{1}{2}} \geq \frac{1}{\sqrt{N_{\rm expr} F(\rho_{\theta},\{M_i\})}}, 
\end{equation}
where $N_{\rm expr}$ is the number of experiments performed, and $F(\rho_{\theta},\{M_i\})$ is the FI of the probability distribution $\{p_{i,\theta} = \trace(\rho_\theta M_i)
\}$~\cite{kay1993fundamentals,lehmann2006theory,kobayashi2011probability}, defined by 
\begin{equation}
\label{eq:FI}
    F(\rho_{\theta},\{M_i\}):= \sum_{i:\trace(\rho_\theta M_i)\neq 0}\frac{(\trace(\partial_\theta\rho_\theta M_i))^2}{\trace(\rho_\theta M_i)}. 
\end{equation}
The CRB is often saturable asymptotically (i.e., when $N_{\rm expr} \rightarrow \infty$) using the maximum likelihood estimator~\cite{kay1993fundamentals,lehmann2006theory,kobayashi2011probability} and therefore the FI, which is inversely proportional to the variance of the estimator, serves as a good measure of the degree of sensitivity of $\{p_{i,\theta}\}$ with respect to $\theta$. One caveat is the CRB only applies to locally unbiased estimators and can be violated by biased estimators. Additionally, there exist singular cases where maximum likelihood estimators are no longer necessarily asymptotically unbiased, e.g., when the support of $\{p_{i,\theta}\}$ varies in the neighborhood of $\theta$, and the CRB may not apply to them~\cite{seveso2019discontinuity}. However, for self-consistency, this paper will focus only on optimizing the FI, regardless of the limitations of the CRB. 

The QFI of $\rho_\theta$ is the FI maximized over all possible quantum measurements on $\rho_\theta$ (see \appref{app:QFI-attainable} for further details) and we will refer to the optimal measurements as QFI-attainable measurements. 
Formally, the QFI is defined by~\cite{helstrom1976quantum,holevo1982probabilistic,braunstein1994statistical}
\begin{equation}
\label{eq:QFI}
    J(\rho_\theta) = \max_{\{M_i\}} F(\rho_{\theta},\{M_i\}), 
\end{equation}
giving rise to the QCRB 
\begin{equation}
\label{eq:QCRB}
    \Delta\hat\theta \geq \frac{1}{\sqrt{N_{\rm expr} J(\rho_\theta)}}, 
\end{equation}
which characterizes the ultimate lower bound on the estimation error. Going forward, we will also overload the notation and write 
\begin{equation}
\label{eq:classical-FI}
     J(\{p_{i,\theta}\}):= \sum_{i:p_{i,\theta}\neq 0} \frac{(\partial_\theta p_{i,\theta})^2}{p_{i,\theta}} , 
\end{equation} 
to denote the FI of a classical probability distribution $\{p_{i,\theta}\}$, satisfying $p_{i,\theta} \geq 0$ and $\sum_i p_{i,\theta} = 1$. Note that, from now on, we will implicitly assume that the summation is taken over terms with non-zero denominators. 

In practice, the optimal measurements achieving the QFI are not always implementable, restricting the range of applications of the QCRB. For example, the projective measurement onto the basis of the symmetric logarithmic operators, which is usually a correlated measurement among multiple probes, is known to be optimal~\cite{braunstein1994statistical}, while quantum measurements in experiments are usually noisy and not exactly projective. Here, we consider a metrological protocol in which arbitrary quantum controls can be implemented, after the unknown parameter $\theta$ has been imparted to the quantum sensor state $\rho_\theta$ and before a fixed quantum measurement is performed (see \figbref{fig:measurement}). We call this additional step ``preprocessing'', ``pre-measurement-processing'' in full. Note that the idea of implementing preprocessing quantum controls to improve sensitivity goes beyond the FI formalism and applies to other figures of merit of quantum sensors~\cite{alderete2022inference}. This model effectively describes quantum experiments where the measurement error is dominant, while the gate implementation error and the state preparation error is relatively small, a noise model that arises naturally in modern quantum devices such as nitrogen-vacancy centers~\cite{dutt2007quantum,hanson2008coherent,taylor2008high,jiang2009repetitive} and superconducting qubits~\cite{krantz2019quantum}. 

To quantify the sensitivity of estimating $\theta$ on $\rho_\theta$ with the measurement $\{M_i\}$ fixed, we define the {FI optimized over all preprocessing quantum channels}, or the \emph{quantum preprocessing-optimized Fisher information (QPFI)}, to be 
\begin{equation}
\label{eq:FI-A}
    F^\bP(\rho_{\theta},\{M_i\}) = \sup_{\mE}F(\mE(\rho_{\theta}),\{M_i\}),
\end{equation}
where $\mE$ is an arbitrary quantum channel (or a CPTP map~\cite{watrous2018theory}). See~\appref{app:property} for mathematical properties of the QPFI. In particular, when the  quantum measurement is fixed, the CRB induced by the QPFI, i.e., 
\begin{equation}
\label{eq:CRB-A}
    \Delta\hat\theta \geq \frac{1}{\sqrt{N_\expr F^\bP(\rho_{\theta},\{M_i\})}},
\end{equation}
provides a practical and tighter Cram\'{e}r--Rao-type bound, compared to the QCRB, for parameter estimation under noisy measurements. We assume in the following discussions that all measurements are non-trivial (i.e., $\exists M_i \not\propto \id$, for all $\{M_i\}$) and $\partial_\theta \rho_\theta \neq 0$ so that the QPFI is always positive. 

Unless stated otherwise, we will denote the systems that $\rho_\theta$ and $\{M_i\}$ act on by $\mH_S$ and $\mH_{S'}$, respectively, and we will refer to $\mH_S$ as the input system and $\mH_{S'}$ as the output system.  We do not assume $\mH_S \congorequal \mH_{S'}$ here.  
This broader context is of particular interest when the quantum state $\rho_\theta$ cannot be directly measured (e.g., 
readout of superconducting qubits 
via a resonator~\cite{krantz2019quantum} 
and readout of nuclear spins via an electron spin in a nitrogen-vacancy center~\cite{gefen2018quantum, cujia2019tracking,pfender2019high,cohen2020achieving}); or when the quantum state is restricted to a subsystem of the entire system while quantum measurement can be performed globally. 

Note that for generic noisy measurements, the supremum in \eqref{eq:FI-A} is usually attainable, i.e., there exists an optimal $\mE$ such that $F(\mE(\rho_{\theta}),\{M_i\})$ is maximized (see \appref{app:singular}). However, there exist singular cases where $F(\mE(\rho_{\theta}),\{M_i\})$ has no maximum, due to the singularity of the FI at the point $\trace(\mE(\rho_\theta) M_i) = 0$ (see \secref{sec:binary} for an example). In such cases, there still exist near-optimal quantum controls that attain $\sup_{\mE}F(\mE(\rho_{\theta}),\{M_i\}) - \eta$ for any small $\eta > 0$. In fact, we prove in \appref{app:singular} that: 
\begin{restatable}{theorem}{thmattain}
\label{thm:attainability}
Let $M_i^{(\epsilon)} = (1-\epsilon)M_i + \epsilon\trace(M_i)\frac{\id}{d}$, where $d = \dim(\mH_{S'})$ and $0 < \epsilon < 1$. Then 
\begin{equation}
    F^\bP(\rho_{\theta},\{M_i\}) = \lim_{\epsilon \rightarrow 0^+} F^\bP(\rho_{\theta},\{M_i^{(\epsilon)}\}), 
\end{equation}
and the QPFI $F^\bP(\rho_{\theta},\{M_i^{(\epsilon)}\})$ is attainable for any $\epsilon \in (0,1]$.
\end{restatable}

In the following, we will focus mostly on the case where the QPFI is attainable. We will discuss the behavior of the QPFI, exploring numerical optimization algorithms and analytical solutions to the optimal controls for certain practically relevant quantum states and measurements. 

We will also examine the FI optimized over all unitary preprocessing channels, which we call the \emph{quantum unitary-preprocessing-optimized Fisher information (QUPFI)}~\cite{haine2018using,len2021quantum}
\begin{equation}
\label{eq:FI-U}
    F^\bU(\rho_{\theta},\{M_i\}) = \sup_{U}F(U\rho_{\theta}U^\dagger,\{M_i\}),
\end{equation}
where $U$ is an arbitrary unitary gate.  
(Note that our QUPFI is the same as the imperfect QFI in~\cite{len2021quantum}.)
Unlike the QPFI, we assume $\mH_{S'} \congorequal \mH_{S}$ (and do not distinguish between $S'$ and $S$) when we talk about the QUPFI, so that it is well defined. We note here that \thmref{thm:attainability} holds for the QUPFI, as well.

The optimal preprocessing controls that attain the QPFI and the QUPFI usually depend on $\theta$, whose value should be roughly known before the experiment. Otherwise, one might use the two-step method by first using $\sqrt{N_{\rm expr}}$ states to obtain a rough estimate $\tilde\theta \approx \theta$, and then performing the optimal controls based on $\tilde \theta$ on the remaining $N_{\rm expr}-\sqrt{N_{\rm expr}}$ states~\cite{barndorff2000fisher,hayashi2011comparison,yang2019attaining}. The two-step procedure introduces a negligible amount of error asymptotically.

Before we proceed, we prove a relation between the QPFI and the QUPFI that will be useful later. 
\begin{proposition}
\label{prop:channel-unitary}
Let $\mH_{S}$ and $\mH_{S'}$ be the input and output systems of $\mE$. Suppose $\mH_{A_1}$ and $\mH_{A_2}$ are ancillary systems such that $\mH_{A_1}\otimes\mH_{S} \congorequal \mH_{A_2}\otimes\mH_{S'}$. If $\dim(\mH_{A_1}) \geq \dim(\mH_{S'})^2$ (or equivalently, $\dim(\mH_{A_2}) \geq \dim(\mH_{S})\dim(\mH_{S'})$), then 
\begin{multline}
\label{eq:channel-unitary}
    F^\bP\big((\rho_\theta)_{S},\{(M_i)_{S'}\}\big) =\\ F^\bU\big((\rho_\theta)_{S}\otimes\ket{0_{A_1}}\bra{0_{A_1}},\{(M_i)_{S'}\otimes\id_{A_2}\}\big),
\end{multline}
where we use subscripts to denote the systems the operators are acting on. 
\end{proposition}
\begin{proof}
Any quantum channel $\mE(\cdot) = \sum_{i=1}^{r_\mE} K_i (\cdot) K_i^\dagger$ from $\mH_S$ to $\mH_{S'}$ can be implemented by acting unitarily on $\mH_S$ and an ancillary system $\mH_{A_1}$, and then tracing over an auxiliary system $\mH_{A_2}$, if $\dim(\mH_{A_2}) \geq r_\mE$ (Stinespring's dilation~\cite{watrous2018theory}). For any quantum channel with the input system $\mH_S$ and the output system $\mH_{S'}$, there always exists a Kraus representation $\mE(\cdot) = \sum_{i=1}^{r_\mE} K_i (\cdot) K_i^\dagger$ such that $r_\mE \leq  \dim(\mH_{S'})\dim(\mH_{S})$~\cite{watrous2018theory}. Therefore, if $\dim(\mH_{A_2}) \geq \dim(\mH_{S'})\dim(\mH_{S})$, the unitary extension should exist. 

Let $\mH_{S}\otimes\mH_{A_1} \congorequal \mH_{S'}\otimes\mH_{A_2}$ be the enlarged, isomorphic input and output Hilbert spaces, respectively. If $\dim(\mH_{A_1}) \geq \dim(\mH_{S'})^2$, then $\dim(\mH_{A_2}) = \frac{\dim(\mH_{A_1}) \dim(\mH_{S})}{\dim(\mH_{S'})} \geq \dim(\mH_{S'})\dim(\mH_{S})$. Thus, there is a unitary $U_\mE$ mapping $\mH_{S}\otimes\mH_{A_1}$ to $\mH_{S'}\otimes\mH_{A_2}$ such that 
\begin{equation}
\label{eq:unitary-channel}
    \mE(\sigma) = \trace_{A_2} (U_\mE (\sigma\otimes\ket{0}\bra{0}) U_\mE^\dagger).
\end{equation}
From \eqref{eq:FI}, it follows that:
\begin{align*}
    F(\mE(\rho_\theta),\{M_i\}) &= F(\trace_{A_2}(U_\mE(\rho_\theta\otimes\ket{0}\bra{0})U_\mE^\dagger),\{M_i\})\\
    &= F(U_\mE(\rho_\theta\otimes\ket{0}\bra{0})U_\mE^\dagger,\{M_i\otimes \id\}), 
\end{align*}
where we omit the subscripts for simplicity. Note that the Stinespring's dilation technique is also useful in relating the QFI of a mixed state to the QFI of its purification in an extended Hilbert space~\cite{fujiwara2008fibre,escher2011general}.
Taking the supremum over $\mE$ in the above equality, we have 
\begin{equation}
F^\bP (\rho_\theta,\{M_i\}) \leq F^\bU(\rho_\theta\otimes\ket{0}\bra{0},\{M_i\otimes \id\}). 
\end{equation}
On the other hand, for any $U$ from $\mH_{S}\otimes\mH_{A_1}$ to $\mH_{S'}\otimes\mH_{A_2}$, $\trace_{A_2} (U ((\cdot)\otimes\ket{0}\bra{0}) U^\dagger)$ is a quantum channel from $\mH_{S}$ to $\mH_{S'}$, proving the other direction of \eqref{eq:channel-unitary}. 
\end{proof}

\section{Error observable formulation}
\label{sec:opt}

In this section, we will formalize the optimization of FI over quantum preprocessing controls as a biconvex optimization problem using the concepts of error observables. Using this new formulation, the preprocessing optimization problem becomes numerically tractable with standard algorithms for biconvex optimization~\cite{gorski2007biconvex}; and also analytically tractable for practically relevant quantum states (see \secref{sec:pure}). 

Here, we consider the preprocessing optimization problem in \eqref{eq:FI-A}. On the surface, it may appear from the definition of FI (\eqref{eq:FI}) that the target function $F(\mE(\rho),\{M_i\})$ is mathematically formidable. To simplify the target function, we introduce the \emph{error observable} $X$ and the \emph{squared error observable} $X_2$, defined by 
\begin{equation}
    X = \sum_i x_i M_i,\quad\text{and}\quad  X_2 = \sum_i x_i^2 M_i, 
\end{equation}
where $x_i$ is interpreted as the difference between the estimator value  $\hat{\theta}(i)$ and the true value $\theta$, i.e., $x_i = \hat\theta(i) - \theta$. We assume there are $r$ measurement outcomes and use $\vx$ to denote the vector $(x_1,\ldots,x_r)$. The local unbiasedness conditions (\eqref{eq:locally-unbiased}) for a single-shot measurement then become
\begin{equation}
\label{eq:locally-unbiased-2}
    \trace(\rho_\theta X) = 0,\quad \text{and}\quad  \trace(\partial_\theta\rho_\theta X) = 1. 
\end{equation}
It can be verified mathematically (which is essentially a proof of the CRB) that the minimum of the variance of the estimator under the local unbiasedness conditions is the inverse of the FI; that is, 
\begin{equation}
\label{eq:opt-x}
    F(\rho_\theta,\{M_i\})^{-1} = \min_{\vx} \,\trace(\rho_\theta X_2),
    \;\text{s.t.~}\; \text{\eqref{eq:locally-unbiased-2}.} 
\end{equation}
The problem above is a convex optimization over variables $\vx$, which can be solved using, e.g., the method of Lagrange multipliers~\cite{boyd2004convex}. The optimal solution to $\vx$ is 
\begin{equation}
\label{eq:x}
    x_i = \frac{\frac{\trace(\partial_\theta\rho_\theta M_i)}{\trace(\rho_\theta M_i)}}{\sum_{j:\trace(\rho_\theta M_j)\neq 0}\frac{(\trace(\partial_\theta\rho_\theta M_j))^2}{\trace(\rho_\theta M_j)}},
\end{equation}
when $\trace(\rho_\theta M_i) \neq 0$, and $x_i = 0$ when $\trace(\rho_\theta M_i) = 0$. 
Note that the error observable formulation was previously used to derive the QCRB~\cite{helstrom1968minimum}, where the QFI satisfies 
\begin{equation}
\label{eq:QFI-opt}
    J(\rho_\theta)^{-1} = \min_{X} \,\trace(\rho_\theta X^2),
    \;\text{s.t.~}\; \text{\eqref{eq:locally-unbiased-2}}, 
\end{equation}
and $X$ an arbitrary Hermitian matrix subject to the constraints in \eqref{eq:locally-unbiased-2}. This formulation has several useful applications~\cite{pezze2009entanglement,escher2012quantum,macieszczak2014bayesian}. In particular, an algorithm was proposed in \cite{len2021quantum} based on \eqref{eq:QFI-opt}, to optimize the QFI of quantum channels. 

Combining \eqref{eq:opt-x} and \eqref{eq:FI-A}, we have that 
\begin{align}
    F^\bP(\rho_\theta,\{M_i\})^{-1} = \inf_{(\vx,\mE)} \;& \trace(\mE(\rho_\theta) X_2),\label{eq:opt-x-E}\\
    \text{s.t.~}\; & 
    \trace(\mE(\rho_\theta) X) = 0,  \nonumber\\
    & \trace(\mE(\partial_\theta\rho_\theta) X) = 1. \nonumber 
\end{align}
Let $\mH_S$ and $\mH_{S'}$ be the input and output systems of $\mE$ and let $\{\ket{k}_{S}\}_{k=1}^{\dim(\mH_S)}$ and $\{\ket{j}_{S'}\}_{j=1}^{\dim(\mH_{S'})}$ be two sets of orthonormal basis of $\mH_S$ and $\mH_{S'}$, respectively. In the rest of this section, we use matrix representations of operators in the above bases. It is convenient to represent a CPTP map $\mE$ using a linear operator acting on $\mH_{S'}\otimes \mH_{S}$. Let $\mE(\cdot) = \sum_i K_i (\cdot) K_i^\dagger$ be the Kraus representation of $\mE$. Then, the linear operator $\Omega = \sum_i \dket{K_i}\dbra{K_i}$ is usually called the Choi matrix of $\mE$~\cite{watrous2018theory}, where $\dket{\star} := \sum_{jk} (\star)_{jk} \ket{j}_{S'}\ket{k}_{S}$ and $(\star)_{jk} = \bra{j}_{S'}(\star)\ket{k}_{S}$. $\Omega$ corresponds to a CPTP map if and only if $\Omega \geq 0$ and $\trace_{S'}(\Omega)=\id_{S}$. $\mE$ acting on any density operator $\sigma$ can be expressed using $\Omega$ through $\mE(\sigma) = \trace_{S} ((\id\otimes\sigma^T)\Omega)$ (we use $(\cdot)^T$ to denote matrix transpose). Using the Choi matrix representation in \eqref{eq:opt-x-E}, we have:
\begin{theorem}
\label{thm:biconvex}
The optimal value of the following biconvex optimization problem gives the inverse of the QPFI. 
\begin{align}
\label{eq:opt-x-Omega}
    F^\bP(\rho_\theta,\{M_i\})^{-1} = \inf_{(\vx,\Omega)} \;& \trace((X_2\otimes\rho_\theta^T)\Omega),\\
    \text{s.t.~}\; & \Omega\geq 0,\;\trace_{S'}(\Omega) = \id_S,  \nonumber\\
    & \trace((X\otimes\rho_\theta^T)\Omega) = 0,  \nonumber\\
    & \trace((X\otimes\partial_\theta\rho_\theta^T)\Omega) = 1.  \nonumber
\end{align}
\end{theorem}
\eqref{eq:opt-x-Omega} is a biconvex optimization problem of variables $\vx$ and $\Omega$. Fixing $\Omega$, \eqref{eq:opt-x-Omega} is a quadratic program with respect to $\vx$, and fixing $\vx$, \eqref{eq:opt-x-Omega} is a semidefinite program with respect to $\Omega$; each of which is efficiently solvable when the system dimensions are moderate and the domain of variables is compact. 

Note that the domain of $\vx$ is unbounded in \eqref{eq:opt-x-Omega}. In practice, one may impose a bounded domain on $\vx$ so that the minimum of \eqref{eq:opt-x-Omega} always exists. For cases where the QPFI is attainable, the optimal value of the bounded version will be equal to the one of \eqref{eq:opt-x-Omega} when the size of the bounded domain is sufficiently large. For singular cases where the QPFI is not attainable, the optimal value of the bounded version will approach the one of \eqref{eq:opt-x-Omega} with an arbitrarily small error as the size of the domain increases. We describe an algorithm called the global optimization algorithm~\cite{floudas2013deterministic} in \appref{app:biconvex} that can solve the bounded version of \eqref{eq:opt-x-Omega}. 
 
Finally, we note that \thmref{thm:biconvex} does not directly generalize to the case of QUPFI because the Choi matrices of unitary operators do not form a convex set. 
On the other hand, besides the set of quantum channels, our approach is also useful in optimizing the FI over other sets of quantum controls  
when the constraints on their Choi matrices can be represented using semidefinite constraints, e.g., the set of quantum channels that act only on a subsystem of the entire system. 
 
\section{Pure states}
\label{sec:pure}

In this section, we consider the special case where $\rho_\theta = \psi_\theta =  \ket{\psi_\theta}\bra{\psi_\theta}$ is pure, which is most common in sensing experiments. We first consider the optimization of the FI over the error vector $\vx$ and the unitary control $U$, and obtain two necessary conditions for the optimality of $(\vx,U)$. We use these conditions to prove equality between the QPFI and the QUPFI for pure states, showing that unitary controls are optimal for such states (when $\mH_S \cong \mH_{S'}$). We also obtain an analytical expression of the QPFI for binary measurements (i.e., measurements with only two outcomes), and a semi-analytical expression and analytical bounds for general commuting-operator measurements (i.e., measurements $\{M_i\}$ that satisfy $[M_i,M_j] = 0$ for all $i,j$). In particular, we prove that the optimal control is given by rotating the pure state and its derivative into a two-dimensional subspace spanned by two of the common eigenstates of the commuting-operator measurements. 

\subsection{Necessary conditions for optimal controls}

\propref{prop:channel-unitary} shows that the optimization for the QPFI can be reduced to an optimization for the QUPFI using the ancillary system. Thus, here we first focus on the following optimization problem over the unitary control
\begin{align}
    F^\bU(\rho_\theta,\{M_i\})^{-1} = \inf_{(\vx,U)} \;& 
    \trace(U \rho_\theta U^\dagger X_2),\label{eq:opt-x-U}\\
    \text{s.t.~}\; & 
    \trace(U \rho_\theta U^\dagger X) = 0,  \label{eq:constraint-1}\\
    & \trace(U \partial_\theta\rho_\theta U^\dagger X) = 1.  \label{eq:constraint-2}
\end{align}
We obtain necessary conditions for the optimality of $(\vx,U)$ that will be useful later. 
\begin{lemma}
\label{lemma:necessary}
If $(\vx,U)$ is an optimal point for \eqref{eq:opt-x-U}, it must satisfy
\begin{equation}
\label{eq:necessary-0}
    \frac{\trace(U \partial_\theta\rho_\theta U^\dagger X) }{\trace(U \rho_\theta U^\dagger X_2)} [X_2, {U \rho_\theta U^\dagger}]= {2  [X,  U \partial_\theta\rho_\theta U^\dagger]}.
\end{equation}
In particular, suppose $\rho_\theta = \ket{\psi_\theta}\bra{\psi_\theta}$ is pure.  Let 
\begin{gather}
\ket{\psi^\perp_\theta} := \frac{1}{\sqrt{\frakn}}(1- \ket{\psi_\theta}\bra{\psi_\theta})\ket{\partial_\theta\psi_\theta},
\\
\label{eq:def-phi}
    \ket{ \phi} := U\ket{\psi_\theta},\quad 
    \ket{ \phi^\perp} := U\ket{\psi^\perp_\theta},
\end{gather}
where the normalization factor $\frakn = \bra{\partial_\theta\psi_\theta}(1- \ket{\psi_\theta}\bra{\psi_\theta})\ket{\partial_\theta\psi_\theta}$.
Then 
\eqref{eq:necessary-0} is equivalent to the following two conditions: 
\begin{enumerate}[(1)]
    \item $X\ket{ \phi} = 1/(2\sqrt{\frakn}) \ket{ \phi^\perp}$.
    \item $\left( {\braket{ \phi|X_2| \phi}} X^2 - \braket{ \phi|X^2| \phi} {X_2}\right)\ket{ \phi} = 0$.  
\end{enumerate}
\end{lemma}

\begin{proof}
Assume $(\vx,U)$ satisfies the constraints \eqref{eq:constraint-1} and \eqref{eq:constraint-2}. Then for any unitary operator $V$ such that $\trace(U \partial_\theta\rho_\theta U^\dagger V^\dagger X V) \neq 0$, 
\begin{equation}
    \left(\frac{\vx - {\trace(U \rho_\theta U^\dagger V^\dagger X V)}{} \mathbf{1}}{\trace(U \partial_\theta\rho_\theta U^\dagger V^\dagger X V)}, V U\right)
\end{equation}
also satisfies the constraints \eqref{eq:constraint-1} and \eqref{eq:constraint-2}, where $\mathbf{1}$ is a $r$-dimensional vector of which each element is $1$. We call the transformation above a ``$V$-transformation'' on $(\vx,U)$. After a $V$-transformation, the target function becomes 
\begin{equation}
\label{eq:target-V}
 \frac{\trace(U \rho_\theta U^\dagger V^\dagger X_2 V)  - \trace(U \rho_\theta U^\dagger V^\dagger X V)^2}{\trace(U \partial_\theta\rho_\theta U^\dagger V^\dagger X V)^2},
\end{equation}
which shall be no smaller than $\trace(U\rho_\theta U^\dagger X_2)$ when $(\vx,U)$ is optimal. Let $V = e^{-idG}$ where $dG$ is an arbitrary infinitesimally small Hermitian matrix. The first order derivative of \eqref{eq:target-V} with respect to $dG$ must be zero, which then implies \eqref{eq:necessary-0}. Specifically, to simplify the notation, let $\tilde\rho := U\rho_\theta U^\dagger$ and $\dot{\tilde\rho} := U \partial\rho_\theta U^\dagger$. Then the difference between the target function after and before the $V$-transformation must be zero up to the first order of $dG$, i.e.,  
\begin{gather}
\frac{\trace({\tilde\rho} X_2) + \trace({\tilde\rho} (-i [dG,X_2])) }{\left( \trace(\dot{\tilde\rho} X) + \trace(\dot{\tilde\rho} (-i[dG,X]))\right)^2}
= \frac{\trace({\tilde\rho} X_2)}{(\trace(\dot {\tilde\rho} X))^2}, \\
\allowdisplaybreaks
\begin{split}
&\Rightarrow \;
\frac{\trace(\dot{\tilde\rho} X)^2}{\trace({\tilde\rho} X_2)} - \frac{2i\trace(\dot{\tilde\rho} X)\trace(\dot {\tilde\rho} [dG,X])}{\trace({\tilde\rho} X_2)} \\
&\quad \qquad + \frac{i\trace(\dot{\tilde\rho} X)^2\trace( {\tilde\rho} [dG,X_2])}{\trace({\tilde\rho} X_2)^2} = \frac{\trace(\dot{\tilde\rho} X)^2}{\trace({\tilde\rho} X_2)},
\end{split}\\
\allowdisplaybreaks
\Rightarrow\; - {2\trace(dG  [X,\dot {\tilde\rho}])} + \frac{\trace(\dot{\tilde\rho} X)\trace(dG  [X_2,{\tilde\rho}])}{\trace({\tilde\rho} X_2)} = 0, \\
\allowdisplaybreaks
\Rightarrow\; \frac{\trace(\dot{\tilde\rho} X) [X_2, {\tilde\rho}]}{\trace(\tilde\rho X_2)} = {2  [X,\dot {\tilde\rho}]}, 
\end{gather}
where in the first step we take the inverse of both sides and ignore higher-order terms, in the second step we mutiply both sides by $-i\trace(\rho X_2)$, and in the last step we use the fact that if an operator $A$ satisfies $\trace(dG A) = 0$ for any Hermitian $dG$, then $A = 0$.

For pure states, \eqref{eq:necessary-0} can be further simplified. Using the definitions of $\ket{\phi}$ and $\ket{\phi^\perp}$, we have $U\rho_\theta U^\dagger = \ket{\phi}\bra{\phi}$ and 
\begin{equation}
\begin{split}
U\partial_\theta\rho_\theta U^\dagger  
&= U\ket{\partial_\theta \psi_\theta}\bra{\psi_\theta}U^\dagger + h.c.\\
&= (\ket{\partial_\theta \psi_\theta}-\ket{\psi_\theta}\braket{\psi_\theta|\partial_\theta\psi_\theta}\bra{\psi_\theta}) + h.c.\\
&=\sqrt{\frakn} (\ket{\phi^\perp}\bra{\phi} + \ket{\phi}\bra{\phi^\perp}), 
\end{split}
\end{equation}
where $ h.c.$ stands for the Hermitian conjugate and we use $\partial_\theta \braket{\psi_\theta|\psi_\theta} = \braket{\partial_\theta\psi_\theta|\psi_\theta} + \braket{\psi_\theta|\partial_\theta\psi_\theta} = 0$ in the second step.
\eqref{eq:necessary-0} becomes
\begin{equation}
\label{eq:necessary-1}
    \ket{u}\bra{\phi} - \ket{{\phi}}\bra{u} + \ket{v}\bra{\phi^\perp} - \ket{{\phi}^\perp}\bra{v} = 0,
\end{equation}
where $\ket{u} = X\ket{\phi^\perp} - \frac{\Re[\braket{\phi|X|\phi^\perp}]}{\braket{\phi|X_2|\phi}}X_2\ket{\phi}$ and $\ket{v} = X\ket{\phi}$. \eqref{eq:necessary-1} is equivalent to $\ket{u},\ket{v} \in {\rm span}\{\ket{\phi},\ket{\phi^\perp}\}$, $\braket{\phi|u},\braket{\phi^\perp|v} \in \bR$ and $\braket{\phi|v} = \braket{u|\phi^\perp}$. Combining these conditions with the local unbiasedness constraints $\braket{\phi|X|\phi} = 0$ and  $2\sqrt{\frakn}\Re[\braket{\phi|X|\phi^\perp}] = 1$,  the two conditions in \lemmaref{lemma:necessary} are then proven. Specifically, we first use $\braket{\phi|v} = \braket{\phi|X|\phi} = 0$ and $\ket{v} \in {\rm span}\{\ket{\phi},\ket{\phi^\perp}\}$ to derive that $X\ket{\phi} \propto \ket{\phi^\perp}$. Then using $\braket{\phi^\perp|v} \in \bR$ and $2\sqrt{\frakn}\Re[\bra{\phi}X\ket{\phi^\perp}] = 1$, we derive Condition~(1). To derive Condition~(2), we first use $\braket{u|\phi^\perp} = \braket{\phi|v} = 0$ to derive that $\ket{u} \propto \ket{\phi}$ and then $2\sqrt{\frakn}\Re[\braket{\phi|X|\phi^\perp}] = 1$ and Condition~(2) to derive that $\braket{\phi|u} = \braket{\phi|X|\phi^\perp} - \frac{\Re[\braket{\phi|X|\phi^\perp}]}{\braket{\phi|X_2|\phi}}\braket{\phi| X_2 |\phi} = \frac{1}{2\sqrt{\frakn}} - \frac{1}{2\sqrt{\frakn}} = 0$. Then we have $\ket{u} = 0$, combining $\ket{u} \propto \ket{\phi}$ and $\braket{\phi|u} = 0$. Note that $\ket{u} = 0$ is equivalent to Condition~(2) after multiplying both sides by $\frac{\braket{\phi|X_2|\phi}}{2\sqrt{\frakn}}$. Finally, we note that from Condition~(1) and Condition~(2), the necessary condition in \eqref{eq:necessary-0} can be recovered straightforwardly, proving the equivalence between \eqref{eq:necessary-0} and Conditions (1) and (2) for pure states.
\end{proof}

As a sanity check, consider the special case where $\{M_i = \ket{i}\bra{i}\}_{i=1}^{\dim(\mH_{S'})}$ is a projection onto an orthonormal basis of $\mH_{S'}$. Then we have $X_2 = X^2$, so Condition~(2) is trivially satisfied. Furthermore, choose $(\vx,U)$ such that the error observable $X = \frac{1}{2\sqrt{\frakn}} (\ket{\phi^\perp}\bra{\phi}+\ket{\phi}\bra{\phi^\perp})$, so that Condition~(1) is satisfied. Moreover, the variance of the estimation is 
\begin{equation}
    \braket{\phi|X_2|\phi} = \braket{\phi|X^2|\phi} = \frac{1}{4\frakn} = J(\rho_\theta)^{-1}, 
\end{equation}
implying that the QFI is achievable using the above projective measurement, since $J(\rho_\theta) = 4\frakn$ for pure states~\cite{wootters1981statistical,braunstein1994statistical}. For general quantum measurements, the QUPFI might be strictly smaller than the QFI, in which case for the optimal choice of $(\vx,U)$, 
\begin{equation}
    \frac{1}{F^\bU(\rho_\theta,\{M_i\})} = \braket{\phi|X_2|\phi} > \braket{\phi|X^2|\phi} = \frac{1}{J(\rho_\theta)}.  
\end{equation}
It is interesting to note that $X_2 \geq X^2$, for general POVM measurements. This follows directly from writing
\begin{equation}
     X_2-X^2 = \sum_i (x_i - X) M_i (x_i-X),  
\end{equation}
and noting that each term in the above sum is positive semi-definite.

\subsection{Unitary controls are optimal}

Using the definitions of $\ket{\phi}$ and $\ket{\phi^\perp}$ in \eqref{eq:def-phi}, we observe that 
\eqref{eq:opt-x-U} can be rewritten as 
\begin{align}
    F^\bU(\psi_\theta,\{M_i\})^{-1} = & \inf_{(\vx,\ket{\phi},\ket{\phi^\perp})} 
    \braket{\phi|X_2|\phi},\label{eq:opt-x-phi-phi-perp}\\
    & \quad\text{s.t.~}
    \braket{\phi|\phi^\perp} = 0,\; \braket{\phi| X |\phi} = 0, \nonumber\\
    & \qquad\quad\Re[\braket{\phi|X|\phi^\perp}] = 1/(2\sqrt{\frakn}), \nonumber
\end{align}
where $\psi_\theta = \ket{\psi_\theta}\bra{\psi_\theta}$ is pure. 
Here $\ket{\phi}$ and $\ket{\phi^\perp}$ are two arbitrary normal vectors that are orthogonal. From \eqref{eq:opt-x-phi-phi-perp}, changing $\ket{\phi}$ to $\ket{\phi}/(2\sqrt{\frakn})$ makes it clear that $F^\bU(\rho_\theta,\{M_i\})$ can be written as the product of 
\begin{equation}
    J(\psi_\theta)=4\frakn
\end{equation}
and a state-independent constant. We have 
\begin{equation}
\label{eq:gamma}
    F^\bU(\psi_\theta,\{M_i\}) = \gamma(\{M_i\}) J(\psi_\theta), 
\end{equation}
where 
\begin{align}
    \gamma(\{M_i\})^{-1} = & \inf_{(\vx,\ket{\phi},\ket{\phi^\perp})} 
    \braket{\phi|X_2|\phi},\label{eq:gamma-x}\\
    & \quad\text{s.t.~}
    \braket{\phi|\phi^\perp} = 0,\; \braket{\phi| X |\phi} = 0, \nonumber\\
    & \qquad\quad\Re[\braket{\phi|X|\phi^\perp}] = 1. \nonumber
\end{align}
Or more explicitly,
\begin{equation}
\label{eq:gamma-no-x}
    \gamma(\{M_i\}) = \sup_{\ket{\phi},\ket{\phi^\perp}} \sum_i  \frac{\Re[\braket{\phi|M_i|\phi^\perp}]^2}{\braket{\phi|M_i|\phi}}. 
\end{equation}
(Note that going from \eqref{eq:gamma-x} to \eqref{eq:gamma-no-x}, we only need to optimize the target function over $\vx$ with a fixed $(\ket{\phi},\ket{\phi^\perp})$ and use standard methods for quadratic programming, e.g., Lagrange multipliers~\cite{boyd2004convex}). Note that \eqref{eq:gamma} and \eqref{eq:gamma-no-x} were also proven using a different method in \cite{len2021quantum}. 
$\gamma(\{M_i\})$ is the \emph{normalized} QUPFI for any pure states with unit QFIs and it is a function of $\{M_i\}$ that lies in $[0,1]$, which is the ratio between the QUPFI and the QFI for any pure states. It is independent of the exact $\psi_\theta$ and can fully characterize the power of quantum measurements in terms of estimation on pure states. 

Note that \eqref{eq:gamma} decomposes the QUPFI into the product of the QFI, as a function of states, and the normalized QUPFI, as a function of measurements. This result is useful when experimentalists have control over input states in sensing processes. It implies when a pure input state $\psi_0$ undergoes unitary evolution $U_\theta$, the optimal choices of the input state that maximizes the output FI are identical in situations with or without measurement noise.

Using Condition~(1) in \lemmaref{lemma:necessary}, we now prove that unitary controls are always optimal, that is, the QPFI is equal to the QUPFI when $\mH_S \cong \mH_{S'}$. We have the following theorem
\begin{theorem}
\label{thm:unitary}
Consider a pure state $\psi_\theta$ and a quantum measurement $\{M_i\}$ acting on the same system. Unitary preprocessing controls are always optimal among quantum preprocessing controls for optimizing the FI, i.e., 
\begin{equation}
\label{eq:channel-equal-unitary}
F^\bP(\psi_\theta,\{M_i\}) = F^\bU(\psi_\theta,\{M_i\}). 
\end{equation}
Or equivalently, 
\begin{equation}
\label{eq:channel-equal-unitary-gamma}
\gamma(\{M_i\}) = \gamma(\{M_i \otimes \id_A\}), 
\end{equation}
where $A$ is an ancillary system of an arbitary size. 
\end{theorem}

\begin{proof}
We first consider the situation where the QUPFI is attainable, that is, there always exists an $(\vx,U)$ such that the infimum in \eqref{eq:opt-x-U} is attainable. Using Condition~(1) in \lemmaref{lemma:necessary}, we can rewrite \eqref{eq:opt-x-phi-phi-perp} as 
\begin{align}
    F^\bU(\psi_\theta,\{M_i\})^{-1} = \min_{(\vx,\ket{\phi})} \;& 
    \braket{\phi|X_2|\phi},\label{eq:opt-x-phi}\\
     \text{s.t.~} \; &
    \braket{\phi| X |\phi} = 0, \nonumber\\
    &  \braket{\phi| X^2 |\phi} = 1/(4\frakn),  \nonumber
\end{align}
where $\frakn = J(\psi_\theta)/4$. 
Let $\dim(\mH_A) \geq \dim(\mH_{S'})^2 = d^2$. $J(\psi_\theta) = J(\psi_\theta\otimes\ket{0_A}\bra{0_A})$ and \propref{prop:channel-unitary} imply
\begin{align}
    F^\bP(\psi_\theta,\{M_i\})^{-1} 
    & = F^\bU(\psi_\theta \otimes\ket{0_A}\bra{0_A},\{M_i\otimes \id_A\})^{-1}
    \nonumber \\
    & = \min_{(\vx,\ket{\phi})} \braket{\phi|X_2 \otimes \id_A|\phi}, \label{eq:opt-x-phi-A}\\
    &\qquad \text{s.t.~} \braket{\phi| X\otimes \id_A |\phi} = 0,\nonumber \\
    &\quad\;\;\,\qquad \braket{\phi| X^2 \otimes \id_A |\phi} = 1/(4\frakn).  \nonumber
\end{align}
It is equivalent to the optimization problem 
\begin{align}
    \min_{(\vx,\sigma)} \;& \trace(\sigma X_2), \label{eq:opt-x-sigma}\\
    \text{s.t.~} \;& \trace(\sigma X)=0,\;\trace(\sigma X^2)=1/(4\frakn),\nonumber
\end{align}
where $\sigma$ is an arbitrary density operator and corresponds to $\trace_A(\ket{\phi}\bra{\phi})$. 
We will show below that for any $\sigma^*$ that is optimal for \eqref{eq:opt-x-sigma}, there exists an optimal pure state solution  $\ket{\phi^{**}}\bra{\phi^{**}}$ for \eqref{eq:opt-x-sigma}. Then the optimal values of \eqref{eq:opt-x-phi} and \eqref{eq:opt-x-sigma} must be the same, proving \eqref{eq:channel-equal-unitary}.

Assume $(\vx^*,\sigma^*)$ is optimal for \eqref{eq:opt-x-sigma}. Without loss of generality, we assume ${\rm supp}(\sigma^*) \subseteq {\rm supp}((X^*)^2)$, because otherwise $\sigma^*$ projected onto the support of $(X^*)^2$ is another optimal solution because the constraints in \eqref{eq:opt-x-sigma} are invariant and the target function is no larger after the projection. We now show there exists another optimal solution $(\vx^{**},\ket{\phi^{**}}\bra{\phi^{**}})$. First, note that $X^* = \sum_i x_i^* M_i$ and $X_2^* = \sum_i (x_i^*)^2 M_i$ satisfy $\trace(\sigma^* X^*) = 0$ and $\trace(\sigma^* (X^*)^2) = 1/(4\frakn)$ from \eqref{eq:opt-x-sigma}, and 
\begin{equation}
\label{eq:inline}
    \gamma(\{M_i\otimes \id_A\}) X_2^*\Pi = (X^*)^2\Pi, 
\end{equation} 
where $\Pi$ is the projection onto the support of $\sigma^*$.
Note that \eqref{eq:inline} is true because 
\begin{enumerate}[(i),wide, labelwidth=!, labelindent=0pt]
\item $\frac{\bra{\phi}(X^*)^2\ket{\phi}}{\bra{\phi}X_2^*\ket{\phi}} = \frac{F^\bU(\psi_\theta,\{M_i\otimes \id_A\})}{4\frakn} = \gamma(\{M_i\otimes \id_A\})$, from Condition~(1) in \lemmaref{lemma:necessary} and, 
\item ${\braket{ \phi|X^*_2| \phi}} (X^*)^2\Pi = \braket{ \phi|(X^*)^2| \phi}{X_2^*}\Pi = 0$ from Condition~(2) in \lemmaref{lemma:necessary}.
\end{enumerate}

Let $\sigma^* = \sum_{k=1}^{d} \mu_k \ket{k}\bra{k}$, where $\{\ket{i}\}_{i=1}^{d'}$ is orthonormal in $\mH_S$. We claim 
that we can always choose 
\begin{equation}
\label{eq:phi-starstar}
    \ket{\phi^{**}} = \sum_{k=1}^{d} e^{i\varphi_k} \sqrt{\mu_k} \ket{k}, 
\end{equation}
such that $\braket{\phi^{**}|X^*|\phi^{**}} = 0$,
by picking a suitable $\{\varphi_k\}_{k=1}^{d}$. 
To see this, observe that:
\begin{equation}
    \braket{\phi^{**}|X^*|\phi^{**}} = \sum_{k \neq k'} e^{i(\varphi_k-\varphi_{k'})} \sqrt{\mu_k \mu_{k'}} \bra{k}X^*\ket{k'},
\end{equation}
is a real, continuous function $f(\varphi_1,\ldots,\varphi_d)$ of $\{\varphi_k\}_{k=1}^d \in \bR^{d}$,
where we omitted the sum over $k=k'$ terms because $\trace(\sigma^* X^*) = 0$ implies that it vanishes.
Note that for any fixed $\{\varphi_k\}_{k=1}^d$, the sum of all $2^d$ terms $f(\varphi_1\pm \frac{\pi}{2},\ldots,\varphi_d\pm \frac{\pi}{2})$ is zero, implying that one, or more, of these terms is zero, or that some are negative and others are positive. In the latter case, the continuity of $f(\varphi_1,\ldots,\varphi_d)$ implies that its image must include zero. Therefore, we can pick a $\{\varphi_k\}_{k=1}^d$ such that $f(\varphi_1,\ldots,\varphi_d) = 0$, based on which the $\ket{\phi^{**}}$ defined by \eqref{eq:phi-starstar} satisfies $\braket{\phi^{**}|X^*|\phi^{**}} = 0$. Furthermore, we choose 
\begin{equation}
    \vx^{**} = \sqrt{\frac{1}{4\frakn\braket{\phi^{**}|(X^*)^2|\phi^{**}}}}\,\vx^{*},
\end{equation}
so that $\bra{\phi^{**}}(X^{**})^2\ket{\phi^{**}} = 1/\left(4\frakn \right)$. Note that $\braket{\phi^{**}|(X^*)^2|\phi^{**}}$ is always positive and thus the above denominator is positive because we assumed $(X^*)^2$ is positive definite on ${\rm supp}(\sigma^*)$. We have now proved that $(\vx^{**},\ket{\phi^{**}}\bra{\phi^{**}})$ satisfies the constraints in \eqref{eq:opt-x-sigma}. Moreover, noting that $\gamma(\{M_i\otimes \id_A\}) X_2^{**}\Pi = (X^{**})^2\Pi$, the value of the target function $\bra{\phi^{**}} X^{**}_2  \ket{\phi^{**}} 
= \frac{\bra{\phi^{**}} (X^{**})^2 \ket{\phi^{**}}}{\gamma(\{M_i\otimes \id_A\})} = \frac{1}{4\frakn \gamma(\{M_i\otimes \id_A\})} = F^\bP(\psi_\theta,\{M_i\})^{-1}$ is also optimal. Therefore, $(\vx^{**},\ket{\phi^{**}}\bra{\phi^{**}})$ is an optimal solution for both \eqref{eq:opt-x-phi} and \eqref{eq:opt-x-sigma}, proving \eqref{eq:channel-equal-unitary}.

When the QPFI of \eqref{eq:opt-x-U} is not attainable, we take $M_i^{(\epsilon)} = (1-\epsilon) M_i + \epsilon\trace(M_i)\frac{\id}{d}$ and using \thmref{thm:attainability}, we have 
\begin{multline}
    F^\bP(\psi_\theta,\{M_i\})  = \lim_{\epsilon\rightarrow 0^+} F^\bP(\psi_\theta,\{M_i^{(\epsilon)}\}) \\ = \lim_{\epsilon\rightarrow 0^+} F^\bU(\psi_\theta,\{M_i^{(\epsilon)}\}) = F^\bU(\psi_\theta,\{M_i\}),
\end{multline}
where in the second step we use the equality between the QPFI and the QUPFI in the case where the QUPFI is attainable. 
 
So far, we have proven that \eqref{eq:channel-equal-unitary-gamma} is true when $\dim(\mH_A) \geq \dim(\mH_{S'})^2$, due to \propref{prop:channel-unitary} and the equality between the QPFI and the QUPFI. It also holds for any $\mH_{A'}$ such that $\dim(\mH_{A'}) \leq \dim(\mH_{S'})^2$ because we have $\gamma(\{M_i\otimes \id_{A}\}) \geq \gamma(\{M_i\otimes \id_{A'}\}) \geq \gamma(\{M_i\})$ by definition. 

\end{proof}

\subsection{Analytical solution for binary measurements}
\label{sec:binary}

Here we provide an analytical solution to the QPFI and the corresponding optimal preprocessing control using \propref{prop:channel-unitary} for binary measurements where $r = 2$. 

\subsubsection{Measurement on a qubit}

We first consider the simplest case where the measurement is on a single qubit. Let $X = x_1 M_1+ x_2 M_2$ where $M_1 = M$ and $M_2 = \id - M$. Without loss of generality, we assume 
\begin{equation}
    M = m_1 \ket{1}\bra{1} + m_2\ket{2}\bra{2},
\end{equation} 
for some  $m_1,m_2\in[0,1]$, where $\{\ket{1},\ket{2}\}$ is an orthonormal basis. Moreover, we assume $m_1 > m_2$ and $1-m_1 \geq m_2$. (When $m_1=m_2$, we must have $\gamma(\{M_i\}) = 0$ because the measurement outcome does not depend on $\theta$.) Here $m_2$ and $1 - m_1$ can be interpreted as the error probabilities that state $\ket{2}$ is mistaken for $\ket{1}$, and state $\ket{1}$ is mistaken for $\ket{2}$, respectively. 

Consider first the case where $1 > m_1 > m_2 > 0$, that is, the error probabilities are both non-zero.  We show in \appref{app:binary-qubit} that  all solutions that satisfy the two necessary conditions in \lemmaref{lemma:necessary} give the same optimal FI. One optimal solution to the preprocessed state is
\begin{gather}
    \ket{\phi^*} = \sqrt{p^*}\ket{1} + \sqrt{1-p^*}\ket{2},\label{eq:opt-phi}\\
    \ket{\phi^{\perp*}} = \sqrt{1-p^*}\ket{1} - \sqrt{p^*}\ket{2},\label{eq:opt-phi-perp}
\end{gather}
where
\begin{equation}
    p^* = \frac{\sqrt{m_2(1-m_2)}}{\sqrt{m_1(1-m_1)}+\sqrt{m_2(1-m_2)}}. 
\end{equation}
Here the optimal unitary control $U^*$ can be chosen as any unitary such that \eqref{eq:def-phi} is true for \eqref{eq:opt-phi} and \eqref{eq:opt-phi-perp}. (In the following, we will only use $(\ket{\phi^*},\ket{\phi^{\perp *}})$ to represent the optimal preprocessing unitary with the implicit assumption that $U^*$ can be chosen as any unitary rotating $(\ket{\psi_\theta},\ket{\psi^\perp_\theta})$ to $(\ket{\phi^*},\ket{\phi^{\perp *}})$). Note that the symmetry transformations $\ket{\phi^{\perp*}} \mapsto - \ket{\phi^{\perp*}}$, $\ket{1} \mapsto e^{i\omega}\ket{1}$ and $\ket{2} \mapsto e^{i\omega'}\ket{2}$ for any $\omega,\omega' \in \bR$ will generate alternative optimal solutions, and they all provide the same optimal normalized FI:
\begin{equation}
\label{eq:binary-gamma}
    \gamma(\{M_i\}) 
    =1-\big(\sqrt{ m_1 m_2} +\sqrt{\left(1-m_{1}\right)\left(1-m_{2}\right)}\big)^{2}.
\end{equation}
Note that this result was obtained also in \cite{len2021quantum} using a different method based on the Bloch sphere representation.
Here $\sqrt{1 - \gamma(\{M_i\})}$ is exactly equal to the fidelity between two binary probability distributions $(m_1,1-m_1)$ and $(m_2,1-m_2)$.

Take the symmetric binary measurement as an example, where $m_1 = 1 - m$, $m_2 = m$ and $m < 1/2$, and $m$ represents the probability of a bit-flip error in the measurement. Then we have $p^* = 1/2$ (as expected from the bit-flip symmetry), and $\gamma(\{M_i\}) = 1 - 4m(1-m)$, which is equal to $1$ in the noiseless case, and drops to $0$ when $m\rightarrow 1/2$. 

In the case of perfect projective measurements where $1 = m_1 > m_2 = 0$, we show in \appref{app:binary-qubit} that the QPFI is equal to the QFI and is attainable for any $0 < p^* < 1$. The case where $1 > m_1 > m_2 = 0$ is singular, in the sense that the QPFI is no longer attainable but only approachable. It corresponds to the situation where one type of error ($\ket{2}$ mistaken for $\ket{1}$) is zero, while the other ($\ket{1}$ mistaken for $\ket{2}$) is non-zero. In this case, we have $\gamma(\{M_i\}) = m_1$ using \eqref{eq:binary-gamma} and \thmref{thm:attainability}. 

\subsubsection{Measurement on a qudit}

Next, we consider the general case where the measurement is on a qudit and we assume $\dim(\mH_{S'}) = d \geq 2$. Without loss of generality, we assume 
\begin{equation}
    M = \sum_{j=1}^{d} m_{j} \ket{j}\bra{j}, 
\end{equation}
where $\{\ket{j}\}_{j=1}^{d}$ is an orthonormal basis of $\mH_{S'}$. We also assume $m_{i}\geq m_{j}$ for all $i \leq j$ without loss of generality. Here we assume $1 > m_1 > m_d > 0$, which guarantees the attainability of the QPFI (see \lemmaref{lemma:noisy} in \appref{app:singular}) and the non-triviality of quantum measurements. (The singular cases where $m_1 = 1$ or $m_d = 0$ can be derived using \thmref{thm:attainability}.) We show in \appref{app:binary-qudit} that the optimal solution to $\ket{\phi}$ is supported on basis states corresponding to at most two different values of $m_i$ and the problem is simplified to selecting the optimal basis states and applying the qubit-case results. We show that 
\begin{gather}
    \ket{\phi^*} = \sqrt{p^*}\ket{1} + \sqrt{1-p^*}\ket{d},\\
    \ket{\phi^{\perp*}} = \sqrt{1-p^*}\ket{1} - \sqrt{p^*}\ket{d},
\end{gather}
is an optimal solution, where
\begin{equation}
    p^* = 
    \frac{\sqrt{m_d(1-m_d)}}{\sqrt{m_1(1-m_1)}+\sqrt{m_d(1-m_d)}}
\end{equation}
The normalized QPFI is given by 
\begin{equation}
\label{eq:gamma-qudit-binary}
    \gamma(\{M_i\}) = 
    1-\big(\sqrt{ m_1 m_d} +\sqrt{\left(1-m_{1}\right)\left(1-m_{d}\right)}\big)^{2}. 
\end{equation}
Viewing $\{(m_i,1-m_i)\}_{i=1}^d$ as $d$ binary probability distributions, the optimal strategy is always to select the two probability distributions that have the minimum fidelity (i.e., the largest distance) between each other.

\subsection{Semi-analytical solution and analytical bounds for commuting-operator measurements}

Here we consider \emph{commuting-operator measurements}, where all measurement operators commute, which is among the most common types of measurements in quantum sensing experiments, e.g., projective measurements affected by detection errors. 

Assume $\dim(\mH_{S'})= d \geq 2$. Without loss of generality, we assume 
\begin{equation}
\label{eq:commuting}
    M_i = \sum_{j=1}^{d} m_{j}^{(i)} \ket{j}\bra{j},
\end{equation}
where $\{\ket{j}\}_{j=1}^{d}$ is an orthonormal basis of $\mH_{S'}$ and $\sum_{i=1}^r m_{j}^{(i)} = 1$ for all $j$. Again, we assume $m_{j}^{(i)} > 0$ for all $i,j$ to exclude the singular cases where the QPFI is not attainable. 

In order to find the optimal control, we first prove the following theorem which states that the optimal $\ket{\phi}$ can be restricted to a two-dimensional subspace spanned by two basis states, i.e., the optimal unitary controls rotate the pure state and its derivative to a subspace spanned by two of the eigenstates of the commuting-operator measurement.

\begin{restatable}{theorem}{thmtwod}
\label{thm:2d}
For commuting-operator measurements (\eqref{eq:commuting}), there always exists an optimal solution to $(\ket{\phi},\ket{\phi^\perp})$ such that $\ket{\phi} = \sqrt{p}\ket{k} + \sqrt{1-p}\ket{l}$ and $\ket{\phi^\perp} = \sqrt{1-p}\ket{k} - \sqrt{p}\ket{l}$ for two basis states $\ket{k}$ and $\ket{l}$ and $0 < p < 1$. 
\end{restatable}

The proof is provided in \appref{app:2d}. Then we see that the normalized QPFI for commuting-operator measurements will be 
\begin{equation}
\label{eq:gamma-commuting}
    \gamma(\{M_i\}) = \max_{1 \leq k < l \leq d} \gamma_{kl}(\{M_i\}),
\end{equation}
using \thmref{thm:2d}, where 
\begin{equation}
    \gamma_{kl}(\{M_i\}) = \gamma(\{M_i\}|_{{\rm span}\{\ket{k},\ket{l}\}}),
\end{equation}
and $\{M_i\}|_{{\rm span}\{\ket{k},\ket{l}\}}$ is the quantum measurement restricted in the subspace spanned by $\ket{k}$ and $\ket{l}$. 

We show in \appref{app:optimal-commuting} that 
\begin{equation}
\label{eq:gamma-kl}
\gamma_{kl}(\{M_i\}) = \sum_i \frac{p^*_{kl} (1 - p^*_{kl}) (m_k^{(i)}-m_l^{(i)})^2 }{p^*_{kl} m_k^{(i)} + (1 - p^*_{kl}) m_l^{(i)}}, 
\end{equation}
where $p^*_{kl} \in (0,1)$ is the unique solution to 
\begin{equation}
\label{eq:p-equation}
\sum_{i=1}^r \frac{m^{(i)}_k (m^{(i)}_k  - m^{(i)}_l)^2}{\big( m^{(i)}_k + \frac{1-p_{kl}}{p_{kl}} m^{(i)}_l\big)^2}
= \sum_{i=1}^r \frac{m^{(i)}_l (m^{(i)}_k  - m^{(i)}_l)^2 }{\big(\frac{p_{kl}}{1-p_{kl}} m^{(i)}_k + m^{(i)}_l\big)^2} 
\end{equation}
and the corresponding optimal preprocessed state in ${\rm span}\{\ket{k},\ket{l}\}$ is 
\begin{gather}
\ket{\phi_{kl}^*} = \sqrt{p_{kl}^*}\ket{k} + \sqrt{1-p_{kl}^*}\ket{l},\\
\ket{\phi_{kl}^{\perp*}} = \sqrt{1-p_{kl}^*}\ket{k} - \sqrt{p_{kl}^*}\ket{l}.
\end{gather}
(The symmetry transformations $\ket{\phi^{\perp*}} \mapsto - \ket{\phi^{\perp*}}$, $\ket{k} \mapsto e^{i\omega}\ket{k}$ and $\ket{l} \mapsto e^{i\omega'}\ket{l}$ for any $\omega,\omega' \in \bR$ will generate alternative optimal solutions.) 
The optimal preprocessed state $(\ket{\phi^*},\ket{\phi^{\perp*}})$ in the entire Hilbert space
that achieves \eqref{eq:gamma-commuting} is chosen as  $(\ket{\phi_{kl}^*},\ket{\phi_{kl}^{\perp*}})$ for $(k,l)$ that maximizes $\gamma_{kl}(\{M_i\})$. 

For the special case where $r = 2$, the problem reduces to the binary measurement problem discussed in \secref{sec:binary} and $p_{kl}^*$ can be found analytically. In general, however, the analytical solution to $p_{kl}^*$ might not exist since it is a root of a high degree polynomial (\eqref{eq:p-equation}) and numerical methods are needed. Nonetheless, a simple analytical upper bound on $\gamma(\{M_i\})$ can still be obtained, as shown in the following theorem (see a detailed proof in \appref{app:uppgamma}).

\begin{restatable}{theorem}{thmuppgamma}
\label{thm:uppgamma}
For commuting-operator measurements (\eqref{eq:commuting}), the normalized QPFI $\gamma(\{M_i\})$ satisfies
\begin{equation}
\label{eq:uppgamma}
    \gamma(\{M_i\}) \leq  1 - \min_{kl} \left(\sum_i \sqrt{m^{(i)}_km^{(i)}_l}\right)^2. 
\end{equation}
When there exists a $(k,l)$ that minimizes $\sum_i \sqrt{m^{(i)}_km^{(i)}_l}$ such that the set $\big\{\frac{m_k^{(i)}}{m_l^{(i)}},\,1\leq i \leq r\big\}$ contains at most two elements, the inequality is tight. 
\end{restatable}

To derive lower bounds on $\gamma(\{M_i\})$, one could replace $p_{kl}^*$ with any $0 \leq p \leq 1$ in the expression \eqref{eq:gamma-kl}. For example, taking $p = 1/2$, we have (as also shown in \cite{len2021quantum})
\begin{align}
    \gamma(\{M_i\}) &\geq \max_{kl}\sum_i\frac{(m^{(i)}_k-m^{(i)}_l)^2}{2(m^{(i)}_k+m^{(i)}_l)} \\
    &\geq 1 - \min_{kl} \sum_i \sqrt{m^{(i)}_km^{(i)}_l},
    \label{eq:lowergamma}
\end{align}
where we use $m^{(i)}_k+m^{(i)}_l \leq (\sqrt{m^{(i)}_k}+\sqrt{m^{(i)}_l})^2$. Combining the upper and lower bounds, we observe that $\gamma(\{M_i\}) \approx 1$ when $\sum_i \sqrt{m^{(i)}_km^{(i)}_l} \approx 0$. It means that the QPFI will be close to the QFI when there exist two basis states $\ket{k}$ and $\ket{l}$ such that the fidelity between two probability distributions $\{m_k^{(i)}\}$ and $\{m_l^{(i)}\}$ is close to zero (meaning that they are almost perfectly distinguishable). 

The upper bound in \eqref{eq:uppgamma} is saturated when the measurement is binary. Another physical example is lossy photodetection. The probability of detecting $i$ photons given a Fock state of $k$ ($i \leq k$) photons is: $m_{k}^{(i)}=\binom{k}{i}(1-\eta)^{i}\eta^{k-i},$ where $1-\eta$ is the quantum efficiency of the photodetector. Assuming the maximal number of photons is $N$, it is simple to see that the optimal basis states are Fock states $\ket{0}$, $\ket{N}$. 
Since $m_{0}^{(0)}=1$, only $m_{0}^{(0)} \slash m_{N}^{(0)}$ is non-vanishing and thus $\gamma(\{M_i\})$ saturates the upper bound: $\gamma(\{M_i\})=1-\eta^{N}$. (Technically, we need to assume all $m_k^{(i)} > 0$ to avoid the singularity issue, but the above statement holds because the value of $\gamma(\{M_i\})$ can be calculated by first adding a small perturbation to the detection errors (like in \thmref{thm:attainability}) and then taking the limit as the perturbation vanishes.)   

Finally, note that although \thmref{thm:2d} and \thmref{thm:uppgamma} do not directly tell us how to choose the two optimal basis states, such a choice may sometimes be obvious. For example, consider a $n$-qubit system $({\rm span}\{\ket{1},\ket{2}\})^{\otimes n}$ measured by $\{M,\id-M\}^{\otimes n}$ (independently on each subsystem) and $M = (1-m)\ket{1}\bra{1} + m \ket{2}\bra{2}$. Then using \thmref{thm:2d}, due to the bit-flip symmetry and the fact that tracing out some parts of the quantum state will not increase its QPFI, it is clear that rotating $(\ket{\psi_\theta},\ket{\psi^\perp_\theta})$ into ${\rm span}\{\ket{1}^{\otimes n},\ket{2}^{\otimes n}\}$, or any other basis states, e.g., $\{\ket{121\cdots 1},\ket{212\cdots 2}\}$ that are distinct on each qubit, must be an optimal choice. In general, it remains open if there is a simple criterion to help us select the optimal $k$ and $l$ besides a direct calculation of \eqref{eq:gamma-kl} (or sometimes \eqref{eq:uppgamma}) for different $k$ and $l$.

\section{Classically mixed states} 
\label{sec:classical}

In this section, we consider another type of quantum states, which we called \emph{classically mixed states}, with commuting-operator measurements. A classically mixed state is a state which commutes with its derivative, e.g., Gibbs states whose temperature is to be estimated~\cite{blundell2010concepts}. In this section, we use the following form of classically mixed states: 
\begin{equation}
\label{eq:classical}
    \zeta_\theta = \sum_{i=1}^D \lambda_{i,\theta} \ket{i}\bra{i},
\end{equation}
where $D = \dim(\mH_{S})$, $\lambda_{i,\theta}$ are functions of $\theta$ (we will drop the subscript $\theta$ for conciseness), $\{\ket{i}\}$ is an orthonormal basis of $\mH_S$ that is independent of $\theta$ and we use $\zeta_\theta$ to represent classically mixed states. Note that the QFI of \eqref{eq:classical}  $J(\zeta_\theta) = \sum_{i=1}^D (\partial_\theta\lambda_i)^2/\lambda_i$ is equal to the FI $ J(\{\lambda_i\})$ of the classical distribution $\{\lambda_i\}_{i=1}^D$.
Also, note that we assume in this section, without loss of generality, that the commuting-operator measurement $\{M_i\}$ and the classically mixed state $\zeta_\theta$ share the same eigenstates $\{\ket{i}\}_{i=1}^{\max\{d,D\}}$, as it is always possible to apply a unitary rotation in the preprocessing control so that they are aligned.

We first show that optimizing the FI over quantum channels is equivalent to finding optimal stochastic matrices (which describes the transitions of a classical Markov chain) for the classical preprocessing optimization problem. Then we prove that the optimal control always corresponds to a stochastic matrix that has only elements $0$ or $1$, which we call a coarse-graining stochastic matrix. It implies that the QPFI is always attainable, and that the QPFI can in some cases be strictly larger than the QUPFI. 
Finally, we closely examine the case of a binary measurement on a single qubit.

\subsection{Optimization over stochastic matrices}

\begin{lemma}
\label{lemma:stochastic}
Consider classically mixed states \eqref{eq:classical} and commuting-operator measurements \eqref{eq:commuting}. Then 
\begin{equation}
\label{eq:classical-0}
F^\bP(\zeta_\theta,\{M_i\}) 
= \sup_{P\in\sto} J(\{{\vm}^{(i)T} P \vlambda\}),
\end{equation}
and when $d = D$, 
\begin{equation}
\label{eq:classical-qupfi}
~~~~F^\bU(\zeta_\theta,\{M_i\}) 
\leq \sup_{P\in\dsto} J(\{{\vm}^{(i)T} P \vlambda\}), 
\end{equation}
where $\sto$ represents the set of $d\times D$ stochastic matrices of which every column vector sums up to one and $\dsto$ represents the set of $D\times D$ doubly stochastic matrices of which every column and row vector sums up to one, $\vm^{(i)}$ is a column vector whose entries are $m^{(i)}_j$, $\vlambda$ is a column vector whose entries are $\lambda_i$. 
\end{lemma}

\begin{proof}
Let $\mE(\cdot) = \sum_j K_j (\cdot) K_j^\dagger$ be an arbitrary quantum channel, then we have 
\begin{equation}
\label{eq:classical-1}
    \trace(M_i\mE(\zeta_\theta)) = \vm^{(i)T} P \vlambda,
\end{equation}
where the matrix $P$ satisfies $P_{\ell k} = \sum_j \abs{\bra{\ell}K_j\ket{k}}^2 = \sum_j |(K_j)_{\ell k}|^2$ 
, which implies 
\begin{equation}
\label{eq:quantum-classical}
     J(\{\trace(M_i\mE(\zeta_\theta))\}) = J(\{{\vm}^{(i)T} P \vlambda\}).
\end{equation}
We must have $\sum_\ell P_{\ell k} = \sum_{\ell j} \bra{k}K_j^\dagger\ket{\ell}\bra{\ell}K_j\ket{k} = \braket{k|k} = 1$, because $\sum_j K_j^\dagger K_j = \id$. Thus, $P$ is a stochastic matrix. For any quantum channel, there exists a stochastic matrix such that \eqref{eq:classical-1} holds true, proving the left-hand side is no larger than the right-hand side in \eqref{eq:classical-0}. Moreover, when $\mE(\cdot) = U(\cdot)U^\dagger$ is a unitary channel, $P_{\ell k} = |U_{\ell k}|^2$ must be doubly stochastic, implying \eqref{eq:classical-qupfi}. 

On the other hand, for any stochastic matrix $P$, we define $K_{(\ell,k)} = \sqrt{P_{\ell k}}\ket{\ell}\bra{k}$ for $1 \leq \ell \leq d$ and $1 \leq k \leq D$. Then we have $\sum_{\ell k}K_{(\ell,k)}^\dagger K_{(\ell,k)} = \sum_{\ell k} P_{\ell k}\ket{k}\bra{k} = \id$. And $\mE(\cdot) = \sum_{(\ell,k)} K_{(\ell,k)} (\cdot) K_{(\ell,k)}^\dagger$ is then a quantum channel.  For any stochastic matrix, there exists a quantum channel such that \eqref{eq:classical-1} holds true, proving the left-hand side is no smaller than the right-hand side in \eqref{eq:classical-0}.

\end{proof}

We show in \lemmaref{lemma:stochastic} that the problem of optimizing preprocessing quantum controls on classically mixed states with commuting-operator measurements is equivalent to a classical version of preprocessing optimization where 
\begin{equation}
\label{eq:classical-FI-A}
    F^\bP(\vlambda,\{\vm^{(i)}\}) := 
    \sup_{P\in\sto} J(\{{\vm}^{(i)T} P \vlambda\})
\end{equation}
represents the classical FI with respect to a classical distribution $\vlambda$ and a noisy measurement $\vm^{(i)}$ satisfying $\sum_i \vm^{(i)} = \v{1}$ ($\v{1}$ is a vector with all elements equal to $1$), optimized over any stochastic mapping described by stochastic matrices.  In particular, for perfect measurements where $(\vm^{(i)})_j = \delta_{ij}$, $F^\bP(\vlambda,\{\vm^{(i)}\}) = J(\vlambda) = \sum_{i=1}^D (\partial_\theta \lambda_i)^2/\lambda_i$ is the classical FI. Note that \thmref{thm:classical} presented later implies that the supremum of the FI over stochastic matrices is always attainable using some $P \in \sto$ and it means we are allowed to replace $\sup_{P\in\sto}$ by $\max_{P\in\sto}$ in the definition (\eqref{eq:classical-FI-A}).

\subsection{Coarse-graining controls are optimal}
\label{sec:coarse-graining}

We first consider the classical case and prove \eqref{eq:classical-FI-A} can always be attained using some $d \times D$ stochastic matrix $P$ where every element of $P$ is either $0$ or $1$. We call this type of stochastic matrix a \emph{coarse-graining stochastic matrix} in the sense that $P$ sums up one or multiple entries of $\vlambda$ to one entry in $P \vlambda$, which is a coarse graining of measurement outcomes. 

\begin{lemma}
\label{lemma:classical-coarse-graining}
Given a classical probability distribution $\vlambda \in \bR^{D}$ and a  measurement $\{\vm^{(i)}\} \subseteq \bR^{d}$ (satisfying $\sum_i \vm^{(i)} = \v{1}$). When $F^\bP(\vlambda,\{\vm^{(i)}\})$ is attainable, there exists a $d\times D$ coarse-graining stochastic matrix $P$ such that, 
\begin{equation}
    F^\bP(\vlambda,\{\vm^{(i)}\}) = J(\{{\vm}^{(i)T} P \vlambda\})
\end{equation}
\end{lemma}

\begin{proof}
Suppose $F^\bP(\vlambda,\{\vm^{(i)}\})$ is attainable and $P^*$ is an optimal solution. We will show that there exists an optimal solution $P$ whose every column vector contains one (and only) element equal to $1$. 
If $P^*$ does not satisfy this condition, without loss of generality, assume $P^*_{11}=t_1^*$ and $P^*_{21} = a_1^* - t_1^*$ where $0 < t_1^* < a_1^* \leq 1$. Let $P(t_1)$ be a stochastic matrix function of $t_1 \in [0,a_1^*]$ where $P(t_1)_{11} = t_1$, $P(t_1)_{21} = a_1^* - t_1$ and $P(t_1)_{\ell k} = (P^*)_{\ell k}$ for $(\ell,k) \neq (1,1),(2,1)$. We have the FI equal to 
\begin{align*}
    f(t_1) &= \sum_i 
    \frac{(\partial_\theta ({\vm}^{(i)T}P(t_1) {\vlambda}))^2}{{\vm}^{(i)T} P(t_1) \vlambda}  \\
    &=  \sum_i 
    \frac{\big((m_1^{(i)} - m_2^{(i)})\partial_\theta{\lambda}_1 t_1 + b^{*(i)}\big)^2}{(m_1^{(i)} - m_2^{(i)})\lambda_1 t_1 + a^{*(i)}},
\end{align*}
where $a^{*(i)}$ and $ b^{*(i)}$ are constants, independent of $t_1$. 
The second order derivative of $f(t_1)$ is 
\begin{equation}
\frac{\partial^2 f(t_1)}{\partial t_1^2} = 
\sum_{i} \frac{2 \big(m_1^{(i)} - m_2^{(i)}\big)^2 \big(a^{*(i)} \partial_\theta{\lambda}_1 - b^{*(i)} \lambda_1 \big)^2}{\big((m_1^{(i)} - m_2^{(i)})\lambda_1 t_1 + a^{*(i)}\big)^3}, 
\end{equation}
which is always non-negative. Therefore, $f(t_1)$ is a convex function and always attains its maximum at the boundary $t_1 = 0$ or $t_1 = a_1^*$. Repeat this argument many times, one can show that there exists an optimal solution $P$ such that there is only one positive entry in every column.
\end{proof}

Note that it is not necessarily true that the optimal coarse-graining stochastic matrix that maximizes $J(\{{\vm}^{(i)T} P \vlambda\})$ is a full-rank matrix. Consider the following example. Let $d=D=3$, $r=2$, $\vlambda = (\cos^2\theta,\frac{1}{2}\sin^2\theta,\frac{1}{2}\sin^2\theta)$, ${\vm}^{(1)} = (1,\frac{1}{2},0)$, and ${\vm}^{(2)} = (0,\frac{1}{2},1)$. Then it is clear that the following stochastic matrix is optimal, 
\begin{equation}
\label{eq:opt-coarse-graining}
    P^* = \begin{pmatrix}
    1 & 0 & 0 \\
    0 & 0 & 0 \\
    0 & 1 & 1
    \end{pmatrix},
\end{equation}
because $J(\{{\vm}^{(i)T} P^* \vlambda\}) = J(\vlambda) = 4$. However, it can be verified by enumeration that $J(\{{\vm}^{(i)T} P \vlambda\}) \leq 3$, whenever $P$ is a permutation matrix, showing the non-optimality of the full-rank stochastic matrices. 

Using \lemmaref{lemma:stochastic}, we can show a similar result to \lemmaref{lemma:classical-coarse-graining} in the quantum case, that is, \emph{coarse-graining channels} are optimal quantum controls. 
\begin{theorem}
\label{thm:classical}
Consider classically mixed states \eqref{eq:classical} and commuting-operator measurements \eqref{eq:commuting}. The QPFI is always attainable using the following type of quantum channels, which we call coarse-graining channels, 
\begin{equation}
\label{eq:coarse-graining-quantum}
    \mE(\cdot) = \sum_{\ell k} P_{\ell k} \ket{\ell}\bra{k}(\cdot)\ket{k}\bra{\ell},
\end{equation}
where $P_{\ell k}$ is a $d\times D$ stochastic matrix satisfying $\sum_{\ell} P_{\ell k} = 1$ and $P_{\ell k} = 0$ or $1$. 
\end{theorem}

\begin{proof}
By definition, there exists a sequence of channels $(\mE_1,\cdots,\mE_n,\cdots)$ such that $\lim_{n\rightarrow \infty}F(\mE_n(\zeta_\theta),\{M_i\}) = F^\bP(\zeta_\theta,\{M_i\})$. According to \eqref{eq:quantum-classical} in the proof of \lemmaref{lemma:stochastic} and the arguments in  \lemmaref{lemma:classical-coarse-graining}, for every $\mE_n$ there exists a channel $\tilde\mE_n$ of the form \eqref{eq:coarse-graining-quantum} such that $F(\mE_n(\zeta_\theta),\{M_i\}) \leq  F(\tilde{\mE}_n(\zeta_\theta),\{M_i\})$. Therefore, $\lim_{n\rightarrow \infty} F(\tilde{\mE}_n(\zeta_\theta),\{M_i\}) = F^\bP(\zeta_\theta,\{M_i\})$. 
Since there are finite number of channels of the form \eqref{eq:coarse-graining-quantum}, there must exist a $\mE^* = \tilde{\mE}_n$ for some $n$ such that 
\begin{equation}
    F(\mE^*(\zeta_\theta),\{M_i\}) = F^\bP(\zeta_\theta,\{M_i\}), 
\end{equation}
proving the attainability of the QPFI. 
\end{proof}

\thmref{thm:classical} also implies that there is a gap between the QUPFI and the QPFI for general quantum states,  unlike for pure states where the QUPFI is equal to the QPFI. 
\begin{theorem}
\label{thm:gap}
There exists a classically mixed state $\zeta_\theta$ and a commuting-operator measurement $\{M_i\}$ such that
\begin{equation}
\label{eq:strict-ineq}
    F^\bU(\zeta_\theta,\{M_i\}) < F^\bP(\zeta_\theta,\{M_i\}). 
\end{equation}
\end{theorem}
\begin{proof}
Consider the example discussed below \lemmaref{lemma:classical-coarse-graining} and here we fix $\theta = \pi/4$. \thmref{thm:classical} implies that for $\zeta_\theta = \cos^2\theta \ket{1}\bra{1} + \frac{1}{2}\sin^2\theta \ket{2}\bra{2} + \frac{1}{2}\sin^2\theta \ket{3}\bra{3}$,  
\begin{equation}
    J(\zeta_\theta) = F^\bP(\zeta_\theta,\{M_i\}) = 4,
\end{equation}
where $M_1 = \ket{1}\bra{1} + \frac{1}{2} \ket{2}\bra{2}$ and $M_2 = \frac{1}{2} \ket{2}\bra{2} + \ket{3}\bra{3}$. 
In general, given any stochastic matrix $P$, the probabilities for measurement outcomes $1$ and $2$ must have the form 
\begin{gather}
    p_1 = {\vm}^{(1)T} P \vlambda = a\cos^2\theta + b\sin^2\theta,\\
    p_2 = {\vm}^{(2)T} P \vlambda = (1-a)\cos^2\theta + (1-b)\sin^2\theta, 
\end{gather}
for some $0 \leq a,b \leq 1$. Moreover, $J(\{p_1,p_2\}) = 4(a-b)^2/(2-(a+b))/(a+b)$. And $J(\{p_1,p_2\}) = J(\zeta_\theta)$ if and only if $(a,b) = (1,0)$ or $(0,1)$. Noting that the situation where $(a,b) = (1,0)$ or $(0,1)$ is not possible if $P$ is doubly stochastic. Applying \lemmaref{lemma:stochastic}, \eqref{eq:strict-ineq} is then proven. 
\end{proof}

The intuition behind this type of gap between the QPFI and the QUPFI stems from the fact that non-unitary operations, e.g., the coarse-graining channel, have the power of reducing the rank of quantum states, while unitary operations do not. Consequently, when certain conditions are met: (i)~the noisy measurement under consideration is noiseless in a lower-dimensional subspace, e.g., ${\rm span}\{\ket{1},\ket{3}\}$ in the example above and (ii)~the rank of the quantum state can be compressed without reducing its QFI, e.g., collapsing ${\rm span}\{\ket{2},\ket{3}\}$ into ${\rm span}\{\ket{3}\}$, non-unitary preprocessing operations can achieve the optimal QFI. In contrast, relying solely on unitary preprocessing for high-rank states results in unavoidable measurement noise and suboptimal performance.

Finally, we note that although the implementation of general quantum preprocessing channels can sometimes be challenging with the requirement of preparing a clean ancillary system that occurs in the Stinespring’s
dilation (see \propref{prop:channel-unitary}), the resources needed to perform coarse-graining channels can be reduced in many cases. Firstly, the ancilla size required to perform coarse-graining channels is in principle smaller than $d^2$ that is required in general cases. In fact, any coarse-graining channel defined by \eqref{eq:coarse-graining-quantum} can be simulated using a $d$-dimensional ancilla, e.g., by first performing a unitary operation on $\mH_S \otimes \mH_{A_1}$ that maps $\ket{k}_S\ket{0}_{A_1} \mapsto \ket{k}_S\ket{\iota(k)}_{A_1}$ for all $k$, where $\iota(k)$ corresponds to the index of the row such that $P_{\iota(k)k} = 1$, and then discarding the probe system $\mH_S$. Secondly, the coarse-graining channel can also be performed on certain quantum states by resetting some parts of the system with no additional ancillas in some cases. For example, consider a two-qubit quantum state $\zeta_\theta = \cos^2\theta \ket{00}\bra{00} + \frac{1}{2}\sin^2\theta \ket{10}\bra{10} + \frac{1}{2}\sin^2\theta \ket{11}\bra{11}$ and measurement operators $M_0 = \ket{00}\bra{00} + \frac{1}{2}\ket{01}\bra{01} + \frac{1}{2}\ket{10}\bra{10}$ and $M_1 = \ket{11}\bra{11} + \frac{1}{2}\ket{01}\bra{01} + \frac{1}{2}\ket{10}\bra{10}$. The coarse-graining channel mapping $\ket{00}\mapsto\ket{00}$, $\ket{10} \mapsto\ket{11}$ and $\ket{11}\mapsto\ket{11}$ is optimal and it can be performed by first resetting the second qubit to $\ket{0}$ and then applying a CNOT gate that maps $\ket{00} \mapsto \ket{00}$ and  $\ket{10} \mapsto \ket{11}$. Note that resetting qubits is usually considered much less noisy than measuring ones, e.g., in nitrogen-vacancy centers~\cite{doherty2013nitrogen,unden2016quantum}.

\subsection{Binary measurement on a single qubit}
\label{sec:binary-classical}

With \thmref{thm:classical}, in principle, one can find the QPFI for classically mixed states and commuting-operator measurements by exhausting all channels of the form \eqref{eq:coarse-graining-quantum} which is contained in a finite set. However, since the number of coarse graining stochastic matrices is large, the exhaustion procedure will be too costly. Here we closely examine a special case where a classically mixed state is measured by a binary measurement on a single qubit. The time to exhausting all coarse graining matrices is exponentially large with respect to the state dimension $D$. We will show that the time to find a solution can be reduced to a linear complexity by narrowing down the possible forms of the optimal controls. 

To be specific, consider the binary measurement $M_1 = M = m_1 \ket{1}\bra{1} + m_2 \ket{2}\bra{2}$, $M_2 = \id - M$ (assuming $m_2 \leq \min\{m_1,1-m_1\}$), and $\zeta_\theta = \sum_{i=1}^D \lambda_i\ket{i}\bra{i}$.
Then using \lemmaref{lemma:stochastic}, we first have $ F^\bP(\zeta_\theta,\{M_i\})  = \max_{\vt} f(\vt)$ where
\begin{equation}
    f(\vt) := J(\{p_{\theta}(\vt),1-p_{\theta}(\vt)\}),
\end{equation}
and 
\begin{equation}
    p_{\theta}(\vt) := m_2 + (m_1 - m_2)\vt^T \vlambda,
\end{equation}
and $\vt$ is a column vector in $[0,1]^{D}$, corresponding to the first row of the stochastic matrix $P$ in the proof of \lemmaref{lemma:stochastic}. (Note that although from \thmref{thm:classical}, it is possible to restrict $\vt$ to $\{0,1\}^{D}$, and we keep the generality of $\vt$ by allowing it to be in $[0,1]^{D}$ for later use.)

Without loss of generality, we can arrange the order of the positive elements in $\vlambda$ such that 
\begin{equation}
    {(\partial_\theta\lambda_i)}/{\lambda_i} \geq {(\partial_\theta\lambda_j)}/{\lambda_j},\quad \forall i < j \text{~and~} \lambda_{i,j} > 0. 
\end{equation}
Then we assert that 
\begin{equation}
\label{eq:classical-binary}
F^\bP(\zeta_\theta,\{M_i\}) = \max_{i\in[1,D-1]} \max\{ f(\v{1}_{\leq i}) ,f(\v{1}_{\geq i})\} ,
\end{equation}
where $\v{1}_{\leq i}$ represents the vector whose the first $i$ elements are equal to $1$ and the rest are zero. $\v{1}_{\geq i+1} = \v{1} - \v{1}_{\leq i}$. 

Now we prove \eqref{eq:classical-binary}. 
Choose an optimal $\vt^* \in [0,1]^{D}$ that maximizes $f(\vt)$.
We prove \eqref{eq:classical-binary} in each of the following three cases:
\begin{enumerate}[(i),wide, labelwidth=!, labelindent=0pt]
\item $\vt^{*T}\partial_\theta\vlambda = 0$. Then the QPFI is zero and \eqref{eq:classical-binary} is trivial.
\item $\vt^{*T}\partial_\theta\vlambda >  0$. If there exists $i < j$ such that $t^*_i < 1$, $t^*_j > 0$ and $\lambda_{i,j} > 0$. Then define $t^{**}_i = t^*_i + \epsilon/\lambda_i$ and $t^{**}_j = t^*_j - \epsilon/\lambda_j$ where $\epsilon =\min\{ (1 - t^*_i)\lambda_i, t^*_j\lambda_j\}$ and $t^{**}_{k\neq i,j} = t^{*}_{k\neq i,j}$. Then we have either $t^{**}_i = 1$ or $t^{**}_j = 0$. Moreover, we have $f(\vt^{**}) \geq f(\vt^*)$ because $\vt^{*T}\vlambda = \vt^{**T}\vlambda$ and $\vt^{*T}\partial_\theta\vlambda \leq \vt^{**T}\partial_\theta\vlambda$. Then $\vt^{**}$ is also optimal. Repeating this procedure, we can always find an optimal $\vt$ of the form $(1\cdots 1\; t \;0\cdots 0) $ for some $t \in [0,1]$. $\,\Omega(1) \leq J(\rho_\theta^{(n)}) \leq e^{o(n)}$. Using the same convexity argument as in the proof of \lemmaref{lemma:classical-coarse-graining}, we can further show $t$ can be taken to be $1$ or $0$. \eqref{eq:classical-binary} is proven. 
\item $\vt^{*T}\partial_\theta\vlambda < 0$. If there exists $i < j$ such that $t^*_i > 0$, $t^*_j < 1$ and $\lambda_{i,j} > 0$. Then define $t^{**}_i = t^*_i - \epsilon/\lambda_i$ and $t^{**}_j = t^*_j + \epsilon/\lambda_j$ where $\epsilon =\min\{ (1 - t^*_j)\lambda_j, t^*_i\lambda_i\}$ and $t^{**}_{k\neq i,j} = t^{*}_{k\neq i,j}$. Moreover, we have $f(\vt^{**}) \geq f(\vt^*)$ because $\vt^{*T}\vlambda = \vt^{**T}\vlambda$ and $\vt^{*T}\partial_\theta\vlambda \geq \vt^{**T}\partial_\theta\vlambda$. Then $\vt^{**}$ is also optimal.Repeating this procedure, we can always find an optimal $\vt$ of the form $(0\cdots 0\; t \; 1\cdots 1)$ for some $t \in [0,1]$. Using the same convexity argument as in the proof of \lemmaref{lemma:classical-coarse-graining}, we can further show $t$ can be taken to be $1$ or $0$. \eqref{eq:classical-binary} is proven. 
\end{enumerate}

\section{General quantum states}
\label{sec:general}

In \secref{sec:pure} and \secref{sec:classical}, we have obtained fruitful results on preprocessing optimization for pure states and classically mixed states. Here, we consider the QPFI for general mixed states and derive useful upper and lower bounds on them.

\subsection{Upper bound}

\begin{theorem}
\label{thm:upper}
Given any density operator $\rho_\theta$ and quantum measurement $\{M_i\}$, we have
\begin{equation}
    F^\bP (\rho_\theta,\{M_i\}) \leq \gamma(\{M_i\}) J(\rho_\theta). 
\end{equation}
\end{theorem}

\begin{proof}
Suppose $\mH_{A_1}$ and $\mH_{A_2}$ are ancillary systems such that $\mH_{A_1} \otimes \mH_{S} \congorequal \mH_{A_2} \otimes \mH_{S'}$ and $\dim(\mH_{A_1}) \geq \dim(\mH_{S'})^2$, where $\mH_S$ and $\mH_{S'}$ are the systems $\rho_\theta$ and $\{M_i\}$ act on. We also define an additional environmental system $\mH_E$ satisfying $\dim(\mH_E) = \dim(\mH_S)$. Let $\psi_\theta = \ket{\psi_\theta}\bra{\psi_\theta}_{ES}$ denote the purifications of $\rho_\theta$ in $\mH_{E}\otimes\mH_{S}$.  Using the purification-based definition of QFI~\cite{fujiwara2008fibre,escher2011general}, we have 
\begin{equation}
J(\rho_\theta) 
= \min_{\psi_\theta:\rho_\theta = \trace_E(\psi_\theta)} J(\psi_\theta). \end{equation}
Choose $\psi_\theta^*$ to be the optimal purification $\psi_\theta^*$ that minimizes $J(\psi_\theta)$ such that $J(\rho_\theta) = J(\psi_\theta^*)$.
Then 
\begin{align}
& \quad\, F^{\bP}(\rho_\theta,\{M_i\}) \nonumber\\
&= F^\bU\big(\rho_\theta\otimes\ket{0}\bra{0}_{A_1},\{M_i\otimes\id_{A_2}\}\big) \\
&\leq F^{\bU}(\psi_\theta^*\otimes\ket{0}\bra{0}_{A_1},\{\id_E \otimes M_i\otimes \id_{A_2}\}) \\
&= \gamma(\{\id_E \otimes M_i\otimes \id_{A_2}\}) J(\psi_\theta^*\otimes\ket{0}\bra{0}_{A_1})\\
&= \gamma(\{M_i\}) J(\rho_\theta),
\end{align}
where we use \propref{prop:channel-unitary}, \eqref{eq:gamma} and \thmref{thm:unitary}.

\end{proof}

\thmref{thm:upper} provides an upper bound on the QPFI for general quantum states. In particular, it shows the ratio between the QPFI and the QFI is always upper bounded by a state-independent constant $\gamma(\{M_i\})$ which is attainable when the state is pure and gives rise to the following CRB for general quantum states under noisy measurement $\{M_i\}$: 
\begin{equation}
    \Delta\hat\theta \geq \frac{1}{\sqrt{N_\expr \gamma(\{M_i\}) J(\rho_{\theta})}}. 
\end{equation}

\subsection{Lower bound}

\begin{lemma}
\label{lemma:lower}
Consider a density operator $\rho_\theta$ and quantum measurement $\{M_i\}$. Assume $\{T_i\}$ is a QFI-attainable measurement, i.e., $F(\rho_\theta,\{T_i\}) = J(\rho_\theta)$. 
Let the quantum-classical channel $\mT(\cdot) = \sum_i \trace((\cdot)T_i) \ket{i}\bra{i}_C$ where $\{\ket{i}_C\}$ is an orthonormal basis of an auxiliary system $\mH_C$. Then 
\begin{equation}
    F^\bP(\rho_\theta,\{M_i\}) \geq F^\bP(\mT(\rho_\theta),\{M_i\}).
\end{equation}
\end{lemma}

The proof \lemmaref{lemma:lower} is straightforward---it immediately follows from the definition of the QPFI. The equality holds true when $\{M_i\}$ is a projection onto an orthonormal basis of $\mH_{S'}$, i.e., $\{M_i=\ket{i}\bra{i}\}_{i=1}^{\dim(\mH_{S'})}$. 

Note that the equality in \lemmaref{lemma:lower} also holds when $\rho_\theta$ is a classically mixed state and the QFI-attainable measurement is chosen to be the projective measurement onto the basis of $\mH_S$ so that $\mT(\rho_\theta) = \rho_\theta$. For general mixed states, since $\mT(\rho_\theta)$ is a classically mixed state, the results in \secref{sec:classical} can be applied here to analyze $F^\bP(\mT(\rho_\theta),\{M_i\})$ and derive lower bounds for general mixed states. For example, one can divide the measurement operators into two subsets, restrict the measurement in a two-dimensional subspace, and then use our previous result of the binary measurement on a qubit for classically mixed states to derive an efficiently computable lower bound on the QPFI.

Note that unlike the upper bound (\thmref{thm:upper}), there are no constant lower bounds independent of $\rho_\theta$ on the ratio between $F^\bP(\rho_\theta,\{M_i\})$ and $J(\rho_\theta)$. For example, consider the single qubit case where $\rho_\theta = \cos^2\theta \ket{1}\bra{1} + \sin^2\theta \ket{2}\bra{2}$, $M_1 = (1-m)\ket{1}\bra{1} + m\ket{2}\bra{2}$, and $M_2 = \id - M_1$ ($0 < m < 1/2$). We have, from \secref{sec:binary-classical}, that 
\begin{equation}
\label{eq:qubit-binary-classical}
    F^\bP(\rho_\theta,\{M_i\}) = \frac{4(1-2m)^2\sin(2\theta)^2}{1-(1-2m)\cos(2\theta)^2}, 
\end{equation}
which tends to zero as $\theta \rightarrow 0$ (and the optimal preprocessing is identity when $\theta \in (0,\pi/4)$). On the other hand, $J(\rho_\theta) = 4$ is a constant, showing that $F^\bP(\rho_\theta,\{M_i\})/J(\rho_\theta)$ has no state-independent constant lower bounds. 

\section{Global preprocessing: asymptotic limits}
\label{sec:asymptotic}

In this section, we consider the power of global quantum preprocessing in the asymptotic limit (see \figref{fig:asymptotic}). We consider a multi-partite system $\mH_S = \mH^{\otimes n}$ and $\mH_{S'} = {\mH'}^{\otimes n}$ where $\dim \mH = D$ and $\dim \mH' = d$, a set of quantum states $\rho_\theta^{(n)}$ in $\mH^{\otimes n}$, and quantum measurements $\{M_i\}^{\otimes n}$ that can be written as tensor products of identical measurements on each subsystem $\mH'$. Arbitrary (and usually global) preprocessing quantum channels $\mE$ are applied before the noisy measurement. We will show that for a generic class of quantum states, the QPFI can reach the QFI asymptotically for large $n$. Note that the QPFI is in general not achievable~\cite{len2021quantum} 
when $\mE$ can only act locally and independently on each subsystem. 

\subsection{Attaining the QFI with noisy measurements}

\begin{theorem}
\label{thm:metrology-capacity}
Given a set of quantum states $\{\rho_\theta^{(n)}\}_{n}$ where $\rho_\theta^{(n)}$ is a function of $\theta$ and acts on $\mH^{\otimes n}$ for each $n$, we have
\begin{equation}
\lim_{n\rightarrow \infty} \frac{F^{\bP}(\rho_\theta^{(n)},\{M_{i}\}^{\otimes n})}{J(\rho_\theta^{(n)})} = 1,
\end{equation} 
if for each $\rho_\theta^{(n)}$ the following are true: 
\begin{itemize}[wide, labelwidth=!, labelindent=0pt]
    \item There exists a quantum measurement $\{T_i^{(n)}\}$ whose number of measurement outcomes is $r_n$ such that 
\begin{equation}
\label{eq:QFI-attainable}
\lim_{n\rightarrow \infty} \frac{F(\rho_\theta^{(n)},\{T_i^{(n)}\})}{J(\rho_\theta^{(n)})} = 1, \,\text{and}\,
\lim_{n\rightarrow\infty} \frac{\log r_n}{n} < C(\mM),
\end{equation}
where $\log$ is the binary logarithm and $C(\mM)$ is the classical capacity of the quantum-classical channel $\mM(\cdot) = \sum_{i} \trace\big((\cdot)M_i\big) \ket{i}\bra{i}_C$ ($\{\ket{i}_C\}$ is an orthonormal basis of an auxiliary system $\mH_C$). 
\item The regularity conditions are satisfied:
\begin{enumerate}[(1)]
\item When $\partial_\theta \lambda_i \neq 0$, $\lambda_i = 1/e^{o(n)}$, where $\lambda_i := \trace(\rho_\theta^{(n)} T_i^{(n)})$ and $\{T_i^{(n)}\}$ is defined above. 
\item $\,\Omega(1) \leq J(\rho_\theta^{(n)}) \leq e^{o(n)}$.
\end{enumerate}

\end{itemize}

\end{theorem}

\thmref{thm:metrology-capacity} provides a sufficient condition to attain the QFI using noisy measurements in the asymptotic limit $n\rightarrow \infty$. We will first provide a proof of \thmref{thm:metrology-capacity}, and return to the physical understandings of the sufficient condition later. Readers who are not interested in the technical details can skip the technical proof and advance to the discussion part.

In the proof, we will make use of a quantum-classical channel $\mT_n$ defined using $\{T_i^{(n)}\}$, and an encoding channel $\Xi_E$, such that $F(\Xi_E \circ \mT_n (\rho_\theta^{(n)}), \{M_i\}^{\otimes n})$ approaches $J(\rho_\theta^{(n)})$ asymptotically (see \figref{fig:asymptotic}). 
Intuitively speaking, the first step $\mT_n$ is to simulate the (asymptotically) QFI-attainable measurement $\{T_i^{(n)}\}$ on $\rho_\theta^{(n)}$ to transform it into a classically mixed state $\mT_n(\rho)$ such that $J(\mT_n(\rho_\theta^{(n)})) = J(\rho_\theta^{(n)})$. The second step is to choose a suitable encoding channel $\Xi_E$ such that the classical information in $\mT_n(\rho_\theta^{(n)})$ is fully preserved under $\mM^{\otimes n}$, i.e., $J(\mM^{\otimes n}\circ\Xi_E \circ \mT_n (\rho_\theta^{(n)})) \approx J(\mT_n (\rho_\theta^{(n)}))$, leading to the asymptotic attainability of the QFI. Here $\Xi_E$, along with a corresponding deconding channel $\Xi_D$, is chosen such that $\Xi_E \circ \mM^{\otimes n} \circ \Xi_D$ is asymptotically equal to a completely dephasing channel with a transmission rate $\approx C(\mM)$, which is guaranteed to exist using the HSW theorem~\cite{holevo1998capacity,schumacher1997sending}.

\begin{proof}[Proof of {\thmref{thm:metrology-capacity}}]
We first choose an $\alpha$ such that $\lim_{n\rightarrow\infty} {\log r_n}/{n} < \alpha < C(\mM)$. 
According to the definition of the classical capacity of quantum channels~\cite{watrous2018theory}, for any $\epsilon > 0$, there exists an $n_0$ such that for any $n > n_0$, there exist an encoding channel $\Xi_E$ and a decoding channel $\Xi_D$ such that 
\begin{equation}
\| \Xi_D \circ \mM^{\otimes n} \circ \Xi_E  - \mD_2^{\otimes \lfloor \alpha n \rfloor} \|_\diamond \leq \epsilon, 
\end{equation}
where $\mD_2$ is a completely dephasing qubit channel acting on qubit Hilbert space $\mH_\tb$, i.e., $\mD_2(\cdot) = \ket{0}\bra{0}(\cdot)\ket{0}\bra{0} +  \ket{1}\bra{1}(\cdot)\ket{1}\bra{1}$ and $\norm{\cdot}_\diamond$ is the diamond norm of a quantum channel~\cite{watrous2018theory} defined by $\norm{\Phi}_\diamond = \max\{\norm{(\Phi\otimes \id)(X)}_1,\norm{X}_1\leq 1\}$ ($\Phi$ and $\id$ act on systems of the same dimension, and $\norm{\cdot}_1$ denotes the trace norm). Moreover, $\epsilon = e^{-\Omega(n)}$ (see a proof in \appref{app:capacity}).
For any operator $\sigma$, we have 
\begin{equation}
\label{eq:capacity}
\big\|\Xi_D \circ \mM^{\otimes n} \circ \Xi_E (\sigma) - \mD_2^{\otimes \lfloor \alpha  n \rfloor} (\sigma)\big\|_1 \leq \epsilon \norm{\sigma}_1.  
\end{equation}
We also assume $n_0$ is large enough such that for any $n > n_0$,  $r_n \leq 2^{\lfloor \alpha  n \rfloor}$.

Let $\mT_n(\cdot) := \sum_{i=1}^{r_n} \trace\big((\cdot)T_i^{(n)}\big) \ket{e_i}\bra{e_i}$ where we choose $\{\ket{e_i}\}_{i=1}^{r_n}$ to be a subset of the computational basis in $\mH_\tb^{\otimes \lfloor \alpha n \rfloor}$. Without loss of generality, we assume $\lambda_i = \trace(\rho_\theta^{(n)}T_i^{(n)}) > 0$ for all $i$ (we can always exclude the terms that are equal to zero), then  $\mT_n(\rho_\theta^{(n)}) := \sum_{i=1}^{r_n} \lambda_i \ket{e_i}\bra{e_i}$ and 
\begin{align}
F(\rho_\theta^{(n)},\{T_i^{(n)}\}) &\geq F(\Xi_E \circ \mT_n (\rho_\theta^{(n)}) , \{M_i\}^{\otimes n}) \nonumber\\ &= J(\mM^{\otimes n} \circ \Xi_E \circ \mT_n (\rho_\theta^{(n)}))
\nonumber\\ &\geq  J( \Xi_D \circ \mM^{\otimes n} \circ \Xi_E \circ \mT_n (\rho_\theta^{(n)})) ,
\end{align}
where we use  the monotonicity of the QFI in the first and third inequalities and $J(\{p_{i,\theta}\}) = J(\sum_{i}p_{i,\theta}\ket{i}\bra{i})$ for any classical probability distribution $\{p_{i,\theta}\}$ in the second equality. Then we have 
\begin{equation}
\label{eq:capacity-upper}
    \frac{J( \Xi_D \circ \mM^{\otimes n} \circ \Xi_E \circ \mT_n (\rho_\theta^{(n)}))}{F(\rho_\theta^{(n)},\{T_i^{(n)}\})}  \leq 1. 
\end{equation}

Next we aim to show $
\frac{J( \Xi_D \circ \mM^{\otimes n} \circ \Xi_E \circ \mT_n (\rho_\theta^{(n)}))}{F(\rho_\theta^{(n)},\{T_i^{(n)}\})   }$ 
is lower bounded by a constant that approaches $1$ for large $n$. First, assume $n > n_0$, we have 
\begin{multline}
J(\mD_2^{\otimes \lfloor \alpha n \rfloor} \circ \mT_n(\rho_\theta^{(n)})) \\ = J(\mT_n(\rho_\theta^{(n)}))   =  F(\rho_\theta^{(n)},\{T_i^{(n)}\}) = \sum_{i=1}^{r_n} \frac{(\partial_\theta\lambda_i)^2}{\lambda_i},
\end{multline}
where we use $\mD_2^{\otimes \lfloor \alpha n \rfloor}(\ket{e_i}\bra{e_i}) = \ket{e_i}\bra{e_i}$ in the first equality. 
On the other hand, consider 
\begin{multline}
\label{eq:capacity-lower-0}
J(\mD' \circ \Xi_D \circ \mM^{\otimes n} \circ \Xi_E \circ \mT_n (\rho_\theta^{(n)})) 
\\ =: \sum_{i=1}^{r_n} \frac{(\partial_\theta\eta_i)^2}{\eta_i} \leq J( \Xi_D \circ \mM^{\otimes n} \circ \Xi_E \circ \mT_n (\rho_\theta^{(n)}))
\end{multline}
where $\mD'(\cdot) = \sum_i \ket{e_i}\bra{e_i}(\cdot)\ket{e_i}\bra{e_i}$, $\eta_i = \lambda_i + \delta_i$, and $\delta_i = \braket{e_i|(\Xi_D \circ \mM^{\otimes n} \circ \Xi_E - \mD_2^{\otimes \lfloor \alpha n \rfloor}) \circ \mT_n(\rho_\theta^{(n)})|e_i}$.
We will also assume $n$ is large enough such that $\delta_i < \lambda_i$, which is possible due to \eqref{eq:capacity} and the regularity condition~(1). 

Then we have 
\begingroup
\allowdisplaybreaks
\begin{align}
&\quad \, \sum_{i=1}^{r_n} \frac{(\partial_\theta\eta_i)^2}{\eta_i} = \sum_{i=1}^{r_n} \frac{(\partial_\theta\lambda_i+\partial_\theta\delta_i)^2}{\lambda_i+\delta_i} \nonumber\\
&= \sum_{i=1}^{r_n} (\partial_\theta\lambda_i+\partial_\theta\delta_i)^2\left( \frac{1}{\lambda_i} - \frac{\delta_i}{\lambda_i^2(1+\xi_i)^2}\right)\nonumber\\
&\geq \sum_{i=1}^{r_n} \left((\partial_\theta\lambda_i)^2+2\partial_\theta\lambda_i\partial_\theta\delta_i\right)\frac{1}{\lambda_i}\left( 1 - \frac{\abs{\delta_i}}{\lambda_i}\right)\nonumber\\
&\geq F(\rho_\theta^{(n)},\{T^{(n)}_i\}) - \sum_{i=1}^{r_n}  \frac{(\partial_\theta\lambda_i)^2}{\lambda_i^2}\abs{\delta_i} - \sum_{i=1}^{r_n}  \frac{2\abs{\partial_\theta\lambda_i\partial_\theta\delta_i}}{\lambda_i}\nonumber
\\
&\geq F(\rho_\theta^{(n)},\{T^{(n)}_i\}) \nonumber\\ 
&\qquad \quad - \epsilon \max_i \left(   {\frac{(\partial_\theta\lambda_i)^2}{\lambda_i^2}} + 2 \bigg|\frac{\partial_\theta\lambda_i}{\lambda_i}\bigg| 
J(\rho_\theta^{(n)})^{1/2}
\right). 
\label{eq:capacity-lower-1}
\end{align}
\endgroup
In the second equality above, we use the Taylor expansion $\frac{1}{1+\delta_i/\lambda_i} = 1 - \frac{\delta_i/\lambda_i}{(1+\xi_i)^2}$ for some $\xi_i \in [0,\delta_i/\lambda_i]$. In the last inequality above, we use \eqref{eq:capacity} to derive that $\sum_i \abs{\delta_i} \leq \epsilon \|\rho_\theta^{(n)}\|_1  = \epsilon$ and 
\begin{multline}
\sum_i \big|\partial_\theta\delta_i\big| \leq \epsilon \|\partial_\theta \mT_n(\rho_\theta^{(n)})\|_1  \\ \leq  \epsilon J(\mT_n(\rho_\theta^{(n)}))^{1/2} \leq \epsilon J(\rho_\theta^{(n)})^{1/2}. 
\end{multline}
Here we use the inequality $\|\partial_\theta \zeta_\theta \|_1 \leq J(\zeta_\theta)^{1/2}$ for any classically mixed state $\zeta_\theta = \sum_i \mu_i \ket{i}\bra{i}$, which is true because $(\sum_i |\partial_\theta \mu_i|)^2 \leq \sum_i (\partial_\theta \mu_i)^2/\mu_i$ from the Cauchy--Schwarz inequality. Note that it also holds that $\|\partial_\theta \sigma_\theta \|_1 \leq J(\sigma_\theta)^{1/2}$ for general mixed states $\sigma_\theta$~\cite{oszmaniec2016random,jarzyna2020geometric}.

\begin{figure}[tb]
    \centering
    \includegraphics[width=0.48\textwidth]{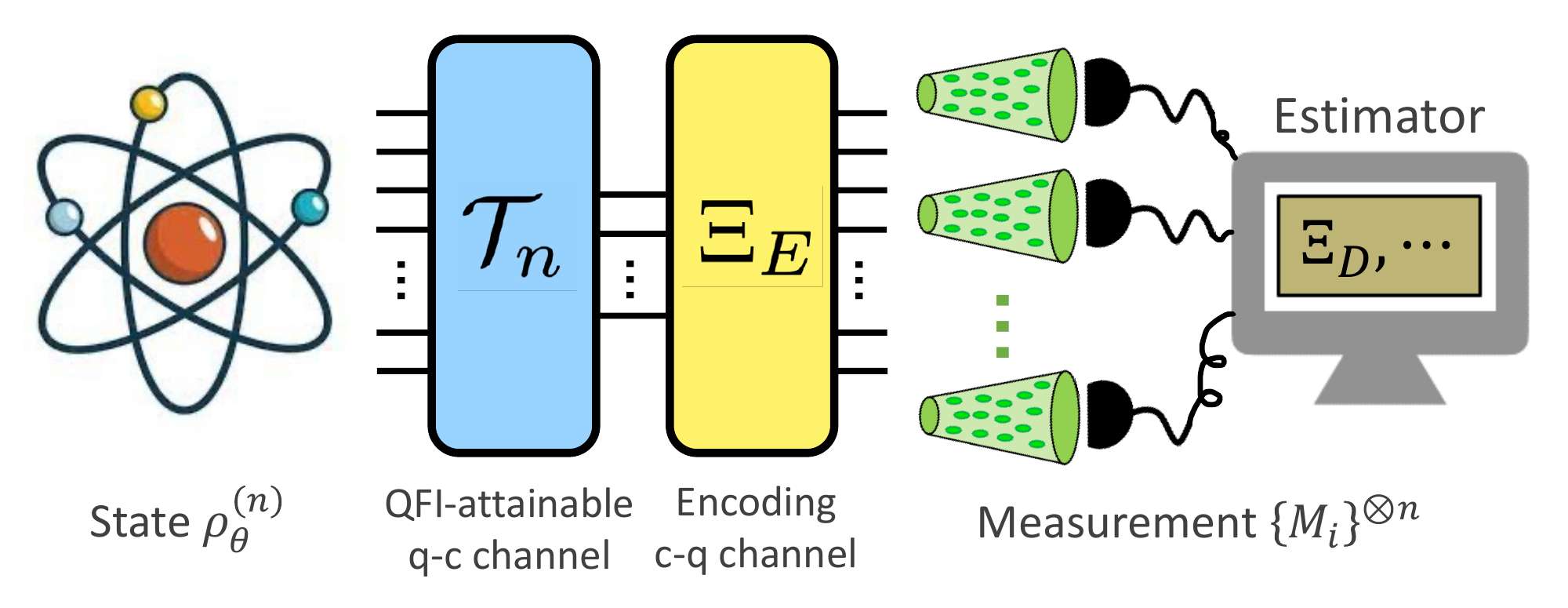}
    \caption{A quantum state $\rho_\theta^{(n)}$ in an $n$-partite system is estimated using $n$ identical noisy measurements acting on each subsystem, described by $\{M_i\}^{\otimes n}$. The QPFI can approach the QFI in the asymptotic limit $n\rightarrow \infty$ if the sufficient condition in \thmref{thm:metrology-capacity} is satisfied. The optimal control is the composition of a quantum-classical channel $\mT_n(\cdot) = \sum_{i=1}^{r_n} \trace\big((\cdot)T_i^{(n)}\big) \ket{e_i}\bra{e_i}$ where the measurement $\{T_i^{(n)}\}$ is asymptotically QFI-attainable, and an encoding channel $\Xi_E$ chosen as the optimal encoding channel for $\mM^{\otimes n}$ from the HSW theorem. Note that the decoding channel $\Xi_D$ from the HSW theorem only needs to be used in a classical post-processing manner. 
    }
    \label{fig:asymptotic}
\end{figure}

Finally, from the monotonicity of the QFI (i.e., $\sum_i \frac{(\partial_\theta\lambda_i)^2}{\lambda_i} \leq J(\rho_\theta^{(n)})$) and the regularity conditions~(1) and (2), we have 
\begin{equation}
\label{eq:upper-subexp}
\max_i{\frac{(\partial_\theta\lambda_i)^2}{\lambda_i^2}} +  2\bigg|\frac{\partial_\theta\lambda_i}{\lambda_i}\bigg| \norm{\partial_\theta\rho}_1  \leq J(\rho_\theta^{(n)}) e^{o(n)}. 
\end{equation}
Taking the limit $n\rightarrow \infty$ in \eqref{eq:capacity-lower-1}, from $\epsilon = e^{-\Omega(n)}$ and \eqref{eq:upper-subexp}, we have 
\begin{equation}
\label{eq:capacity-lower-2}
\lim_{n\rightarrow \infty} \frac{1}{F(\rho_\theta^{(n)},r\{T^{(n)}_i\})}\sum_{i=1}^{r_n} \frac{\partial_\theta\eta_i^2}{\eta_i} \geq 1. 
\end{equation}
Combining \eqref{eq:capacity-upper}, \eqref{eq:capacity-lower-0}, and \eqref{eq:capacity-lower-2}, we have 
\begin{equation}
\lim_{n\rightarrow \infty} \frac{J( \Xi_D \circ \mM^{\otimes n} \circ \Xi_E \circ \mT_n (\rho_\theta^{(n)}))}{F(\rho_\theta^{(n)},\{T_i^{(n)}\})}  = 1. 
\end{equation}
Since $\lim_{n\rightarrow \infty} \frac{F(\rho_\theta^{(n)},\{T^{(n)}_i\})}{J(\rho_\theta^{(n)})} = 1$ and 
\begin{multline}
F^{\bP}(\rho_\theta^{(n)},\{M_i\}^{\otimes n}) 
\geq F(\Xi_E \circ \mT_n (\rho_\theta^{(n)}) , \{M_i\}^{\otimes n}) \\
\geq J( \Xi_D \circ \mM^{\otimes n} \circ \Xi_E \circ \mT_n (\rho_\theta^{(n)})),
\end{multline}
we must have 
\begin{equation}
\lim_{n\rightarrow \infty} \frac{F^{\bP}\left(\rho_\theta^{(n)},\{M_{i}\}^{\otimes n}\right)}{J(\rho_\theta^{(n)})} = 1, 
\end{equation}
proving the theorem. 
\end{proof}

\subsection{Discussion}
\label{sec:discussion-asymptotic}

Here we discuss the intuitions behind the sufficient condition in \thmref{thm:metrology-capacity} and describe the relevant situations where it is satisfied. We will see that the sufficient condition is satisfied 
for a generic class of quantum states $\rho_\theta^{(n)}$ and noisy measurements $\{M_i\}$. 

Let us first explain the meaning of the condition \eqref{eq:QFI-attainable}. It states that there exists an (asymptotically) QFI-attainable measurement for $\rho_\theta^{(n)}$ that has a \emph{small} number of measurement outcomes. Specifically, the number of measurement outcomes $r_n$ should be smaller than $2^{C(\mM)n}$ (asymptotically) where $C(\mM)$ is the classical capacity of the quantum measurement $\{M_i\}$ under consideration, i.e., \thmref{thm:metrology-capacity} applies when 
\begin{equation}
\label{eq:subexponential}
\log r_n < C(\mM) n + o(n). 
\end{equation}

The requirement (\eqref{eq:subexponential}) is satisfied by many practically relevant quantum states and measurements. 
In fact, whenever the classical capacity of $\mM$ is positive, $r_n = e^{o(n)}$ is a sufficient (but not necessary) condition of \eqref{eq:subexponential}. Below we provide several typical examples where the QFI-attainable measurement with a subexponential number of outcomes exists. See \appref{app:QFI-attainable} for additional details. 
\begin{enumerate}[(1),wide, labelwidth=!, labelindent=0pt,nolistsep]
\item \emph{Low-rank states.~}For pure states, it was known that there exist 2-outcome QFI-attainable measurements~\cite{braunstein1994statistical}. (Note that~\cite{len2021quantum} contains another proof of \thmref{thm:metrology-capacity} when $\rho_\theta^{(n)}$ is pure.) More generally, any $\rho_\theta^{(n)}$ that is supported on a subspace with a subexponential dimension also has a QFI-attainable measurement with a subexponential number of outcomes. 
\item \emph{Symmetric states.~}The second example with a QFI-attainable measurement with a subexponential number of outcomes is symmetric (permutation-invariant) states (e.g., tensor products of $n$ identical mixed states). According to the Schur--Weyl duality~\cite{fulton2013representation,hayashi2017group}, $\mH_S = (\bC^{D})^{\otimes n}$ can be decomposed as $\bigoplus_{\nu} (\mH_\nu({\rm U}(D)) \otimes \mH_\nu(S_n))$, where $\mH_\nu({\rm U}(D))$ and $\mH_\nu(S_n)$ are irreducible representation spaces of the unitary group ${\rm U}(D)$ and the permutation group $S_n$ with index $\nu$. Any symmetric state $\rho_\theta^{(n)}$ can be written as 
\begin{equation}
    \rho_\theta^{(n)} = \bigoplus_{\nu} \left( p_\nu \rho_\nu^{(n)} \otimes \frac{\id_\nu}{\dim(\mH_\nu(S_n))} \right), 
\end{equation}
where $\rho_\nu^{(n)}$ are mixed states acting on $\mH_\nu({\rm U}(D))$ and $p_\nu$ satisfies $\sum_\nu p_\nu = 1$ (both of which can be functions of $\theta$). Then a QFI-attainable measurement with a subexponential number of outcomes $\{\bigoplus_\nu (T_i)_\nu \otimes \id_\nu\}$ of $\rho_\theta^{(n)}$ can be constructed from a QFI-attainable measurement $\{\bigoplus_\nu (T_i)_\nu\}$ of $\bigoplus_{\nu} p_\nu \rho_\nu^{(n)}.$
Let us estimate the number of measurement outcomes: $\nu$ corresponds to Young diagrams (i.e., partitions of $n$ into $D$ parts), implying the number of different indices $\nu$ is $O(n^{D-1})$.
For any $\nu$, $\dim(\mH_\nu({\rm U}(D))$ is equal to the number of semistandard Young tableaux, which is at most $O(n^{D(D-1)/2})$ according to the Weyl dimension formula~\cite{goodman2000representations}. The number of measurement outcomes is thus upper bounded by $\sum_\nu \dim(\mH_\nu({\rm U}(D)) = O(n^{(D-1)(D/2+1)}).$

\item \emph{Gibbs states.~}For classically mixed states $\rho_\theta^{(n)}$, the projection onto the eigenstates of $\rho_\theta^{(n)}$ is QFI-attainable but has exponentially many measurement outcomes. However, we argue that in many cases, a subexponential number of projections onto direct sums of eigenspaces are sufficient to attain the QFI up to the leading order, so that \thmref{thm:metrology-capacity} applies. For instance, consider the Gibbs state 
\begin{equation}
    \rho_\theta^{(n)} = \frac{1}{\sum_\nu e^{-\theta E_\nu}} \sum_\nu e^{-\theta E_\nu}\ket{\nu}\bra{\nu},
\end{equation}
where $\{\ket{\nu}\}$ are energy eigenstates with eigenvalues $\{E_\nu\}$ and $\theta$ is the inverse temperature to be estimated. The QFI is equal to the variance of energy, i.e., 
\begin{equation}
J(\rho_\theta^{(n)}) = \sum_\nu p_\nu E_\nu^2 
- \bigg(\sum_\nu p_\nu E_\nu \bigg)^2,     
\end{equation} 
where $p_\nu = \frac{e^{-\theta E_\nu } }{ \sum_\nu e^{-\theta E_\nu}}$. Assume the energy eigenvalues lie in $[0,\widehat{E})$, where $\widehat{E} = \Theta(n)$ (which is a standard assumption in condensed matter systems) and divide them into intervals $\{I_k = [E_k,E_{k+1})\}_{k=1}^{n^2}$ such that $E_0 = 0$, $E_n = \widehat{E}$ and $E_{k+1} - E_k = \Delta E = \widehat{E}/n^2$. Consider the projections $\{\Pi_{k}\}_{k=1}^n$ onto the direct sums of eigenspaces corresponding to all eigenvalues in $I_k$. The FI is 
\begin{equation}
    F(\rho_\theta^{(n)},\{\Pi_k\}) = \sum_{k=1}^n p_k E_k^2 - \bigg(\sum_\nu p_\nu E_\nu \bigg)^2, 
\end{equation}
where $p_k = \sum_{\nu:E_\nu \in I_k} p_\nu$ and $E_k = \frac{1}{p_k} \sum_{\nu:E_\nu \in I_k} p_\nu E_\nu$. Then we have $J(\rho_\theta^{(n)}) - F(\rho_\theta^{(n)},\{\Pi_k\}) \leq \sum_k p_k (E_{k+1}^2 - E_{k}^2) \leq 2 \widehat{E} \Delta E = \Theta(1)$. Combining with the regularity condition~(2), it implies that $F(\rho_\theta^{(n)},\{\Pi_k\})$ is equal to $J(\rho_\theta^{(n)})$ up to the leading order. 
\end{enumerate}

Next, let us explain the intuitions behind the regularity conditions: 
\begin{enumerate}[(1),wide, labelwidth=!, labelindent=0pt,nolistsep]
    \item {Regularity condition~(1)} states that when the probability of obtaining measurement outcome $i$ depends on $\theta$ (i.e., $\partial_\theta \lambda_i \neq 0$), it must be no smaller than an inverse of a subexponential function of $n$, that is, the probability to detect $i$ cannot be exponentially small. This is also a practically reasonable assumption as we would want to exclude the singular cases where an exponentially small signal provides a non-trivial contribution to the QFI. 
    \item {Regularity condition~(2)} requires that the QFI of $\rho_\theta^{(n)}$ does not decrease with $n$ asymptotically, which should be satisfied in any practically relevant cases. It also requires the QFI to be subexponential, which is a natural assumption in quantum sensing experiments (note that the Heisenberg limit implies $J(\rho_\theta^{(n)}) = O(n^2)$).  
\end{enumerate}

Lastly, we briefly comment on the resources required to implement optimal preprocessing controls. 
First, the total number of ancillary qubits required to implement the desired preprocessing channel $\Xi_E\circ\mT_n$ is at most $O(n)$, because in general $\log(D^n r_n)$ ancillary qubits are sufficient to implement the QFI-attainable q-c channel $\mT_n$ and another $\log(D^n r_n)$ ancillary qubits are sufficient to implement the encoding channel $\Xi_E$. The gate complexity to implement $\mT_n$ is expected to depend on the structure of the quantum state $\rho_\theta^{(n)}$. For example, for symmetric states, the Schur transform, efficiently implementable~\cite{bacon2006efficient}, can be an important step in $\mT_n$. Unitary gates that are used in aligning the output basis of $\mT_n$ to the input basis of the encoding channel $\Xi_E$ should also be taken into consideration. For example, in the special case where $\rho_n$ is a low-rank classically mixed state, $\mT_n$ should be a rotation that matches $\rho_n$ eigenstates to the input basis of $\Xi_E$. The gate complexity to implement the optimal encoding channel $\Xi_E$ is high in general. However, when $r_n$ is subexponential (as we discussed above), the encoding channel does not need to be capacity-achieving as it only needs to reliably transmit an exponentially small amount of information, potentially making it relatively easier to implement (the details are left for future discussion).  For example, when $r_n = 2$, a simple repetition code mapping $\ket{0}$ to $\ket{0}^{\otimes n}$ and $\ket{1}$ to $\ket{1}^{\otimes n}$ will be optimal.  Finally, note that although we have shown that $\Xi_E\circ\mT_n$ is optimal, other simpler optimal preprocessing channels may still exist. For example, for pure states, unitary controls are optimal according to \thmref{thm:unitary}, requiring no ancillas; and a design of an optimal preprocessing unitary is presented in~\cite{len2021quantum}.

\subsection{Examples}

Lastly, we present three simple but natural examples with powerful global preprocessing controls that can be efficiently implemented using $O(\log^2 n)$-depth circuits, assuming arbitrary two-qubit gates and all-to-all connectivity (see details in \appref{app:circuit}). In these three examples, we always assume $\mH = \mH' = {\rm span}\{\ket{0},\ket{1}\}$ are qubit systems and the  quantum measurement is $\{M_i\} = \{M_0,M_1\}$ where $M_0 = (1-m)\ket{0}\bra{0} + m\ket{1}\bra{1}$ ($0 < m < 1/2$) and $M_1 = \id - M_0$.

\begin{figure}[tbp]
    \centering
    \includegraphics[width=0.49\textwidth]{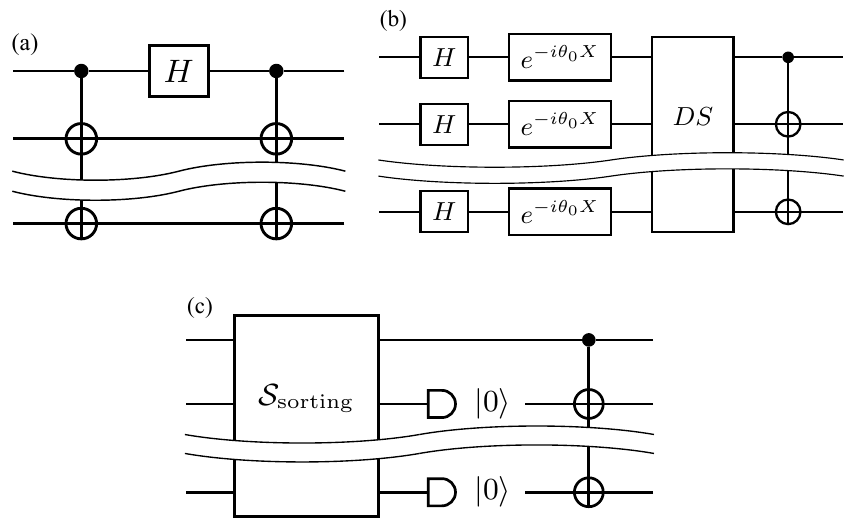}
\caption{Examples of preprocessing circuits ($U^{\textsc{G}}$ or $\mE^{\textsc{G}}$). 
(a) Phase sensing using GHZ states. The optimal circuit is composed of two C(NOT)$^{n-1}$ gates, which, conditioning on the first qubit being $\ket{1}$, performs $X^{\otimes {n-1}}$ ($X$ is the Pauli-X gate) on the remaining $n-1$ qubits, and a Hadamard gate $H$ acting on the first qubit between them. The circuit depth is $O(\log n)$. (b) Phase sensing using product states. The optimal circuit is composed of a global $H$ gate and a global Pauli-X rotation, a desymmetrization gate $DS$ that maps W state $\frac{1}{\sqrt{n}}(\ket{10\cdots 0} + \ket{010\cdots 0} + \cdots + \ket{0\cdots 01})$ to $\ket{10\cdots 0}$ and $\ket{0\cdots 0}$ to $\ket{0\cdots 0}$, and a C(NOT)$^{n-1}$ gate. The circuit depth is $O(\log n)$. (c) Phase sensing using classically mixed states. $\mS_{\rm sorting}$ is a sorting channel with circuit depth $O(\log^2 n)$, that uses $O(n\log^2 n)$ ancillary qubits and outputs one qubit in state \eqref{eq:classical-first-qubit}. It first sorts the bit-string and then swaps the first qubit with the $\lfloor n \sin^2\theta_0 \rfloor$-th qubit. (The D-shape detectors mean the qubits are completely discarded.)}
    \label{fig:circuit}
\end{figure}

In the first two examples, our preprocessing circuits manage to achieve a FI that is asymptotically equal to the QFI for any noise rate $m$. In the third example, our preprocessing circuit achieves a FI that is asymptotically equal to $2/\pi$ of the QFI, which still beats local controls in the noise regime $m \geq 0.1011$. The guideline to design these circuits is to convert the quantum state to a two-level state in ${\rm span}\{\ket{0}^{\otimes n},\ket{1}^{\otimes n}\}$ whose probability (or amplitude) distribution encodes $\theta$. Then a \emph{majority voting post-processing method} can be used to estimate $\theta$ with a vanishingly small measurement error. Specifically, in the majority voting post-processing method, we partition the measurement outcomes from measuring the two-level state using $\{M_0,M_1\}^{\otimes n}$, which are represented by $n$-bit strings in $\{0,1\}^n$, into two sets depending on whether the Hamming weight of the string is larger than $\lfloor n/2\rfloor$. The FI of this binary probability distribution achieves the desired value asymptotically.

The first example is phase sensing using GHZ states~\cite{giovannetti2006quantum}, where
\begin{equation}
\label{eq:example-GHZ}
    \ket{\psi_\theta^{(n)}} = \frac{e^{i n \theta}\ket{0}^{\otimes n} + e^{-i n \theta}\ket{1}^{\otimes n}}{\sqrt{2}}, 
\end{equation}
and an optimal preprocessing circuit $U^{\textsc{G}}$ that achieves 
\begin{equation}
F(U^{\textsc{G}}\psi_\theta^{(n)}(U^{\textsc{G}})^\dagger,\{M_i\}^{\otimes n}) \xrightarrow{n\rightarrow \infty} J(\psi_\theta^{(n)}) = 4n^2    
\end{equation}
is shown in \figaref{fig:circuit}, mapping  $\ket{\psi_\theta^{(n)}}$ to 
\begin{equation}
    \cos(n\theta)\ket{0}^{\otimes n} + i\sin(n\theta)\ket{1}^{\otimes n}. 
\end{equation}
The majority voting post-processing method gives an optimal estimator of $\theta$. 

The second example is phase sensing using product pure states (usually known as Ramsey interferometry~\cite{ramsey1950molecular}), where 
\begin{equation}
\label{eq:example-product-pure}
    \ket{\psi_\theta^{(n)}} = \left(\frac{e^{i \theta}\ket{0} + e^{-i \theta}\ket{1}}{\sqrt{2}}\right)^{\otimes n}, 
\end{equation}
and an optimal preprocessing circuit of depth $O(\log n)$ that achieves 
\begin{equation}
F(U^{\textsc{G}}\psi_\theta^{(n)}(U^{\textsc{G}})^\dagger,\{M_i\}^{\otimes n}) \xrightarrow{n\rightarrow \infty}  J(\psi_\theta^{(n)}) = 4n    
\end{equation}
is shown in \figbref{fig:circuit}. Here we assume $\theta_0$ is a rough estimate of $\theta$ such that $|\theta - \theta_0| \ll 1/\sqrt{n}$. The first step is to implement global Hadamard gates and Pauli-X rotations such that $\ket{\psi_\theta^{(n)}}$ is mapped to $(e^{-i\theta_0 X})^{\otimes n} H^{\otimes n} \ket{\psi_\theta^{(n)}} = \left(\cos(\theta-\theta_0)\ket{0} + i\sin(\theta-\theta_0)\ket{1}\right)^{\otimes n}$. The second step is to apply a desymmetrization gate $DS$ and a C(NOT)$^{n-1}$ gate such that the state is approximately mapped to 
\begin{equation}
\cos(\sqrt{n}(\theta-\theta_0))\ket{0}^{\otimes n} + i\sin(\sqrt{n}(\theta-\theta_0))\ket{1}^{\otimes n}, 
\end{equation}
with an error $O(n(\theta-\theta_0)^2)$. The majority voting post-processing method gives an optimal estimator of $\theta$.  

The third example is phase sensing using classically mixed states (which can be seen as \eqref{eq:example-product-pure} after global Hadamard gates and completely dephasing noise), where 
\begin{equation}
\label{eq:example-product-mixed}
    \rho_\theta^{(n)} = \left(\cos^2\theta \ket{0}\bra{0} + \sin^2\theta \ket{1}\bra{1}\right)^{\otimes n}. 
\end{equation}
We assume $\theta \in (0+\varepsilon,\pi/4-\varepsilon)$ for some constant $\varepsilon > 0$. 
We show a preprocessing channel $\mE^{\textsc{G}}$ (in \figcref{fig:circuit}) of circuit depth $O(\log^2 n)$ that achieves 
\begin{equation}
    F(\mE^{\textsc{G}}(\rho_\theta^{(n)}),\{M_i\}^{\otimes n}) \xrightarrow{n\rightarrow \infty}  \frac{8}{\pi}n = \frac{2}{\pi}J(\rho_\theta^{(n)}). 
\end{equation}
After a sorting channel $\mS_{\rm sorting}$ and discarding all qubits except the first qubit, the first qubit is in state 
\begin{equation}
\label{eq:classical-first-qubit}
    p_\theta\ket{0}\bra{0} + (1-p_\theta) \ket{1}\bra{1}, 
\end{equation}
where $p_\theta$ is the probability that after flipping $n$ coins whose probability of getting heads are $\sin^2\theta$, the number of heads are smaller than or equal to $\lfloor n \sin^2\theta_0 \rfloor$ ($\theta_0$ is a rough estimate of $\theta$ satisfying $\abs{\theta - \theta_0} \ll 1/\sqrt{n}$). A FI asymptotically equal to $\frac{8}{\pi}n$ can then be achieved using a C(NOT)$^{n-1}$ gate with $n-1$ ancillas initialized in $\ket{0}^{\otimes n-1}$ and the majority voting post-processing method. 

Note that $\rho_\theta^{(n)}$ is a symmetric state. According to the discussion in \secref{sec:discussion-asymptotic}, the QPFI should be asymptotically equal to the QFI, but whether there exists an efficient implementation of the optimal preprocessing circuits is unknown. 
Here we demonstrate the advantage of global controls by providing an efficient but suboptimal circuit in \figcref{fig:circuit}.
The first part ($\mS_{\rm sorting}$) of our circuit can be viewed as the optimal quantum-classical channel $\mT_n$ in \thmref{thm:metrology-capacity}. The second part that encodes one qubit into $n$ qubits is, however, suboptimal. (In order to faithfully transmit all probability distribution information, the encoding channel in the second part needs to encode $\log(n+1)$ qubits into $n$ qubits.)

Nonetheless, our circuit in \figcref{fig:circuit} is superior to any local preprocessing controls in the noise regime where 
\begin{equation}
m \geq \frac{1}{2}-\frac{1}{\sqrt{2\pi}} \approx 0.101.       
\end{equation}
This can be proven noting that the optimal FI achievable using arbitrary local channels satisfies
\begin{multline}
\max_{\mE = \mE_1\otimes \mE_2 \otimes\cdots \otimes \mE_n}F(\mE(\rho_\theta^{(n)}),\{M_i\}^{\otimes n}) \\ = F^\bP(\rho_\theta^{(1)},\{M_i\}) n \leq (1-2m)^2 n,  
\end{multline}
from \eqref{eq:qubit-binary-classical}. Thus it is always smaller than $\frac{8}{\pi}n$ when $m \geq \frac{1}{2}-\frac{1}{\sqrt{2\pi}}$. Specially, when $m \rightarrow \frac{1}{2}$, the linear constant of the locally optimized FI is vanishingly small, while the supoptimal global one is still above a positive number. 

\section{Conclusions and outlook}
\label{sec:conclusion}

We conducted a systematic study of the preprocessing optimization problem for noisy quantum measurements in quantum metrology. The QPFI (i.e., the FI of noisy measurement statistics optimized over all preprocessing quantum channels), that we defined and investigated in depth, sets an ultimate precision bound for noisy measurement of quantum states. Our results provide, in many cases, both numerically and analytically, approaches to identifying the optimal preprocessing controls that will be of great importance in alleviating the effect of measurement noise in quantum sensing experiments. 

We also considered, specifically, the asymptotic limit of the QPFI in multi-probe systems with individual measurement on each probe. We proved the convergence of the QPFI to the QFI when there exists an (asymptotically) QFI-attainable measurement with a sufficiently small number of measurement outcomes, by establishing a  connection to the classical channel capacity theorem. It would be interesting to explore, in future works, if the number of outcomes for QFI-attainable measurements can be easily bounded given a quantum state.

Although we've discussed only two types of quantum preprocessing controls, CPTP maps and unitary maps, our biconvex formulation might be generalized to cover other more restricted types of quantum controls. We also narrowed the analytical forms of optimal controls for pure states and classically mixed states down to rotations onto the span of two basis states and coarse-graining channels, respectively, but it remains open whether a simple method exists to help us determine the exact operations.  

Finally, there are a few important directions to extend our results to, e.g., incorporating the state preparation optimization into the QPFI optimization problem, considering the preprocessing optimization in multi-parameter estimation where the incompatibility of optimal preprocessings for different parameters might become an issue, and finding optimal preprocessings for other information processing tasks beyond quantum metrology such as state tomography and discrimination.

~~~

~~~

\section*{Acknowledgments}

We thank Senrui Chen, Kun Fang, Jan Ko{\l}ody{\'n}ski, Yaodong Li, Zi-Wen Liu, John Preskill, Alex Retzker, Mark Wilde, and Tianci Zhou for helpful discussions. 
The authors acknowledge funding provided by the Institute for Quantum Information and Matter, an NSF Physics Frontiers Center (NSF Grant PHY-1733907). T.G. further acknowledges funding provided by the Quantum Science and Technology Scholarship of the Israel Council for Higher Education.

\bibliography{refs-measurement}{}

\onecolumngrid
\newpage
\appendix

\setcounter{theorem}{0}
\setcounter{proposition}{0}
\setcounter{lemma}{0}
\setcounter{figure}{0}
\renewcommand{\thefigure}{S\arabic{figure}}
\renewcommand{\thelemma}{S\arabic{lemma}}
\renewcommand{\thetheorem}{S\arabic{theorem}}
\renewcommand{\thecorollary}{S\arabic{corollary}}
\renewcommand{\theproposition}{S\arabic{proposition}}
\renewcommand{\theHfigure}{Supplement.\arabic{figure}}
\renewcommand{\theHlemma}{Supplement.\arabic{lemma}}
\renewcommand{\theHtheorem}{Supplement.\arabic{theorem}}
\renewcommand{\theHcorollary}{Supplement.\arabic{corollary}}

\section{QFI-attainable measurements}
\label{app:QFI-attainable}

In this appendix, we provide several simple facts about the QFI-attainable measurements for quantum states, that will be useful in the main text. 

\begin{enumerate}[(1)]
    \item \emph{Necessary and sufficient condition~\cite{braunstein1994statistical}.} Given a quantum state $\rho_\theta$, $\{M_i\}$ attains the QFI if and only if 
    \begin{equation}
    \label{eq:QFI-attainable-condition}
        \forall i=1,\ldots,r,\;\exists c_i \in \bR,\text{  s.t.  } M_i^{1/2} \rho_\theta^{1/2} = c_i M_i^{1/2} L_\theta \rho_\theta^{1/2},
    \end{equation}
    where $L_\theta$ is the symmetric logarithmic derivative (SLD) operator which is a Herminian operator defined by $\partial_\theta \rho_\theta = \frac{1}{2}(L_\theta \rho_\theta +  \rho_\theta L_\theta)$ and the QFI $J(\rho_\theta) = \trace(\rho_\theta L_\theta^2)$. (Note that although the definition of $L_\theta$ is not unique when $\rho_\theta$ is not full rank, one can chose an arbitrary one and the QFI and the necessary and sufficient condition will be invariant.) Clearly, the necessary and sufficient condition \eqref{eq:QFI-attainable-condition} is satisfied when the measurement is chosen as the rank-one projection onto eigenstates of $L_\theta$. 
    \item  \emph{Low-rank states.} If $\rho_\theta$ and $\partial_\theta\rho_\theta$ are supported on a $D'$-dimensional subspace of $\mH_{S}$ ($D' \leq D$) (which should be true when $\rho_\theta$ is supported on a $\lfloor D'/2 \rfloor$-dimensional subspace), there exists a SLD $L_\theta$ supported on the $D'$-dimensional subspace. Moreover, let $L_\theta = \sum_{i=1}^{D'} \ell_i \ket{i}\bra{i}$ where $\ell_i > 0$ and $\{\ket{i}\}$ is an orthonormal basis for the $D'$-dimensional subspace, it can be verified that $\{M_i\}_{i=1}^{D'}$ (that has $D'$ measurement outcomes) is QFI-attainable when $M_i = \ket{i}\bra{i}$ for $i=1,\ldots,D'-1$ and $M_{D'} = \id - \sum_{i=1}^{D'} M_i$. That means, a state of a subexponential rank must have a QFI-attainable measurement with a subexponential number of measurement outcomes. In particular, for pure states, binary measurements are sufficient to attain the QFI.
    \item \emph{Classically mixed states.} For classically mixed states $\rho_\theta = \sum_{i=1}^D \lambda_i \ket{i}\bra{i}$, the SLD is also diagonal in the basis $\{\ket{i}\}$. It implies that the rank-one projection onto the basis is QFI-attainable for classically mixed states. 
    \item \emph{Block-diagonal states.} More generally, consider block-diagonal states $\rho_\theta = \bigoplus p_{\nu,\theta} \rho_{\nu,\theta}$, where $\rho_{\nu,\theta}$ are supported on different orthogonal subspaces. The SLD can also be block-diagonal, implying that there exists a QFI-attainable measurement of the form $\{\bigoplus_\nu (T_i)_\nu\}$, whose number of measurement outcomes is at most the rank of $\rho_\theta$, where $\{(T_i)_\nu\}$ is a quantum measurement for each fixed $\nu$. 
\end{enumerate}

\section{Mathematical properties of the QPFI}
\label{app:property}

Here we briefly discuss the mathematical properties of the QPFI. 
We will always assume quantum measurements are non-trivial in this appendix (i.e., $\exists M_i \not\propto \id$ for all $\{M_i\}$) . 

\begin{enumerate}[(1)]
    \item \emph{Faithfulness.} The QPFI is faithful, that means, \begin{equation}
        F^\bP(\rho_\theta,\{M_i\}) \geq 0,
    \end{equation}
    and the equality holds if and only if $\partial_\theta \rho_\theta = 0$. The non-negativity follows from the faithfulness of the classical FI. In order to see when $\partial_\theta \rho_\theta \neq 0$, $F^{\bP}(\rho_\theta,\{M_i\}) > 0$, we use \lemmaref{lemma:lower} in \secref{sec:general}. Assume $\partial_\theta \rho_\theta \neq 0$.  $F^{\bP}(\rho_\theta,\{M_i\}) > F^{\bP}(\mT(\rho_\theta),\{M_i\})$, where $\mT(\rho_\theta)$ is a classically mixed state that satisfies $J(\mT(\rho_\theta)) = J(\rho_\theta) > 0$ (from the faithfulness of the QFI). Consider the following simplification of measurements: dividing $\{M_i\}$ into two subsets and restricting them in a two-dimensional subspace such that the measurement becomes a non-trivial binary measurement on a qubit $\{\widetilde{M},\id-\widetilde{M}\}$. Then we see that $F^{\bP}(\mT(\rho_\theta),\{M_i\}) \geq F^{\bP}(\mT(\rho_\theta),\{\widetilde{M},\id-\widetilde{M}\}) > 0$ from the exact expression of the QPFI in \eqref{eq:classical-binary}. 
    \item \emph{QFI as an upper bound. } The QPFI is always no larger than the QFI. Note that 
    \begin{equation}
    F(\mE(\rho_\theta),\{M_i\}) = F(\rho_\theta,\{\mE^\dagger(M_i)\}) \leq J(\rho_\theta),    
    \end{equation} 
    where $\mE^\dagger$ is the dual map of $\mE$ and the second inequality follows from the definition of the QFI and the fact that $\{\mE^\dagger(M_i)\}$ is still a POVM when $\mE$ is a CPTP map. Taking the supremum over $\mE$ on both sides in the above inequality, we've proved $F^\bP(\rho_\theta,\{M_i\}) \leq J(\rho_\theta)$.
    \item \emph{Monotonicity.} The QPFI is monotonic, i.e., $F^\bP(\rho_\theta,\{M_i\}) \geq F^\bP(\mE(\rho_\theta),\{M_i\})$ for any CPTP map $\mE$ by definition. Or equivalently, $F^\bP(\rho_\theta,\{M_i\}) \geq F^\bP(\rho_\theta,\{\mE^\dagger(M_i)\})$ where $\mE^\dagger$ is the dual map of $\mE$.
    \item \emph{Convexity.} 
    The QPFI is convex in quantum states when $\{M_i\}$ is fixed. That is, for $p \in (0,1)$ independent of $\theta$, 
    \begin{align}
        F^\bP(p\rho_{\theta}+(1-p)\sigma_\theta,\{M_i\})
        &= \sup_{\mE} F(\mE(p\rho_{\theta}+(1-p)\sigma_\theta,\{M_i\}) \\
        &\leq \sup_{\mE} p F(\mE(\rho_{\theta}),\{M_i\})
        + (1-p) F(\mE(\sigma_\theta),\{M_i\}) \\
        &\leq p F^\bP(\rho_{\theta},\{M_i\}) + (1-p) F^\bP(\sigma_\theta,\{M_i\}), 
    \end{align}
    where we use the convexity of the classical FI in the second step. Similarly, we also have 
    \begin{equation}
        F^\bP(\rho_\theta,\{p M_i+ (1-p) M'_i\}) \leq  p F^\bP(\rho_\theta,\{M_i\}) + (1-p) F^\bP(\rho_\theta,\{M'_i\}).
    \end{equation}
    \item \emph{Additivity.} When quantum states under consideration are pure, the QPFI is additive because of \eqref{eq:gamma} and the additivity of the QFI, i.e., $F^\bP(\psi_\theta \otimes \psi_\theta',\{M_i\}) = F^\bP(\psi_\theta,\{M_i\}) + F^\bP(\psi_\theta',\{M_i\})$ when $\psi_\theta$ and $\psi_\theta'$ are pure. 
    
    For general mixed states, there is no general inequality relation between $F^\bP(\rho_\theta \otimes \sigma_\theta,\{M_i\})$ and $F^\bP(\rho_\theta,\{M_i\}) + F^\bP(\sigma_\theta,\{M_i\})$. Consider the simple example where one (or two identical) qubit state $\rho_\theta = \cos^2\theta \ket{1}\bra{1} + \sin^2\theta \ket{2}\bra{2}$ ($\theta \in (0,\pi/2)$), is measured using the binary measurement on a qubit: $M_1 = (1-m)\ket{1}\bra{1} + m\ket{2}\bra{2}$ and $M_2 = \id - M_1$ ($0 < m < 1/2$). We have, from \eqref{eq:classical-binary}, that 
    \begin{equation}
    F^\bP(\rho_\theta,\{M_i\}) = \frac{4(1-2m)^2\sin(2\theta)^2}{1-(1-2m)\cos(2\theta)^2}, \end{equation}
    and 
    \begin{equation}
    \begin{split}
    F^\bP(\rho_\theta\otimes\rho_\theta,\{M_i\}) &= \max\Big\{\frac{16(1-2m)^2\cos(\theta)^6\sin(\theta)^2}{(1-m)m+(1-2m)^2(\cos\theta)^4(1-(\cos\theta)^4)},\\ &\qquad\qquad \frac{16(1-2m)^2\sin(\theta)^2\cos(\theta)^6}{(1-m)m+(1-2m)^2(\sin\theta)^4(1-(\sin\theta)^4)}\Big\}. 
    \end{split}
    \end{equation}
    Consider the limit $m \rightarrow 0$, we have $  F^\bP(\rho_\theta\otimes\rho_\theta,\{M_i\}) < 2 F^\bP(\rho_\theta,\{M_i\})$ (which is expected because $\{M_i\}$ when $m=0$ is QFI-attainable for $\rho_\theta$ and the QFI is additive). On the other hand, one can immediately find cases where $  F^\bP(\rho_\theta\otimes\rho_\theta,\{M_i\}) > 2 F^\bP(\rho_\theta,\{M_i\})$, e.g., when $m = 0.1$ and $\theta = \pi/8$. For any fixed $m > 0$, there is a threshold of $\theta$ above which the sign of $F^\bP(\rho_\theta\otimes\rho_\theta,\{M_i\}) - 2 F^\bP(\rho_\theta,\{M_i\})$ changes from positive to negative. 

    Finally, we can consider multiple states under multiple measurements. We have, by definition, $F^\bP(\rho_\theta\otimes\sigma_\theta,\{M_i\}\otimes\{M_i'\}) \geq F^\bP(\rho_\theta,\{M_i\})+F^\bP(\sigma_\theta,\{M_i'\})$ and the inequality can be strict  (see the convergence to the QFI in the asymptotic limit in \secref{sec:asymptotic}).  
    
\end{enumerate}

\section{Attainability of the QPFI}
\label{app:singular}

Here we prove several results that are related to the attainability of the QPFI and the QUPFI. 

We first show the existence of the optimal controls (or unitaries) for generic noisy measurement that has non-zero noise in all subspaces. 
\begin{lemma}
\label{lemma:noisy}
For arbitrary quantum states $\rho_\theta$ and noisy measurements $\{M_i\}$
 such that $\min_i \lambda_{\min} (M_i) > 0$, the supremums in \eqref{eq:FI-A} and \eqref{eq:FI-U} are attainable. Here $\lambda_{\min}(\cdot)$ represents the minimum eigenvalue of an operator. 
\end{lemma}
\begin{proof}
By definition, there exists a sequence of quantum channels $(\mE_1,\cdots,\mE_n,\cdots)$ such that 
\begin{equation}
\label{eq:eta-1}
    F(\mE_n(\rho_\theta),\{M_i\}) \geq F^{\bP}(\rho_\theta,\{M_i\}) - \eta_n,
\end{equation}
where $\lim_{n\rightarrow \infty} \eta_n = 0$. Since the set of quantum channels is bounded and closed, there exists a limiting point of $(\mE_1,\cdots,\mE_n,\cdots)$ that we denote by $\mE$. Without loss of generality, we assume the sequence converges and $\mE = \lim_{n\rightarrow \infty} \mE_n$. And 
\begin{align}
    \lim_{n\rightarrow \infty} F(\mE_n(\rho_\theta),\{M_i\}) 
    &= 
    \lim_{n\rightarrow \infty}  \sum_{i} \frac{(\trace(\mE_n(\partial_\theta \rho_\theta) M_i))^2}{\trace(\mE_n(\rho_\theta) M_i)} \\
    &= 
    \sum_{i} \frac{\lim_{n\rightarrow \infty} (\trace(\mE_n(\partial_\theta \rho_\theta) M_i))^2}{\lim_{n\rightarrow \infty} \trace(\mE_n(\rho_\theta) M_i)} \\
    &= \sum_{i} \frac{(\trace(\mE(\partial_\theta \rho_\theta) M_i))^2}{\trace(\mE(\rho_\theta) M_i)} = F(\mE(\rho_\theta),\{M_i\}),
\end{align}
where we use $\trace(\mE_n(\rho_\theta) M_i) > \min_i \lambda_{\min} (M_i) > 0$ for all $n$. Then we must have $F(\mE(\rho_\theta),\{M_i\}) \geq F^\bP(\rho_\theta,\{M_i\})$ using \eqref{eq:eta-1}. Since  $F(\mE(\rho_\theta),\{M_i\}) \leq F^\bP(\rho_\theta,\{M_i\})$ by definition, we have 
\begin{equation}
    F(\mE(\rho_\theta),\{M_i\}) = F^\bP(\rho_\theta,\{M_i\}),
\end{equation}
proving the existence of the optimal channels. The existence of the optimal unitaries can also be proven analogously. 
\end{proof}

As a corollary of \lemmaref{lemma:noisy}, we show that for any measurement whose QPFI (or QUPFI) may not be attainable, there always exists a measurement in its neighborhood such that its QPFI (or QUPFI) is attainable and close to that of the original measurement. 

\begin{lemma}
\label{lemma:dense}
For any quantum state $\rho_\theta$, quantum measurement $\{M_i\}$ and $\eta > 0$, there always exists $\{M_i^{(\epsilon)}\}$ and a constant $c > 0$ such that the corresponding QPFI and the QUPFI are attainable, and when $\epsilon < c$, 
\begin{gather}
\label{eq:eta-2}
    F^\bP(\rho_\theta,\{M_i^{(\epsilon)}\})  \geq F^\bP(\rho_\theta,\{M_i\}) - \eta,\\
\label{eq:eta-3}
    F^\bU(\rho_\theta,\{M_i^{(\epsilon)}\})  \geq F^\bU(\rho_\theta,\{M_i\}) - \eta.
\end{gather} 
\end{lemma}

\begin{proof}
Assume $F^\bP(\rho_\theta,\{M_i\}) - \eta \geq F^\bU(\rho_\theta,\{M_i\}) - \eta > 0$. 
By definition, we can pick $\mE$ and $U$ such that 
\begin{gather}
    F(\mE(\rho_\theta),\{M_i\})  \geq F^\bP(\rho_\theta,\{M_i\}) - \eta/2,\\
    F(U\rho_\theta U^\dagger,\{M_i\})  \geq F^\bU(\rho_\theta,\{M_i\}) - \eta/2.
\end{gather}
Let 
\begin{equation}
    c' = \min\Big\{\min_{i:\trace(\mE(\rho_\theta) M_i) \neq 0} \frac{\trace(\mE(\rho_\theta) M_i)}{\trace(M_i)},\min_{i:\trace(U(\rho_\theta)U^\dagger M_i) \neq 0} \frac{\trace(U(\rho_\theta)U^\dagger M_i)}{\trace(M_i)}\Big\}, 
\end{equation}
define (assuming $d = \dim(\mH_{S'})$)
\begin{equation}
    M_i^{(\epsilon)} = (1-\epsilon) M_i + \epsilon \trace(M_i) \frac{\id}{d}, 
\end{equation}
and assume  $\epsilon$ is small enough such that 
\begin{equation}
    1-\epsilon-\frac{\epsilon}{dc'}(1-\epsilon) > \frac{F^\bU(\rho_\theta,\{M_i\}) - \eta}{F^\bU(\rho_\theta,\{M_i\}) - \eta/2}. 
\end{equation}
Note that a similar construction of $\rho^{(\epsilon)}$, instead of $M_i^{(\epsilon)}$, was used in \cite{vsafranek2018simple} to remove singularity of the QFI. 
Using \lemmaref{lemma:noisy}, it is clear that the QPFI and the QUPFI for $\{M_i^{(\epsilon)}\}$ are attainable. 
Furthermore, we have 
\begin{align}
     F^\bP(\rho_\theta,\{M_i^{(\epsilon)}\})
     &\geq \sum_{i:\trace(\mE(\rho_\theta) M_i^{(\epsilon)})\neq 0} \frac{(\trace(\mE(\partial_\theta \rho_\theta) M_i^{(\epsilon)}))^2}{\trace(\mE(\rho_\theta) M_i^{(\epsilon)})}  \geq 
     \sum_{i:\trace(\mE(\rho_\theta) M_i)\neq 0} \frac{(1-\epsilon)^2(\trace(\mE(\partial_\theta \rho_\theta) M_i))^2}{(1-\epsilon)\trace(\mE(\rho_\theta) M_i) +  \trace(M_i) \epsilon/r} \allowdisplaybreaks\\
     &\geq  \sum_{i:\trace(\mE(\rho_\theta) M_i)\neq 0} \frac{(\trace(\mE(\partial_\theta \rho_\theta) M_i))^2}{\trace(\mE(\rho_\theta) M_i)} \left(1-\epsilon-\frac{\epsilon}{dc'}(1-\epsilon)\right) \allowdisplaybreaks\\
     &\geq  F(\mE(\rho_\theta),\{M_i\}) \frac{F^\bU(\rho_\theta,\{M_i\}) - \eta}{F^\bU(\rho_\theta,\{M_i\}) - \eta/2} \geq F^\bP(\rho_\theta,\{M_i\}) - \eta, 
\end{align}
proving \eqref{eq:eta-2}. \eqref{eq:eta-3} is also true, similarly.
When $F^\bP(\rho_\theta,\{M_i\}) - \eta \leq 0$ or  $F^\bU(\rho_\theta,\{M_i\}) - \eta \leq 0$, the results also follow trivially. 
\end{proof}

Finally, we are ready to provide a proof of \thmref{thm:attainability}, which shows a way to calculate the QPFI by considering the limit of the QPFI for a set of generic noisy measurements in its neighborhood. (Note that the theorem stated below also holds for the QUPFI.)

\thmattain*

\begin{proof}

For any $\eta > 0$, following \lemmaref{lemma:dense}, we have 
\begin{equation}
    F^\bP(\rho_\theta,\{M_i\}) \leq F^\bP(\rho_\theta,\{M_i^{(\epsilon)}\}) + \eta, 
\end{equation}
when $\epsilon$ is small enough, where $M_i^{(\epsilon)} = (1-\epsilon) M_i + \epsilon \trace(M_i)\frac{\id}{d}$. Take the limit $\epsilon\rightarrow 0^+$ on both sides, we have $    F^\bP(\rho_\theta,\{M_i\}) \leq \liminf_{\epsilon \rightarrow 0^+} F^\bP(\rho_\theta,\{M_i^{(\epsilon)}\}) + \eta$ for any $\eta > 0$, implying 
\begin{equation}
\label{eq:eta-4}
    F^\bP(\rho_\theta,\{M_i\}) \leq \liminf_{\epsilon \rightarrow 0^+} F^\bP(\rho_\theta,\{M_i^{(\epsilon)}\}). 
\end{equation} 

On the other hand, 
consider a quantum channel for $0 < \epsilon < 1$, 
\begin{equation}
    \mE_{\epsilon}(\sigma) = (1-\epsilon)\sigma + \epsilon \trace(\sigma) \frac{\id}{d}, 
\end{equation}
Then we have $\trace(\mE_{\epsilon}(\sigma)M_i) = \trace(\sigma \mE_{\epsilon}^\dagger(M_i)) = \trace(\sigma M_i^{(\epsilon)})$ for any $\sigma$. By definition, we have 
\begin{equation}
    F^\bP(\rho_\theta,\{M_i\}) = \sup_\mE F(\mE(\rho_\theta),\{M_i\}) \geq \sup_\mE F(\mE_\epsilon(\mE(\rho_\theta)),\{M_i\}) = F^\bP(\rho_\theta,\{\mE_\epsilon^\dagger(M_i)\}) = F^\bP(\rho_\theta,\{M_i^{(\epsilon)}\}). 
\end{equation}
Take the limit $\epsilon\rightarrow 0^+$ on both sides, we have 
\begin{equation}
\label{eq:eta-5}
F^\bP(\rho_\theta,\{M_i\}) \geq \limsup_{\epsilon\rightarrow 0^+}F^\bP(\rho_\theta,\{M_i^{(\epsilon)}\}). 
\end{equation}
The theorem then follows from \eqref{eq:eta-4} and \eqref{eq:eta-5}. 

\end{proof}

\section{Global optimization algorithm for biconvex optimization problems}
\label{app:biconvex}

In \secref{sec:opt}, we showed that the QPFI can be obtained from the following biconvex optimization problem (\eqref{eq:opt-x-Omega}): 
\begin{align}
    F^\bP(\rho_\theta,\{M_i\})^{-1} = \inf_{(\vx,\Omega)} \;& \trace((X_2\otimes\rho_\theta^T)\Omega),\\
    \text{s.t.~}\; & \Omega\geq 0, \nonumber\\
    & \trace_{S'}(\Omega) = \id_S,\; \trace((X\otimes\rho_\theta^T)\Omega) = 0, \; \trace((X\otimes\partial_\theta\rho_\theta^T)\Omega) = 1.  \nonumber
\end{align}
The constraints on $\Omega$ guarantee any feasible $\Omega$ is contained in a convex compact set $R_2$ (the absolute value of each entry of $\Omega$ should not be larger than $\dim(\mH_{S})$). We could also set a convex compact region $R_1$ on $\vx$, so that the following optimization problem generates the same optimal value as \eqref{eq:opt-x-Omega}. 
\begin{align}
\label{eq:opt-x-Omega-app}
    \min_{(\vx,\Omega)} \;& \trace((X_2\otimes\rho_\theta^T)\Omega),\\
    \text{s.t.~}\; & \Omega\geq 0, \nonumber\\
    & \trace_{S'}(\Omega) = \id_S,\; \trace((X\otimes\rho_\theta^T)\Omega) = 0, \; \trace((X\otimes\partial_\theta\rho_\theta^T)\Omega) = 1, \nonumber\\
    & \vx \in R_1, \;\Omega \in R_2. \nonumber
\end{align}
As discussed in \secref{sec:opt}, this is possible when the size of $R_1$ is suffciently large, in normal cases when the infimum in \eqref{eq:opt-x-Omega} is attainable. Otherwise, the optimal value of \eqref{eq:opt-x-Omega-app} can still approach that of \eqref{eq:opt-x-Omega} for sufficiently large size of $R_1$. 

Here we describe the global optimization algorithm~\cite{floudas2013deterministic} for \eqref{eq:opt-x-Omega-app} that is guaranteed to converge to the global optimum of \eqref{eq:opt-x-Omega-app} in finite steps. One may seek \cite{gorski2007biconvex} for a general survey on algorithms from biconvex optimization. 

We first rewrite \eqref{eq:opt-x-Omega-app} as 
\begin{align}
    \min_{(\vx,\Omega)\in R_1\times R_2} \;& f(\vx,\Omega),\\
    \text{s.t.~~~~~}\; & \Omega\geq 0, \quad \forall i,\,h_i(\vx,\Omega) = 0,\nonumber
\end{align}
where $f(\vx,\Omega)$ is the biconvex target function and $h_i(\vx,\Omega)$ are bi-affine functions. The global optimization algorithm finds the global optimum of \eqref{eq:opt-x-Omega-app} by solving a set of primal problems and relaxed dual problems which generate upper and lower bounds on the optimum respectively. The upper and lower bounds converge to the global optimum up to a small error in finite steps. The algorithm is described as follows. 
\begin{enumerate}[Step~1:]
    \item \textbf{Initialization}. 
    
    Define initial upper and lower bounds $(f^U,f^L)$ on the global optimum, where $f^U$ and $-f^L$ can be chosen as two very large numbers. Set the counter $K = 1$. Set a convergence tolerance parameter $\varepsilon$. Choose a starting point $\vx^1$. Define three empty sets $\frakK^{feas}$ (set of feasible problems), $\frakK^{infeas}$ (set of infeasible problems), $\frakS$ (set of candidates of lower bound).
    \item \textbf{Primal problem}. 
    \begin{enumerate}[(1),leftmargin=*]
        \item Consider the primal problem for $\vx = \vx^K$ if it is feasible (that is, if there exists some $\Omega \in R_2$ that satisfies the constraints): 
    \begin{align}
    \label{eq:primal}
    P(\vx^K) = \min_{\Omega \in R_2} \;& f(\vx^K,\Omega),\\
    \text{s.t.~}\; & \Omega\geq 0, \quad \forall i,\,h_i(\vx^K,\Omega) = 0.\nonumber
    \end{align}
    The strong duality theorem~\cite{boyd2004convex} indicates that $P(\vx^K)$ can be solved through
    \begin{equation}
    \label{eq:primal-dual}
        P(\vx^K) = \max_{\vy, Z \geq 0}\min_{\Omega \in R_2}  L(\vx^K,\Omega,\vy,Z),
    \end{equation}
    where the Lagrange function 
    \begin{equation}
    L(\vx,\Omega,\vy,Z) := f(\vx,\Omega) + \sum_i y_i h_i(\vx,\Omega) - \trace(\Omega Z),    
    \end{equation}
    $Z$ is a semidefinite positive matrix acting on $\mH_{S'} \otimes \mH_S$ and $\vy$ is a vector of real numbers. 
    
    Solve \eqref{eq:primal-dual} and store the optimal values $(\Omega^K,\vy^K,Z^K)$. Set $f^U = \min\{f^U,P(\vx^K)\}$ and $\frakK^{feas} = \frakK^{feas} \cup \{K\}$. 
    \item 
    If \eqref{eq:primal} is infeasible, solve the relaxed primal problem for $\vx = \vx^K$ instead: 
    \begin{align}
    \label{eq:relaxed-primal}
    \delta(\vx^K) = \min_{\Omega \in R_2,\alpha \geq 0} \;& \alpha,\\
    \text{s.t.~~~~~}\; & \Omega  + \alpha \id_{S'S} \geq 0, \quad \forall i,\,h_i(\vx^K,\Omega) = 0.\nonumber
    \end{align}
    The strong duality theorem implies 
    \begin{align}
    \label{eq:relaxed-primal-dual}
    \delta(\vx^K) &= \max_{\vy, Z \geq 0}\min_{\Omega \in R_2,\alpha \geq 0} \alpha + \sum_i y_i h_i(\vx,\Omega) - \trace((\Omega+\alpha\id_{S'S}) Z)\\
    &= \max_{\vy, Z \geq 0, \trace(Z) \leq 1}\min_{\Omega \in R_2} L_1(\vx^K,\Omega,\vy,Z),  \nonumber
    \end{align}
    where the Lagrange function $L_1(\vx,\Omega,\vy,Z) := \sum_i y_i h_i(\vx,\Omega) - \trace(\Omega Z)$. 
    
    Solve \eqref{eq:relaxed-primal-dual} and store the optimal values $(\Omega^K,\vy^K,Z^K)$. Let $\frakK^{infeas} = \frakK^{infeas} \cup \{K\}$. 
    \end{enumerate}
    
    \item \textbf{Determine the current region of $\vx$}. 
    
    Suppose $\Omega$ is parameterized by a vector of real numbers $ \Omega_i$. Since $\Omega$ is contained in a compact set, $\Omega_i$ has upper and lower bounds that we denote by $\Omega_i^U$ and $\Omega_i^L$. Consider the partial derivatives of the Lagrange functions defined by $g^k_i(\vx) := \frac{\partial}{\partial\Omega_i} L(\vx,\Omega,\vy^k,Z^k)|_{\Omega^k}$ for $k \in \frakK^{feas}$ and $g^k_i(\vx) := \frac{\partial}{\partial\Omega_i} L_1(\vx,\Omega,\vy^k,Z^k)|_{\Omega^k}$ for $k \in \frakK^{infeas}$. Define the set of indices for connected variables $I_k:= \{i|g^k_i(\vx)\text{ is a nontrivial function of }\vx\} = \{i|g^k_i(\vx) = 0,\forall \vx\}$ (the last equality follows from the KKT conditions~\cite{boyd2004convex}) and $\Omega_i$ is called a connected variable of the Lagrange functions if $i \in I_k$.  
    We can also define the linearized Lagrange functions $L(\vx,\Omega,\vy^k,Z^k)|_{\Omega^k}^{lin} := L(\vx,\Omega^k,\vy^k,Z^k) + \sum_{i\in I_k} g_i^k(\vx)(\Omega_i-\Omega_i^k)$ and $L_1(\vx,\Omega,\vy^k,Z^k)|_{\Omega^k}^{lin} := L_1(\vx,\Omega^k,\vy^k,Z^k) + \sum_{i\in I_k} g_i^k(\vx)(\Omega_i-\Omega_i^k)$. The linearized functions $L(\vx,\Omega,\vy^k,Z^k)|_{\Omega^k}^{lin}$ and $L_1(\vx,\Omega,\vy^k,Z^k)|_{\Omega^k}^{lin}$ are functions of the connected variables only and independent of $\Omega_i$ if $i \notin I_k$. 
    
    Let $\frakB^k := \cross_{i\in I_k}\{\Omega_i^{L},\Omega_i^{U}\}$ be the set of combinations of upper and lower bounds on $\Omega_i$ for all $i \in I_k$. We abuse the notation a bit and use $\Omega \in \frakB^k$ to denote the case where the part of connected variables $\Omega_{I_k}$ in $\Omega$ is contained in $\frakB^k$ and the other part is arbitrary. We will see that the other part is irrelevant in our calculations and can be ignored. In this sense, there are in total $2^{|I_k|}$ number of $\Omega \in \frakB^k$ which is finite. We also define $R(k,\Omega)$ to be a region of $\vx$ as a function of $\Omega \in \frakB^k$ defined by 
    \begin{equation}
    R(k,\Omega) := \{\vx|\forall i \in I_k \;g_i^k(\vx) \leq_{\Omega_i} 0\},
    \end{equation}
    where ``$\leq_{\Omega_i}$'' represents ``$\leq$'' if $\Omega_i = \Omega_i^U$, and ``$\geq$'' if $\Omega_i = \Omega_i^L$.
    
    Let $\frakB^{(k,K)} = \{\Omega\in\frakB^k|\vx^K \in R(k,\Omega)\}$. The relaxed dual problem in the next step will be solved in the region of $\vx$ that is contained in $\bigcap_{k=1}^{K-1} \bigcap_{\Omega \in \frakB^{(k,K)}}R(k,\Omega)$.

    \item \textbf{Relaxed dual problem}. 
    
    Determine the set of indices for connected variables $I_K$. Note that for any $k$, $L(\vx,\Omega,\vy^k,Z^k)|^{lin}_{\Omega^k}$ is a function of the connected variables only and is fixed if the connected variables $\Omega_{I_k}$ of $\Omega$ is fixed. Therefore we will also write $L(\vx,\Omega,\vy^k,Z^k)|^{lin}_{\Omega^k} = L(\vx,\Omega_{I_k},\vy^k,Z^k)|^{lin}_{\Omega^k}$. 
    
    For each $\Omega_\star \in \frakB^K = \cross_{i \in I_K}\{\Omega_i^{L},\Omega_i^{U}\}$ (there are $2^{|I_K|}$ different $\Omega_\star$ in total), solve the following relaxed dual problem: 
    \begin{align}
    \label{eq:relaxed-dual}
    \min_{\vx \in R_1,\mu} &\;  \mu\\
    \text{s.t.~~} &\,  
    \begin{rcases}
    & \mu \geq L(\vx,\Omega,\vy^k,Z^k)|^{lin}_{\Omega^k},\\
    & \vx \in R(k,\Omega),
    \end{rcases}
    \quad \forall \, \Omega \in \frakB^{(k,K)}, \; 1 \leq k \leq K-1, \;  k \in \frakK^{feas}, 
    \nonumber\\
    &\,  
    \begin{rcases}
    & 0 \geq L_1(\vx,\Omega,\vy^k,Z^k)|^{lin}_{\Omega^k},\\
    & \vx \in R(k,\Omega),
    \end{rcases}
    \quad \forall \, \Omega \in \frakB^{(k,K)}, \; 1 \leq k \leq K-1, \;  k \in \frakK^{infeas}, 
    \nonumber\\
    &\,  
    \begin{rcases}
    & \mu \geq L(\vx,\Omega_\star,\vy^K,Z^K)|^{lin}_{\Omega^K},\\
    & \vx \in R(k,\Omega_\star),
    \end{rcases}
    \quad K \in \frakK^{feas}, 
    \nonumber\\
    &\,  
    \begin{rcases}
    & 0 \geq L_1(\vx,\Omega_\star,\vy^K,Z^K)|^{lin}_{\Omega^K},\\
    & \vx \in R(k,\Omega_\star),
    \end{rcases}
    \quad K \in \frakK^{infeas}, \nonumber
    \end{align}
    For each $\Omega_\star$, store the solution $(\mu_\star,\vx_\star)$ of \eqref{eq:relaxed-dual}  in $\frakS$. 
    \item \textbf{Select a new lower bound and determine $\vx^{K+1}$}. 
    
    From the set $\frakS$, select the minimum $\mu^{\min}$ and the corresponding $\vx^{\min}$. Set $f^L = \mu^{\min}$ and $\vx^{K+1} = \vx^{\min}$. Delete $(\mu^{\min},\vx^{\min})$ from the set of candidates of lower bound $\frakS$.  
    
    \item \textbf{Check for convergence}. 
    
    Check if $f^L > f^U - \varepsilon$, if yes, STOP; otherwise, set $K = K+1$ and return to Step 2. 
\end{enumerate}

The global optimization algorithm described above works in a branch-and-bound way where $\vx$ is partitioned into different regions and different candidates of lower bounds of the global optimum are explored in each iteration. The subproblems that are solved in each iteration are semidefinite programs (\eqref{eq:primal-dual} and \eqref{eq:relaxed-primal-dual}) and quadratically constrained quadratic programs (\eqref{eq:relaxed-dual}) which can be solved efficiently (for a moderate system dimension) using algorithms for convex optimization~\cite{boyd2004convex}. The running time of the entire algorithm depends largely on the number of subproblems that are solved in each iteration which is exponential in the number of connected variables. Methods that can reduce this complexity were discussed in~\cite{floudas2013deterministic}.

\section{Binary measurements on pure states}

\subsection{Measurement on a qubit}
\label{app:binary-qubit}

Here consider a binary measurement on a single qubit where $X = x_1 M_1+ x_2 M_2$, $M_1 = M$ and $M_2 = \id - M$. Without loss of generality, we assume
\begin{equation}
M = \begin{pmatrix}
m_1 & 0 \\
0 & m_2 \\
\end{pmatrix},\quad m_1,m_2\in[0,1] \text{~~and~~} m_1 > m_2.
\end{equation}
Our goal is to find $(\vx,U)$ such that the two necessary conditions in \lemmaref{lemma:necessary} are satisfied. 

Let $\ket{\phi} = \sqrt{p}\ket{1} + \sqrt{1-p}\ket{2}$ where $0 \leq p \leq 1$ (any additional phases in the amplitudes of $\ket{\phi}$ do not change the results). Condition (1) translates into
\begin{gather}
\label{eq:system-1}
\begin{pmatrix}
\sqrt{p} & \sqrt{1 - p}
\end{pmatrix}
\begin{pmatrix}
x_1 m_1 + x_2 (1 - m_1) & 0 \\
0 & x_1 m_2 + x_2 (1 - m_2)
\end{pmatrix}
\begin{pmatrix}
\sqrt{p}\\
\sqrt{1 - p}
\end{pmatrix} = 0, 
\\
\label{eq:system-2}
\begin{pmatrix}
\sqrt{p} & \sqrt{1 - p}
\end{pmatrix}
\begin{pmatrix}
(x_1 m_1 + x_2 (1 - m_1))^2 & 0 \\
0 & (x_1 m_2 + x_2 (1 - m_2))^2
\end{pmatrix}
\begin{pmatrix}
\sqrt{p}\\
\sqrt{1 - p}
\end{pmatrix} = \frac{1}{4\frakn }, 
\end{gather}
and Condition~(2) is trivially true when $p = 0,1$, otherwise translates into: 
\begin{equation}
\frac{\braket{\phi|X^2|\phi}}{\braket{\phi|X_2|\phi}} = \frac{(x_1 m_1 + x_2 (1 - m_1))^2}{x_1^2 m_1 + x_2^2 (1 - m_1)} = \frac{(x_1 m_2 + x_2 (1 - m_2))^2}{x_1^2 m_2 + x_2^2 (1 - m_2)} = \gamma(\{M_i\}) \leq 1,
\end{equation}
where we use 
\begin{equation}
\label{eq:system-3}
F^{\bU}(\psi_\theta,\{M_i\})^{-1}
= \braket{\phi|X_2|\phi} = \frac{\braket{\phi|X_2|\phi}}{\braket{\phi|X^2|\phi}} \cdot \braket{\phi|X^2|\phi} = (\gamma(\{M_i\}) J(\psi_\theta))^{-1},
\end{equation}
and $J(\psi_\theta) = 4\frakn$. 

Solving the system of equations (\eqref{eq:system-1}--\eqref{eq:system-3}), we find the following results. 
\begin{enumerate}[(1)]
    \item When $1 > m_1 > m_2 > 0$, according to \lemmaref{lemma:noisy}, the QPFI is attainable. Moreover, the only solution satisfying the necessary conditions (\eqref{eq:system-1}--\eqref{eq:system-3}) is 
\begin{equation}
\frac{x_2}{x_1} = -\sqrt{\frac{m_1m_2}{(1-m_1)(1-m_2)}},
\end{equation}
\begin{equation}
\label{eq:binary-1}
p = \frac{\sqrt{m_2(1-m_2)}}{\sqrt{m_1(1-m_1)}+\sqrt{m_2(1-m_2)}},
\end{equation}
which must be the optimal solution. It gives 
\begin{equation}
\label{eq:binary-0}
\gamma(\{M_i\}) = 1-\big(\sqrt{ m_1 m_2} +\sqrt{\left(1-m_{1}\right)\left(1-m_{2}\right)}\big)^{2}. 
\end{equation}

\item When $1 = m_1 > m_2 = 0$, we must have $X^2 = X_2$, implying $\gamma(\{M_i\}) = 1$. The QFI is achievable as long as a solution to $p x_1 + (1-p) x_2 = 0$ and $p x_1^2 + (1-p) x_2^2 = 1/(4\frakn)$ exists, which means that any $0 < p < 1$ is optimal. 

\item When $1 > m_1 > m_2 = 0$, the necessary conditions (\eqref{eq:system-1}--\eqref{eq:system-3}) have no solutions. Thus, this is a singular case where the QPFI is not attainable. And we have from \thmref{thm:attainability} that 
\begin{align}
    \gamma(\{M_i\}) &= \lim_{\epsilon\rightarrow 0^+}\gamma(\{M_i^{(\epsilon)}\}) \\&= \lim_{\epsilon\rightarrow 0^+} \left(1-\frac{\epsilon}{2}\right)m_1+\frac{\epsilon}{2} -\left(1-\frac{\epsilon}{2}\right)m_1\epsilon -2\sqrt{\frac{\epsilon}{2}\left(1-\frac{\epsilon}{2}\right) \left(1-\frac{\epsilon}{2}\right)m_1\left(1-\left(1-\frac{\epsilon}{2}\right)m_1\right)} = m_1. 
\end{align}

\end{enumerate}

\subsection{Measurement on a qudit}
\label{app:binary-qudit}

Now consider a $d$-dimensional system with $d>2$ and a binary measurement $M_1 = M$ and $M_2 = \id - M$ on it where  
\begin{equation}
M = \begin{pmatrix}
m_1 & & & \\
& m_2 & & \\
& & \ddots & \\
& & & m_d \\
\end{pmatrix},
\end{equation}
and $1 > m_1\geq \cdots \geq m_d > 0$. 

Let $\ket{\phi} = \sum_{i=1}^{d} \phi_i \ket{i}$. We now show that the support of $\ket{\phi}$: ${\rm supp}\{\ket{\phi}\} = \{i:\phi_i\neq 0\}$ must correspond to at most two different values of $m_i$ when $\ket{\phi}$ is optimal. We prove this by contradiction. Without loss of generality, assume $\abs{\phi_{1,2,3}} > 0$ and $m_1 > m_2 > m_3$. Condition~(2) implies that 
\begin{equation}
\frac{\braket{1|X^2|1}}{\braket{1|X_2|1}} = \frac{\braket{2|X^2|2}}{\braket{2|X_2|2}} = \frac{\braket{3|X^2|3}}{\braket{3|X_2|3}},
\end{equation}
\begin{equation}
\Rightarrow~~~ 
\frac{(x_1 m_1 + x_2 (1 - m_1))^2}{x_1^2 m_1 + x_2^2 (1 - m_1)} = \frac{(x_1 m_2 + x_2 (1 - m_2))^2}{x_1^2 m_2 + x_2^2 (1 - m_2)}
= \frac{(x_1 m_3 + x_2 (1 - m_3))^2}{x_1^2 m_3 + x_2^2 (1 - m_3)},
\end{equation}
\begin{equation}
\Rightarrow~~ 
-\sqrt{\frac{m_1m_2}{(1-m_1)(1-m_2)}} = -\sqrt{\frac{m_1m_3}{(1-m_1)(1-m_3)}} = -\sqrt{\frac{m_2m_3}{(1-m_2)(1-m_3)}}, 
\end{equation}
which contradicts $m_1 > m_2 > m_3$. Thus, we conclude that the support of $\ket{\phi}$ must correspond to at most two different values of $m_i$. 

Therefore we have
\begin{equation}
\gamma(\{M_i\}) = \max_{1 \leq k \leq l \leq d}\gamma_{kl}(\{M_i\}) 
= \max_{kl} 1 - \big(\sqrt{m_k m_l} + \sqrt{(1-m_k)(1-m_l)}\big)^2. 
\end{equation}
In fact, assume $m_1 \geq m_2 \geq \cdots \geq m_d$, we must have 
\begin{equation}
\gamma(\{M_i\}) = 
   1-\big(\sqrt{ m_1 m_d} +\sqrt{\left(1-m_{1}\right)\left(1-m_{d}\right)}\big)^{2}. 
\end{equation}
The reason is that when $m_k \geq m_l$, increasing $m_k$ or decreasing $m_l$ while the other element is fixed will only increases $\gamma_{kl}(\{M_i\})$. We see that by computing the derivative of $\gamma_{kl}(\{M_i\})$ with respect to $m_k$. We have 
\begin{equation}
\begin{split}
 \frac{\partial}{\partial m_k}\gamma_{kl}(\{M_i\}) = 1-2m_l - \frac{\sqrt{(1-m_l)m_l}}{\sqrt{(1-m_k)m_k}} (1-2m_k) \geq 0, 
\end{split} 
\end{equation}
when $m_k \geq m_l$ because $\frac{1-2m_l}{\sqrt{(1-m_l)m_l}} = \sqrt{\frac{1-m_l}{m_l}} - \sqrt{\frac{m_l}{1-m_l}} \geq \frac{1-2m_k}{\sqrt{(1-m_k)m_k}}$.

\section{Commuting-operator measurements on pure states}

We take one step further from binary measurements and consider the commuting-operator measurements where 
\begin{equation}
M_i = \begin{pmatrix}
m^{(i)}_1 &  & & \\
 & m^{(i)}_2 & & \\
 & & \ddots & \\
 & & & m^{(i)}_d
\end{pmatrix},
\end{equation}
and $\sum_i M_i = \id$. We also assume $m^{(i)}_j > 0$ for all $i,j$. 

\subsection{Proof of \texorpdfstring{\thmref{thm:2d}}{Theorem~6} }
\label{app:2d}

We first prove \thmref{thm:2d}: 

\thmtwod*

\begin{proof}

Assume $(\vx,\ket{\phi})$ satisfies Condition~(2) in \lemmaref{lemma:necessary}, where we write $x_i = y(x + a_i)$, $\vx \neq 0$ and the support of $\ket{\phi}$ contains $\ket{1}$ and $\ket{2}$.
Then according to Condition~(2), we must have 
\begin{equation}
\label{eq:xx-aa}
\frac{(x + \braket{a}_1 )^2}{x^2 + 2 \braket{a}_1 x + \braket{a^2}_1} = \frac{(x + \braket{a}_2 )^2}{x^2 + 2 \braket{a}_2 x + \braket{a^2}_2} = \frac{\braket{\phi|X^2|\phi}}{\braket{\phi|X_2|\phi}}, 
\end{equation}
where $\braket{a}_j = \sum_i a_i m_j^{(i)}$ and $\braket{a^2}_j = \sum_i a_i^2 m_j^{(i)} $. We also define $\braket{\Delta a^2}_j = \braket{a^2}_j - \braket{a}^2_j$. 

Assume $\braket{a}_1 > \braket{a}_2$, we have the following possible solutions of \eqref{eq:xx-aa}. 
\begin{enumerate}[(1)]
\item When $\braket{\Delta a^2}_1 = \braket{\Delta a^2}_2$, 
\begin{equation}
\label{eq:x-gamma-a-0}
    x = -\frac{1}{2}(\braket{a}_1+\braket{a}_2),\text{~~~and~~~} \frac{\braket{\phi|X_2|\phi}}{\braket{\phi|X^2|\phi}} = 1 + \left(\frac{\sqrt{\braket{\Delta a^2}_1} + \sqrt{\braket{\Delta a^2}_2 }}{\braket{a}_1 - \braket{a}_2}\right)^2. 
\end{equation}
\item When $\braket{\Delta a^2}_1 - \braket{\Delta a^2}_2 \neq 0$, we have either 
\begin{equation}
\begin{aligned}
\label{eq:x-a-1}
&x = \frac{-\braket{a}_2 \braket{\Delta a^2}_1 + \braket{a}_1 \braket{\Delta a^2}_2 - (\braket{a}_1 - \braket{a}_2) \sqrt{\braket{\Delta a^2}_1 \braket{\Delta a^2}_2}}{\braket{\Delta a^2}_1 - \braket{\Delta a^2}_2}
\end{aligned}
\end{equation}
and 
\begin{equation}
\label{eq:gamma-a-1}
\frac{\braket{\phi|X_2|\phi}}{\braket{\phi|X^2|\phi}} = 
1 + \left(\frac{\sqrt{\braket{\Delta a^2}_1} + \sqrt{\braket{\Delta a^2}_2 }}{\braket{a}_1 - \braket{a}_2}\right)^2,
\end{equation}
or 
\begin{equation}
\begin{aligned}
\label{eq:x-a-2}
&x = \frac{-\braket{a}_2 \braket{\Delta a^2}_1 + \braket{a}_1 \braket{\Delta a^2}_2 + (\braket{a}_1 - \braket{a}_2) \sqrt{\braket{\Delta a^2}_1 \braket{\Delta a^2}_2}}{\braket{\Delta a^2}_1 - \braket{\Delta a^2}_2}
\end{aligned}
\end{equation}
and 
\begin{equation}
\label{eq:gamma-a-2}
\frac{\braket{\phi|X_2|\phi}}{\braket{\phi|X^2|\phi}} = 
1 + \left(\frac{\sqrt{\braket{\Delta a^2}_1} - \sqrt{\braket{\Delta a^2}_2 }}{\braket{a}_1 - \braket{a}_2}\right)^2.
\end{equation}
\end{enumerate}

Next we show that there always is an optimal solution such that its support contains only two elements. 
Without loss of generality, assume $\ket{\phi^*} = \sum_{i=1}^{d} \phi^*_i \ket{i}$ is an optimal solution (which is guaranteed to exist thanks to \lemmaref{lemma:noisy}). The corresponding error vector $\vx^*$ is written as $x^*_i = y^*(x^* + a^*_i)$. We have from Condition~(1) in \lemmaref{lemma:necessary} that $\sum_{i=1}^{d} |\phi^*_i|^2 y^*(x^* + \braket{a^*}_i) = 0$ and $\sum_{i=1}^{d} |\phi^*_i|^2 (y^*(x^* + \braket{a^*}_i))^2 = 1/(4\frakn)$. 
Clearly, we must have some $i \neq j$, such that $\abs{\phi^*_{i,j}} > 0$, $\braket{a^*}_{i} \neq \braket{a^*}_{j}$ and $(x^* + \braket{a^*}_{i})(x^* + \braket{a^*}_{j}) < 0$. Without loss of generality, we assume $i=1$, $j=2$ and $\braket{a^*}_{1} > \braket{a^*}_{2}$. 
Then $(\vx^*,\ket{\phi^*})$ must satisfy either \eqref{eq:x-gamma-a-0} or \eqref{eq:x-a-1} and \eqref{eq:gamma-a-1} from the previous discussion. (Note that if \eqref{eq:x-a-2} and \eqref{eq:gamma-a-2} cannot be true because $(x^* + \braket{a^*}_{i})(x^* + \braket{a^*}_{j}) < 0$.) 

We then assert that $\ket{\phi^{**}} = \sqrt{p} \ket{1} + \sqrt{1 - p} \ket{2}$ ($0 \leq p \leq 1$) is also an optimal solution, when $p$ satisfies
\begin{equation}
\begin{pmatrix}
\sqrt{p} & \sqrt{1 - p}
\end{pmatrix}
\begin{pmatrix}
x^* + \braket{a^*}_1 & 0 \\
0 & x^* + \braket{a^*}_2
\end{pmatrix}
\begin{pmatrix}
\sqrt{p}\\
\sqrt{1 - p}
\end{pmatrix} = 0. 
\end{equation}
Using \eqref{eq:x-gamma-a-0}--\eqref{eq:gamma-a-1}, it is easy to see that 
\begin{equation}
p = \frac{\sqrt{\braket{\Delta (a^*)^2}_2}}{\sqrt{\braket{\Delta (a^*)^2}_2} + \sqrt{\braket{\Delta (a^*)^2}_1}}. 
\end{equation}
We take $x^{**}_i = y^{**}(x^* + a^*_i)$, where $y^{**}$ is solved from 
\begin{equation}
\begin{pmatrix}
\sqrt{p} & \sqrt{1 - p}
\end{pmatrix}
\begin{pmatrix}
(x^* + \braket{a^*}_1)^2 & 0\\
0 & (x^* + \braket{a^*}_2)^2
\end{pmatrix}
\begin{pmatrix}
\sqrt{p}\\
\sqrt{1 - p}
\end{pmatrix} = \frac{1}{4\frakn  (y^{**})^2}. 
\end{equation}
Both equations are derived from Condition~(1). Now we have a new solution $(\vx^{**},\ket{\phi^{**}})$ such that $y^{**}$ and $\ket{\phi^{**}}$ are solved by the equations above. Note that we still let $a_i^{**} = a_i^*$ and $x^{**} = x^*$. The new solution have the same FI as the original, i.e., $\braket{\phi|X_2|\phi}$ does not change, because $\braket{\phi|X_2|\phi}/\braket{\phi|X^2|\phi}$ is independent of $y$ (due to \eqref{eq:x-gamma-a-0} and \eqref{eq:gamma-a-1}) and $\braket{\phi|X^2|\phi} = 1/(4\frakn)$ is invariant. 
The new solution is thus supported on a two-dimensional subspace spanned by $\{\ket{1},\ket{2}\}$, proving \thmref{thm:2d}. 
\end{proof}

\subsection{Optimal solution for commuting-operator measurements}
\label{app:optimal-commuting}

Now we proceed to compute general $\gamma(\{M_i\})$ for commuting-operator measurements. First, consider the optimization for measurements restricted in a two-dimensional subspace spanned by $\ket{k},\ket{l}$ for some $k\neq l$, i.e., 
\begin{equation}
    (M_i)_{kl} = m_{k}^{(i)}\ket{k}\bra{k} + m_{l}^{(i)}\ket{l}\bra{l},
\end{equation}
and $\sum_i (M_i)_{kl} = \id_{\mathrm{span}\{\ket{k},\ket{l}\}}$. 

Let $(\vx^*,\ket{\phi^*})$ be an optimal solution when $\ket{\phi},\ket{\phi^\perp}$ are restricted in $\mathrm{span}\{\ket{k},\ket{l}\}$ and $\ket{\phi^*} = \sqrt{p_{kl}}\ket{k} + \sqrt{1-p_{kl}}\ket{l}$ (we also assume $\braket{a^*}_k > \braket{a^*}_l$). Using \eqref{eq:x}, we see that the optimal $a_i^*$ 
\begin{equation}
\label{eq:commuting-00}
\begin{split}
y^*(x^* +a_i^*) &= x_i^* = \frac{y^*}{\gamma_{kl}(\{M_i\})} \frac{\braket{\phi^*|M_i X^*|\phi^*}}{\braket{\phi^*|M_i|\phi^*}} \\
&= \frac{y^*}{\gamma_{kl}(\{M_i\})} \frac{p_{kl} m^{(i)}_k (x^* + \braket{a^*}_k) + (1-p_{kl}) m^{(i)}_l (x^* + \braket{a^*}_l)}{p_{kl} m^{(i)}_k + (1-p_{kl}) m^{(i)}_l}\\
&=  \frac{p_{kl} y^* (x^* + \braket{a^*}_k)}{\gamma_{kl}(\{M_i\})} \frac{m^{(i)}_k  - m^{(i)}_l}{p_{kl} m^{(i)}_k + (1-p_{kl}) m^{(i)}_l},\\
\end{split}
\end{equation}
where we use 
\begin{equation}
\label{eq:commuting-0}
    \braket{\phi^*|X^*|\phi^*} = p_{kl} y^*(x^* + \braket{a^*}_k) + (1-p_{kl}) y^*(x^* + \braket{a^*}_l) = 0, 
\end{equation}
in the last step. From \eqref{eq:commuting-00}, we have, 
\begin{gather}
(y^*)^2((x^*)^2 + 2 \braket{a^*}_k x^* + \braket{(a^*)^2}_k) = \left(\frac{p_{kl}y^* (x^* + \braket{a^*}_k )}{\gamma_{kl}(\{M_i\})}\right)^2 \sum_i \frac{(m^{(i)}_k  - m^{(i)}_l)^2 m^{(i)}_k}{(p_{kl} m^{(i)}_k + (1-p_{kl}) m^{(i)}_l)^2},
\allowdisplaybreaks \label{eq:commuting-3}\\
(y^*)^2((x^*)^2 + 2 \braket{a^*}_l x^* + \braket{(a^*)^2}_l) = \left(\frac{p_{kl}y^* (x^* + \braket{a^*}_k )}{\gamma_{kl}(\{M_i\})}\right)^2 \sum_i \frac{(m^{(i)}_k  - m^{(i)}_l)^2 m^{(i)}_l}{(p_{kl} m^{(i)}_k + (1-p_{kl}) m^{(i)}_l)^2}. \label{eq:commuting-4}
\end{gather}
According to Condition~(2), 
\begin{equation}
\label{eq:commuting-5}
\frac{(x^*)^2 + 2 \braket{a^*}_k x^* + \braket{(a^*)^2}_k}{(x^*+\braket{a^*}_k)^2} = \frac{(x^*)^2 + 2 \braket{a^*}_l x^* + \braket{(a^*)^2}_l}{(x^*+\braket{a^*}_l)^2}.
\end{equation}
From \eqref{eq:commuting-0}--\eqref{eq:commuting-5},
we have 
\begin{equation}
\label{eq:commuting-7}
p_{kl}^2 \left({\sum_i \frac{(m^{(i)}_k  - m^{(i)}_l)^2 m^{(i)}_k}{(p_{kl} m^{(i)}_k + (1-p_{kl}) m^{(i)}_l)^2}}\right) = (1-p_{kl})^2 \left({\sum_i \frac{(m^{(i)}_k  - m^{(i)}_l)^2 m^{(i)}_l}{(p_{kl} m^{(i)}_k + (1-p_{kl}) m^{(i)}_l)^2}}\right),  
\end{equation}
It will give us a unique solution to $p_{kl}$ because the left-hand side is a monotonically increasing function in $[0,\sum_i (m_k^{(i)}-m_l^{(i)})^2]$ of $p_{kl} \in [0,1]$ and the right-hand side is a monotonically decreasing function in $[0,\sum_i (m_k^{(i)}-m_l^{(i)})^2]$ of $p_{kl} \in [0,1]$. However, a simple analytical solution to $p_{kl}$ from \eqref{eq:commuting-7} might not exist because it is a root of a high degree polynomial. Then we have 
\begin{equation}
\label{eq:commuting-6}
    \gamma_{kl}(\{M_i\}) = \sum_i \frac{(\Re[\braket{\phi^*|M_i|\phi^{\perp*}}])^2}{\braket{\phi^*|M_i|\phi^*}} = \sum_i \frac{p^*_{kl} (1 - p^*_{kl}) (m_k^{(i)}-m_l^{(i)})^2 }{p^*_{kl} m_k^{(i)} + (1 - p^*_{kl}) m_l^{(i)}},
\end{equation}
where $p^*_{kl}$ is the unique solution to \eqref{eq:commuting-7}. 

Finally,
\begin{equation}
\gamma(\{M_i\}) = \max_{kl} \gamma_{kl}(\{M_i\}),
\end{equation}
using \thmref{thm:2d}. Note that although \eqref{eq:commuting-7} might only be solvable numerically in practice for a multiple-outcome measurement. Our solution for pure states and commuting-operator measurements still has a huge simplification compared to the original biconvex problem for general states and measurements. 

\subsection{Proof of \texorpdfstring{\thmref{thm:uppgamma}}{Theorem~7} }
\label{app:uppgamma}

Here we prove a simple upper bound on the normalized QPFI: 

\thmuppgamma*

\begin{proof}
From the discussions in \appref{app:2d} and \appref{app:optimal-commuting}, we have that 
\begin{equation}
    \gamma(\{M_i\}) = \max_{kl} \left(\max_{0\leq p\leq 1}\sum_i \frac{p (1 - p) (m_k^{(i)}-m_l^{(i)})^2 }{p m_k^{(i)} + (1 - p) m_l^{(i)}} \right)= \max_{kl} \sum_i \frac{p^*_{kl} (1 - p^*_{kl}) (m_k^{(i)}-m_l^{(i)})^2 }{p^*_{kl} m_k^{(i)} + (1 - p^*_{kl}) m_l^{(i)}}. 
\end{equation}
\eqref{eq:uppgamma} is then proven, noting that for any $p \in [0,1]$, 
\begin{align}
&\quad \sum_i \frac{p(1-p)({m^{(i)}_k}-{m^{(i)}_l})^2}{p {m^{(i)}_k} + (1-p) {m^{(i)}_l}} - 1 \\
&= \sum_i \frac{p(1-p)({m^{(i)}_k}-{m^{(i)}_l})^2-(p{m^{(i)}_k} + (1-p){m^{(i)}_l})^2}{p {m^{(i)}_k} + (1-p) {m^{(i)}_l}}  \\
&= \sum_i \frac{(1-2p)(p {m^{(i)}_k}^2-(1-p){m^{(i)}_l}^2) - 4p(1-p) {m^{(i)}_k}{m^{(i)}_l}}{p {m^{(i)}_k} + (1-p) {m^{(i)}_l}}\\
&= \sum_i \frac{(1-2p)\left( (p {m^{(i)}_k} + (1-p){m^{(i)}_l})({m^{(i)}_k} - {m^{(i)}_l}) - (2p - 1){m^{(i)}_k} {m^{(i)}_l}\right)- 4p(1-p) {m^{(i)}_k}{m^{(i)}_l}}{p {m^{(i)}_k} + (1-p) {m^{(i)}_l}}\\
&= - \left(\sum_i \frac{{m^{(i)}_k}{m^{(i)}_l}}{p {m^{(i)}_k} + (1-p) {m^{(i)}_l}}\right)\left(\sum_i {p {m^{(i)}_k} + (1-p) {m^{(i)}_l}}\right) \leq - \left(\sum_i \sqrt{{m^{(i)}_k} {m^{(i)}_l}}\right)^2,
\end{align}
where in the first equality we use $\sum_i p m_k^{(i)} + (1-p) m_l^{(i)} = 1$, in the last equality we multiply the expression by a factor of $1 = \sum_i p m_k^{(i)} + (1-p) m_l^{(i)}$, and the last inequality follows from Cauchy--Schwarz. 

Assume $(k,l)$ minimizes $\sum_i \sqrt{m_k^{(i)} m_l^{(i)}}$. Then the equality above holds when $\exists p$, such that for any $i,j$, 
\begin{equation}
    \frac{m^{(i)}_km^{(i)}_l}{(pm^{(i)}_k + (1-p)m^{(i)}_l)^2} = \frac{m^{(j)}_km^{(j)}_l}{(pm^{(j)}_k + (1-p)m^{(j)}_l)^2}. 
\end{equation}
When there are at most two different $i$ and $j$ (i.e., $r = 2$), such a $p$ always exists, and the upper bound is tight (which also follows directly from \eqref{eq:gamma-qudit-binary}). In general, when the set
\begin{equation}
    \left\{ \frac{m^{(i)}_k}{m^{(i)}_l},\, 1 \leq i \leq r\right\}
\end{equation}
contains at most two distinct elements, the upper bound is tight 
and the optimal preprocessed state can be chosen as 
\begin{equation}
    \ket{\phi^*} = \sqrt{p^*_{kl}}\ket{k} + \sqrt{1-p^*_{kl}}\ket{l},\quad 
    \ket{\phi^{\perp*}} = \sqrt{1-p^*_{kl}}\ket{1} - \sqrt{p^*_{kl}}\ket{l},
\end{equation}
where $\big\{ \frac{m^{(i)}_k}{m^{(i)}_l},\, 1 \leq i \leq r\big\} = \{\frakm_{kl,1},\frakm_{kl,2}\}$ and $p^*_{kl} = \frac{1}{1 + \sqrt{\frakm_{kl,1} \frakm_{kl,2}}}$. 
Note that in this case $\big\{ \frac{m^{(i)}_k}{m^{(i)}_l},\, 1 \leq i \leq r\big\}$ must contain at least two distinct elements---otherwise, $\{M_i\}$ must be trivial.

Alternatively, we can also prove
\eqref{eq:uppgamma} directly from its original definition (\eqref{eq:gamma-no-x}) without using \thmref{thm:2d}. We have 
\begin{equation}
    \gamma(\{M_i\}) = \max_{\substack{\braket{\phi|\phi}=1,\,\braket{\phi^\perp|\phi^\perp}=1, \\ \braket{\phi|\phi^\perp}=0}} \sum_i \frac{(\Re[\braket{\phi|M_i|\phi^\perp}])^2}{\braket{\phi|M_i|\phi}} = \max_{\substack{\sum_{k} a_k^2 = 1,\,\sum_{k} b_k^2 \leq 1, \\ \sum_{k} a_k b_k = 0}} \sum_i \frac{\left(\sum_k m_k^{(i)}a_kb_k\right)^2}{\sum_k m_k^{(i)} a_k^2},
\end{equation}
where $a_k = \braket{k|\phi}$ (which we assume to be real, without loss of generality) and $b_k = \Re[\braket{k|\phi^\perp}]$ for all $1 \leq k \leq d$. Clearly, the maximum is reached at $\sum_k b_k^2 = 1$. For any $\{a_k\}$ and $\{b_k\}$ satisfying $\sum_{k} a_k^2 = 1$, $\sum_{k} b_k^2 = 1$ and $\sum_{k} a_k b_k = 0$, we have 
\begin{align}
    \sum_i \frac{\left(\sum_k m_k^{(i)}a_kb_k\right)^2}{\sum_k m_k^{(i)} a_k^2} 
    &= \sum_i \frac{ \sum_k (m_k^{(i)})^2 a_k^2b_k^2  + \sum_{k\neq l} m_k^{(i)} m_l^{(i)} a_k b_k a_l b_l }{\sum_k m_k^{(i)} a_k^2}\\
    &= \sum_i \frac{ \left(\sum_k m_k^{(i)} a_k^2\right)\left(\sum_l m_l^{(i)} b_l^2\right) + \frac{1}{2}\sum_{k\neq l} m_k^{(i)} m_l^{(i)} (2 a_k b_k a_l b_l - a_k^2 b_l^2 - a_l^2 b_k^2)  }{\sum_k m_k^{(i)} a_k^2} \\
    &= 1 - \frac{1}{2}\sum_{k\neq l} (a_k b_l - a_l b_k)^2  \sum_i \frac{ m_k^{(i)} m_l^{(i)}  }{ \sum_k m_k^{(i)} a_k^2} \\
    &\leq 1 - \min_{k \neq l} \left(\sum_i \frac{ m_k^{(i)} m_l^{(i)}  }{ \sum_k m_k^{(i)} a_k^2}\right)\left( \sum_{i,k} m_k^{(i)} a_k^2 \right) \leq 1 - \min_{k l}  \left( \sum_i \sqrt{ m_k^{(i)} m_l^{(i)}  } \right)^2,
\end{align}
where we use $\sum_{i,l} m_l^{(i)} b_l^2 = 1$ in the second step, along with the identities $\sum_{k} a_k^2 = \sum_{l} b_l^2 = 1$, and $\sum_{k} a_k b_k = 0$ to simplify $\frac{1}{2}\sum_{k\neq l} (a_k b_l - a_l b_k)^2 = \frac{1}{2}(\sum_{kl} a_k^2 b_l^2 + a_l^2 b_k^2 - 2a_k b_k a_l b_l) = 1$ in the third step, which allows us to minimize the remaining expression, multiplying by $1=\sum_{i,k} m_k^{(i)} a_k^2$ to leverage the Cauchy--Schwarz inequality in the last step. \eqref{eq:uppgamma} is then proven.

\end{proof}

\section{Classical capacity of quantum channels}
\label{app:capacity}

In this section, we prove the following lemma, which is generally known to be true in quantum information theory (see e.g.~\cite{hayashi2007error}): 
\begin{lemma}
\label{lemma:capacity}
Consider a quantum channel $\Phi$, its classical channel capacity $C(\Phi)$ and a constant $\alpha$ satisfying $0 < \alpha < C(\Phi)$. Then for all but finitely many positive integers $n$, there exist channels $\Xi_E$ and $\Xi_D$, such that $$\|\Xi_D \circ \Phi^{\otimes n} \circ \Xi_E - \mD_2^{\otimes \lfloor \alpha n \rfloor}\|_\diamond \leq e^{-\beta n},$$ for some $\beta > 0$, with $\mD_2$ and $\|\|_\diamond$ defined in the proof of {\thmref{thm:metrology-capacity}}.
\end{lemma}

\lemmaref{lemma:capacity} essentially states that, fixing any $\alpha$ that is smaller than the classical channel capacity of $\Phi$, in the large $n$ limit,  $\Phi^{\otimes n}$ with suitable encoding and decoding channels can be used to transmit classical binary information reliably at a rate $\alpha$ with an exponentially small error with respect to $n$. The proof of \lemmaref{lemma:capacity} follows almost exactly from the proof of the HSW theorem~\cite{holevo1998capacity,schumacher1997sending}, with a slight refinement in the error analysis using Hoeffding's inequality, i.e., 
\begin{lemma}[Hoeffding's inequality~\cite{hoeffding1994probability}]
\label{lemma:hoeffding}
Let $X_1,\ldots,X_n$ be independent random variables such that $0 \leq X_i \leq 1$ for $i = 1,\ldots,n$. Then for any $\varepsilon > 0$, 
\begin{equation}
    Y_n = \frac{X_1+\cdots+X_n}{n}\quad \text{satisfies} \quad  \prob\left(\abs{Y_n - \bE\left[Y_n\right]} \geq \varepsilon\right) \leq 2 e^{-2n\varepsilon^2}, 
\end{equation}
\end{lemma}
\noindent to show that the error is exponentially small. Therefore, we will only provide the error analysis part for the proof of \lemmaref{lemma:capacity} that is different from the standard proof of the HSW theorem and skip the steps that can readily be found in standard quantum information theory textbooks~\cite{watrous2018theory,wilde2013quantum}.

Following the proof of Theorem~8.27 in~\cite{watrous2018theory}, it is clear that in order to prove \lemmaref{lemma:capacity}, it is sufficient to prove the following lemma (which is a refinement of Theorem~8.26 in \cite{watrous2018theory}):
\begin{lemma}
Let $\eta = (p(a),\sigma_a)$ be an ensemble of quantum states satisfying $\sum_a p(a) = 1$, where $\sigma_a$ are density operators and $a \in \Sigma$ ($\Sigma$ is an alphabet whose order is equal to the square of the dimension of the system the $\sigma_a$ act on). Let 
\begin{equation}
    \alpha < \chi(\eta) := H\Big(\sum_{a\in\Sigma} p(a)\sigma_a\Big) - \sum_{a\in\Sigma} p(a)H(\sigma_a), 
\end{equation}
where $H(\sigma_a)$ is the von Neumann entropy of $\sigma_a$ and $m = \lfloor \alpha n \rfloor$. For all but finite number of $n$, there exists a function $f:\{0,1\}^m \rightarrow \Sigma^n$ and a quantum measurement $\{M_\vb\}_{\vb\in\{0,1\}^m}$ such that 
\begin{equation}
\trace(M_{\vb}\sigma_{f(\vb)}) > 1- e^{-\beta n}, 
\end{equation}
for every $\vb = b_1\cdots b_m \in \{0,1\}^m$, $\sigma_{f(\vb)} = \sigma_{f(\vb)_1}\otimes \cdots \sigma_{f(\vb)_n}$ and some $\beta > 0$. 
\end{lemma}

\begin{proof}
Choose a sufficiently small $\varepsilon$ such that $\alpha < \chi(\eta) - 3\varepsilon$. Following the proof of Theorem~8.26 in \cite{watrous2018theory}, there exists a function $f:\{0,1\}^m \rightarrow \Sigma^n$ and a quantum measurement $\{M_\vb\}_{\vb\in\{0,1\}^m}$ such that 
\begin{equation}
\trace(M_{\vb}\sigma_{f(\vb)}) > 1- \delta, 
\end{equation}
for every $\vb = b_1\cdots b_m \in \{0,1\}^m$ where 
\begin{equation}
\label{eq:capacity-0}
    \delta = 4\bigg(3 - 2\trace(\Pi\sigma^{\otimes n}) - \sum_{\va \in \Sigma^{ n}} p(a_1) \cdots p(a_n) \trace(\Lambda_{\va}\sigma_{\va})\bigg) + 2^{m+4-n(\chi(\eta) - 2\varepsilon)}.
\end{equation}
Here $\va = a_1\cdots a_n \in \Sigma^n$, $\sigma = \sum_{a\in\Sigma}p(a)\sigma_a$, $\sigma_{\va} = \sigma_{a_1}\otimes\cdots\otimes \sigma_{a_m}$, $\Pi$ is the projection onto the $\varepsilon$-typical subspace with respect to $\sigma$, $\Lambda_{\va}$ is the projection onto the $\varepsilon$-typical subspace conditioned on $\va = a_1\cdots a_n$. 
Specifically, let $\sigma = \sum_a p'(a) \ket{u_a}\bra{u_a}$ where $\{\ket{u_a},a\in\Sigma\}$ is an orthonormal basis and $p(a)\sigma_a = \sum_{c\in\Gamma} p(a,c) \ket{u_{ac}}\bra{u_{ac}}$ where $\{\ket{u_{ac}},a\in\Gamma\}$ is an orthonormal basis for each $a \in \Sigma$. Let $p(a) = \sum_{c\in\Gamma} p(a,c)$ and $H(p(a))$ be the Shannon entropy of $p(a)$. Then the definitions of $\Pi$ and $\Lambda_{\va}$ are 
\begin{gather}
    \Pi = \sum_{\va \in T_\varepsilon} \ket{u_{a_1}}\bra{u_{a_1}}\otimes\cdots\otimes\ket{u_{a_n}}\bra{u_{a_n}},\\
    \Lambda_{\va} = \sum_{\vc \in K_{\va,\varepsilon}} \ket{u_{a_1c_1}}\bra{u_{a_1c_1}}\otimes\cdots\otimes\ket{u_{a_nc_n}}\bra{u_{a_nc_n}},
\end{gather}
where $T_\varepsilon$ is the set of $\va$ satisfying $2^{-n(H(p'(a)) + \varepsilon)} < p'(a_1)\cdots p'(a_n) < 2^{-n(H(p'(a)) - \varepsilon)}$ and $K_{\va,\varepsilon}$ is the set of $\vc$ satisfying $2^{-n(H(p(a,c)) - H(p(a)) + \varepsilon)} < \frac{p(a_1,c_1) \cdots p(a_n,c_n)}{p(a_1)\cdots p(a_n)} < 2^{-n(H(p(a,c)) - H(p(a)) - \varepsilon)}$ for any $\va$ satisfying $p(a_1)\cdots p(a_n) > 0$. We have 
\begin{gather}
\label{eq:capacity-1}
    \trace(\Pi\sigma^{\otimes n})  = \sum_{\va\in T_\varepsilon} p'(a_1)\cdots p'(a_n),\\
\label{eq:capacity-2}
    \sum_{\va\in\Sigma^n}\trace(\Lambda_{\va}\sigma_{\va}) = \sum_{\va\in\Sigma^n}\sum_{\vc \in K_{\va,\varepsilon}} p(a_1,c_1)\cdots p(a_n,c_n). 
\end{gather}
Define two random variables $X: \Sigma \rightarrow [0,x_\upp]$ and $Y: \Sigma \times \Gamma \rightarrow [0,y_\upp]$, where $x_\upp = \max_{a:p'(a)\neq 0} -\log(p'(a))$ and $y_\upp = \max_{a,b:p(a,c)\neq 0} -\log(p(a,c))+\log(p(a))$, as  
\begin{gather}
    X(a) = -\log(p'(a)) \text{ if } p'(a) > 0, \text{ and } 0 \text{ otherwise},\\
    Y(a,b) = -\log(p(a,c))+\log(p(a)) \text{ if } p(a,c) > 0, \text{ and } 0 \text{ otherwise}. 
\end{gather}
Let $X_1,\ldots,X_n$ be $n$ independent random variables each identically distributed to $X$. Using the Hoeffding inequality, we have 
\begin{gather}
    \prob\left(\left|\frac{X_1+\cdots +X_n}{n} - H(p'(a))\right| \geq \varepsilon \right) \leq 2\exp\left(-{2n\varepsilon^2}/{x_\upp^2}\right).
\end{gather}
On the other hand, according to the definition of $T_\varepsilon$, we have
\begin{equation}
    \prob\left(\left|\frac{X_1+\cdots +X_n}{n} - H(p'(a))\right| \geq \varepsilon \right) = 1 - \sum_{\va\in T_\varepsilon} p'(a_1)\cdots p'(a_n). 
\end{equation}
It implies $\trace(\Pi \sigma^{\otimes n}) \geq 1 - 2\exp\left(-{2n\varepsilon^2}/{x_\upp^2}\right)$ using \eqref{eq:capacity-1}. Similarly, using the Hoeffding inequality for independent random variables distributed to $Y$ and \eqref{eq:capacity-2}, we have $\sum_{\va\in\Sigma^n}\trace(\Lambda_{\va} \sigma_{\va}) \geq 1 - 2\exp\left(-{2n\varepsilon^2}/{y_\upp^2}\right)$. Plugging in these bounds in \eqref{eq:capacity-0}, we have 
\begin{equation}
    \delta \leq 16 e^{-2n\varepsilon^2/x_{\upp}^2} + 8 e^{-2n\varepsilon^2/y_{\upp}^2} + 2^{4-n\varepsilon},
\end{equation}
proving the lemma. 
\end{proof}

\section{Examples of global preprocessing controls}
\label{app:circuit}

\subsection{Achievable FIs}

We first calculate the asymptotic limits of the achievable FI using circuits shown in \figref{fig:circuit}. 

The first example is phase sensing using GHZ states. The initial state is 
\begin{equation}
    \ket{\psi_\theta^{(n)}} = \frac{e^{i n \theta}\ket{0}^{\otimes n} + e^{-i n \theta}\ket{1}^{\otimes n}}{\sqrt{2}}. 
\end{equation}
The optimal circuit $U^{\textsc{G}}$ is composed of a C(NOT)$^{n-1}$ gate, a Hadamard gate and another C(NOT)$^{n-1}$ gate. The final state is 
\begin{equation}
    \ket{\psi_\theta^{(n),{\rm final}}} = U^{\textsc{G}} \ket{\psi_\theta^{(n)}} = \cos(n\theta)\ket{0}^{\otimes n} + i\sin(n\theta)\ket{1}^{\otimes n}. 
\end{equation}
The noisy measurement is $\{M_i\}^{\otimes n} = \{M_0,M_1\}^{\otimes n}$ where $M_0 = (1-m)\ket{0}\bra{0} + m\ket{1}\bra{1}$ ($0 < m < 1/2$) and $M_1 = \id - M_0$. The probability of getting a measurement result $\vb \in \{0,1\}^n$ is 
\begin{equation}
    \prob(\vb) = \prob(b_1b_2\cdots b_n) = \cos(n\theta)^2  m^{\abs{\vb}}(1-m)^{n-\abs{\vb}} + \sin(n\theta)^2 m^{n-\abs{\vb}}(1-m)^{\abs{\vb}},
\end{equation}
where $\abs{\vb}$ is the weight of $\vb$. We divide the measurement outcomes into two sets where the outcome is $1$ when $\abs{\vb}$ is larger than $\lfloor n/2\rfloor$ and the outcome is $0$ when $\abs{\vb}$ is smaller than or equal to $\lfloor n/2\rfloor$. We call this a \emph{majority voting post-processing method}. Let 
\begin{equation}
    M_{0,{\mj}} = \sum_{\abs{\vb}\leq \lfloor n/2\rfloor} M_{b_1}\otimes \cdots \otimes M_{b_n},\quad \text{and}\quad M_{1,{\mj}} = \id - M_{0,{\mj}} = \sum_{\abs{\vb} > \lfloor n/2\rfloor} M_{b_1}\otimes \cdots \otimes M_{b_n}. 
\end{equation}
Then the probability of getting outcome $0$ is 
\begin{equation}
    \prob(0)_{\mj} = \trace(\psi_\theta^{(n),{\rm final}} M_{0,\mj}) 
    = \cos(n\theta)^2 \prob\left(Y_n \leq \frac{\lfloor n/2\rfloor}{n}\right) + \sin(n\theta)^2  \prob\left(Y_n > 1 - \frac{\lfloor n/2\rfloor}{n}\right). 
\end{equation}
Here we let $X_1,\ldots,X_n$ be i.i.d. random variables such that $X_i = 0$ with probability $1-m$ and $X_i = 1$ with probability $m$ and $Y_n = (X_1+\cdots+X_n)/n$. According to the Hoeffding's inequality (\lemmaref{lemma:hoeffding}), 
\begin{equation}
    \prob\left(\abs{Y_n - m} > \varepsilon \right) \leq 2 e^{-2n\varepsilon^2}.
\end{equation}
For a sufficiently large $n$, we have $\lfloor n/2\rfloor/n - m = \Omega(1)$, and 
\begin{gather}
    1 - 2 e^{-2n(\lfloor n/2\rfloor/n - m)^2} \leq \prob\left(Y_n \leq \frac{\lfloor n/2\rfloor}{n}\right) \leq 1, \\
    0 \leq \prob\left(Y_n > 1 - \frac{\lfloor n/2\rfloor}{n}\right) \leq 2 e^{-2n(1 - \lfloor n/2\rfloor/n - m)^2}.
\end{gather}

Without loss of generality, we assume that $\frac{\pi}{6n} \leq \theta \leq \frac{\pi}{3n}$. Otherwise, we can always first find a rough estimate of  $\theta_0$ such that $\abs{\theta - \theta_0} \leq \frac{\pi}{12n}$ and then insert a Pauli-X rotation $e^{-in\theta_0 X}$ after the Hadamard gate, such that the final state becomes 
\begin{equation}
    \cos(n(\theta-\theta_0))\ket{0}^{\otimes n} + i\sin(n(\theta-\theta_0))\ket{1}^{\otimes n},
\end{equation}
and $\frac{\pi}{6n} \leq \theta-\theta_0 \leq \frac{\pi}{3n}$. Then one can estimate $\theta - \theta_0$ using the majority voting post-processing method. 

Assuming $\frac{\pi}{6n} \leq \theta \leq \frac{\pi}{3n}$, the achievable FI is  
\begin{align}
    F(\psi_\theta^{(n),{\rm final}},\{M_{0,\mj},M_{1,\mj}\}) 
    &= \frac{(\partial_\theta \prob(0)_{\mj})^2}{\prob(0)_{\mj}(1-\prob(0)_{\mj})}\\
    &= \frac{4 n^2   }{1 + \frac{(2-b_n)b_n-a_n^2}{a_n^2 \sin^2(2n\theta)} + \frac{(2-2b_na_n) \cos(2n\theta)}{a_n^2 \sin^2(2n\theta)}  }, 
\end{align}
where $a_n = \prob\left(Y_n \leq \frac{\lfloor n/2\rfloor}{n}\right) - \prob\left(Y_n > 1 - \frac{\lfloor n/2\rfloor}{n}\right)$ and $b_n = \prob\left(Y_n \leq \frac{\lfloor n/2\rfloor}{n}\right) + \prob\left(Y_n > 1 - \frac{\lfloor n/2\rfloor}{n}\right)$. Noting that $\sin^2(2n\theta) \in [3/4,1]$, $\lim_{n\rightarrow \infty} a_n = \lim_{n\rightarrow \infty} b_n = 1$, we have 
\begin{equation}
    F(\psi_\theta^{(n),{\rm final}},\{M_{0,\mj},M_{1,\mj}\}) \xrightarrow{n\rightarrow \infty} 4n^2. 
\end{equation}

To illustrate the importance of global preprocessing controls, we also compute the FI when we measure $\psi_\theta^{(n)}$ using a noisy local Pauli-X operator measurement, which is equivalent to applying a transversal Hadamard gate $H^{\otimes n}$ as the preprocessing control. Let $M_0' = (1-m)\ket{+}\bra{+}+m\ket{-}\bra{-}$ and $M_1' = m\ket{+}\bra{+}+(1-m)\ket{-}\bra{-}$, where $\ket{\pm} = \frac{\ket{0} \pm \ket{1}}{\sqrt{2}}$ is the basis of Pauli-X operator. We have 
\begin{align}
F(\psi_\theta^{(n)} ,\{M_0', M_1'\}^{\otimes n}) &= F(H^{\otimes n}\psi_\theta^{(n)} H^{\otimes n},\{M_i\}^{\otimes n}) \\
&= \frac{(\partial_\theta p_{\rm odd})^2}{p_{\rm odd}} + 
\frac{(\partial_\theta p_{\rm even})^2}{p_{\rm even}}= \frac{4n^2 \sin(2n\theta)^2}{\sin(2n\theta)^2 - 1 + (1-2m)^{-2n}},
\end{align}
where $p_{\rm odd}$ (or $p_{\rm even}$) is the probability of obtaining measurement outcomes that form an odd (or even) parity bit string, and 
\begin{align}
    p_{\rm odd}  = \frac{1}{2}\left(1 + (1-2m)^n\right)\sin^2 n\theta + \frac{1}{2}\left(1 - (1-2m)^n\right)\cos^2 n\theta, \\
    p_{\rm even} = \frac{1}{2}\left(1 - (1-2m)^n\right)\sin^2 n\theta + \frac{1}{2}\left(1 + (1-2m)^n\right)\cos^2 n\theta. 
\end{align}
Note that when $m = 0$, the FI $F(H^{\otimes n}\psi_\theta^{(n)} H^{\otimes n},\{M_i\}^{\otimes n}) = 4n^2$ is equal to the QFI. However, when $m > 0$,  the FI $F(H^{\otimes n}\psi_\theta^{(n)} H^{\otimes n},\{M_i\}^{\otimes n}) = e^{-\Omega(n)}$ is exponentially small, demonstrating the necessity of performing global preprocessing controls. In general, it was illustrated in Ref.~\cite{len2021quantum} that in the presence of noisy measurements, Heisenberg scaling cannot be achieved with local control.

The second example is phase sensing using product states, where 
\begin{equation}
    \ket{\psi_\theta^{(n)}} = \left(\frac{e^{i \theta} \ket{0} + e^{-i \theta}\ket{1}}{\sqrt{2}}\right)^{\otimes n}.
\end{equation}
The optimal circuit is composed of a global Hadamard gate, a global Pauli-X rotation of angle $\theta_0/2$, a desymmetrization gate $DS$ and a C(NOT)$^{n-1}$ gate. Here $\theta_0$ is chosen such that $1/n < \abs{\theta - \theta_0} \ll 1/\sqrt{n}$ (assuming $n$ is large enough). After the first two gates, the quantum state becomes 
\begin{equation}
    \left(\cos(\theta-\theta_0)\ket{0} + i\sin(\theta-\theta_0)\ket{1}\right)^{\otimes n}.
\end{equation}
The desymmetrization gate $DS$ is a unitary gate such that 
\begin{equation}
    \ket{0}^{\otimes n} \mapsto \ket{0}^{\otimes n}, \quad 
    \ket{W} = \frac{1}{\sqrt{n}}(\ket{10\cdots 0} + \ket{010\cdots 0} + \cdots + \ket{0\cdots 01}) \mapsto \ket{10\cdots 0}. 
\end{equation}
After $DS$ and C(NOT)$^{n-1}$ gates, the final quantum state is 
\begin{equation}
    \ket{\psi_\theta^{(n),\mathrm{final}}} = \cos^n(\theta-\theta_0) \ket{0}^{\otimes n} + i\sqrt{n}\sin(\theta-\theta_0)\cos^{n-1}(\theta-\theta_0)  \ket{1}^{\otimes n} + \cdots,
\end{equation}
where we omit in ``$\cdots$'' a state perpendicular to $\ket{0}^{\otimes n}$ and $\ket{1}^{\otimes n}$ whose norm is of magnitude $O(n(\theta-\theta_0)^2)$. Using the majority voting method as in the first example, we have 
\begin{gather}
    \prob(0)_{\mj} = \trace(\psi_\theta^{(n),{\rm final}} M_{0,\mj}) 
    = \cos^{2n}(\theta-\theta_0) \prob\left(Y_n \leq \frac{\lfloor n/2\rfloor}{n}\right) ~~~~~~~~~~~~~~~~~~~~~~~~~~\\~~~~~~~~~~~~~+ n \sin(\theta-\theta_0)^2 \cos^{2(n-1)}(\theta-\theta_0) \prob\left(Y_n > 1 - \frac{\lfloor n/2\rfloor}{n}\right) + O(n^2(\theta-\theta_0)^4),\\
    \prob(1)_{\mj} = \trace(\psi_\theta^{(n),{\rm final}} M_{1,\mj}) 
    = \cos^{2n}(\theta-\theta_0) \prob\left(Y_n > \frac{\lfloor n/2\rfloor}{n}\right) ~~~~~~~~~~~~~~~~~~~~~~~~~~\\~~~~~~~~~~~~~+ n \sin(\theta-\theta_0)^2 \cos^{2(n-1)}(\theta-\theta_0) \prob\left(Y_n \leq 1 - \frac{\lfloor n/2\rfloor}{n}\right) + O(n^2(\theta-\theta_0)^4). 
\end{gather}

For a sufficiently large $n$, we have $\lfloor n/2\rfloor/n - m = \Omega(1)$, and the achievable FI is 
\begin{align}
    F(\psi_\theta^{(n),{\rm final}},\{M_{0,\mj},M_{1,\mj}\}) 
    &= \frac{(\partial_\theta \prob(0)_{\mj})^2}{\prob(0)_{\mj}(1-\prob(0)_{\mj})}\\
    &= \frac{(2n\sin(\theta-\theta_0)\cos^{2n-1}(\theta-\theta_0) + e^{-\Omega(n)} +  O(n^2(\theta-\theta_0)^3))^2}{n\cos^{4n-2}(\theta-\theta_0)\sin(\theta-\theta_0)^2 + e^{-\Omega(n)}  + O(n^2(\theta-\theta_0)^4) }. 
\end{align}
Note that the first terms in the numerator and the denominator must be dominant terms because we have assumed $\abs{\theta - \theta_0} \ll 1/n$. Therefore we have, 
\begin{equation}
    F(\psi_\theta^{(n),{\rm final}},\{M_{0,\mj},M_{1,\mj}\}) \xrightarrow{n\rightarrow \infty} 4n. 
\end{equation}

The third example is phase sensing using classically mixed states, where 
\begin{equation}
    \rho_\theta^{(n)} = \left(\cos^2\theta \ket{0}\bra{0} + \sin^2\theta \ket{1}\bra{1}\right)^{\otimes n}. 
\end{equation}
Our preprocessing circuit is composed of a sorting channel $\mS_{\rm sorting}$, discarding $n-1$ qubits and preparing $\ket{0}^{\otimes n-1}$ and a C(NOT)$^{n-1}$ gate. In particular, the sorting channel $\mS_{\rm sorting}$ first performs a sorting network such that 
\begin{equation}
    \ket{b_1b_2\cdots b_n} \mapsto \ket{1^{\abs{\vb}} 0^{n-\abs{\vb}}},
\end{equation}
where $\vb \in \{0,1\}^n$. 
The second step of the sorting channel is to swap the first qubit with the $\lfloor n \sin^2\theta_0 \rfloor$-th qubit, where $\theta_0$ is chosen such that $\abs{\theta - \theta_0} \ll 1/\sqrt{n}$, i.e., $\abs{\theta - \theta_0}\sqrt{n} = o(1)$.  
Note that $\mS_{\rm sorting}$ is not a unitary channel. It can be viewed as an optimal quantum-classical channel as in \thmref{thm:metrology-capacity} that converts the state into a $(n+1)$-level quantum system with no sensitivity loss. We will later specify a circuit implementation of $\mS_{\rm sorting}$ using $O(n\log^2 n)$ ancillary qubits. 

After discarding all but the first probe qubit, we have 
\begin{equation}
    p_\theta\ket{0}\bra{0} + (1-p_\theta) \ket{1}\bra{1}, 
\end{equation}
where $p_\theta = \prob(Y^{\theta}_n \leq \lfloor n \sin^2\theta_0 \rfloor/n)$ is the probability that after flipping $n$ biased coins whose head probability is $\sin^2\theta$, the number of heads are smaller than or equal to $\lfloor n \sin^2\theta_0 \rfloor$. 
Here we let $X^{\theta}_1,\ldots,X^{\theta}_n$ be i.i.d. random variables such that $X^{\theta}_i = 0$ with probability $\cos^2(\theta)$ and $X^{\theta}_i = 1$ with probability $\sin^2(\theta)$ and $Y^\theta_n = (X^\theta_1+\cdots+X^\theta_n)/n$. 
In particular, we use the Berry--Esseen theorem which states that 
\begin{lemma}[Berry--Esseen theorem~\cite{berry1941accuracy,esseen1956moment}]
Let $B_1,\cdots,B_n$ be i.i.d. random variables such that $\bE[B_i] = 0$, $\bE[B_i^2] = \sigma^2 > 0$, and $\bE[\abs{B_i}^3] < \infty$. $F_n$ is the cumulative distribution function of $A_n\sqrt{n}/\sigma$, where $A_n = (B_1+B_2+\cdots+B_n)/n$ and $\Phi$ is the cumulative distribution function of the standard normal distribution, i.e., $\Phi(x) = \int_{-\infty}^x \frac{1}{\sqrt{2\pi}} e^{-\frac{1}{2}y^2} dy$. Then for all $x$ and $n$,
\begin{equation}
    \abs{F_n(x) - \Phi(x)} \leq \frac{C}{\sigma^3\sqrt{n}}
\end{equation}
for some constant $C$. 
\end{lemma}
\noindent Then we have 
\begin{align}
    p_\theta 
    &= \prob\left(Y^{\theta}_n \leq \lfloor n \sin^2\theta_0 \rfloor/n\right) = F_n\left( \frac{\sqrt{n}\left( \frac{\lfloor n \sin^2\theta_0\rfloor}{n} - \sin^2\theta\right)}{\abs{\cos(\theta)\sin(\theta)}}\right) 
    \\ 
    &= \Phi\left( \frac{\sqrt{n}\left( \frac{\lfloor n \sin^2\theta_0\rfloor}{n} - \sin^2\theta\right)}{\abs{\cos(\theta)\sin(\theta)}}\right) + O\left(\frac{1}{\sqrt{n}}\right) = \frac{1}{2} + o\left(1\right),
\end{align}
where we use $\Phi(0) = 1/2$ and $\frac{\lfloor n \sin^2\theta_0\rfloor}{n} - \sin^2\theta = o(1/\sqrt{n})$. 
On the other hand, 
\begin{align}
    \partial_\theta p_\theta 
    &= \partial_\theta\,\prob\left(Y^{\theta}_n \leq \lfloor n \sin^2\theta_0 \rfloor/n\right) \\
    &= \partial_\theta\,\sum_{0 \leq n_h \leq \lfloor n \sin^2\theta_0 \rfloor } \binom{n}{n_h} \sin^{2n_h}\theta\cos^{2(n-n_h)}\theta\\
    &= \sum_{0 \leq n_h \leq \lfloor n \sin^2\theta_0 \rfloor } \binom{n}{n_h}  \left( 2n_h\sin^{2n_h-1}\theta\cos^{2(n-n_h)+1}\theta - 2(n-n_h) \sin^{2n_h+1}\theta\cos^{2(n-n_h)-1}\theta \right)\\
    &= -2n \binom{n-1}{\lfloor n \sin^2\theta_0 \rfloor} \sin^{2\lfloor n \sin^2\theta_0 \rfloor+1}\theta\cos^{2(n-\lfloor n \sin^2\theta_0 \rfloor)-1}\theta,
\end{align}
and asymptotically we have (using $\simeq$ to mean asymptotically equivalence), 
{\small 
\begin{align}
    \partial_\theta p_\theta  & \simeq -\frac{2n}{\sqrt{2\pi}} \frac{(n-1)^{n-\frac{1}{2}}}{\lfloor n \sin^2\theta_0 \rfloor^{\lfloor n \sin^2\theta_0 \rfloor + \frac{1}{2}}(n-1-\lfloor n \sin^2\theta_0 \rfloor)^{n-\frac{1}{2}-\lfloor n \sin^2\theta_0 \rfloor}}
\\
&\simeq
-\frac{2\sqrt{n}}{\sqrt{2\pi}} \left(\frac{(n-1)\sin^2\theta}{(n-1)\sin^2\theta + c_n}\right)^{(n-1)\sin^2\theta + c_n}\left(\frac{(n-1)\cos^2\theta}{(n-1)\cos^2\theta - c_n}\right)^{(n-1)\cos^2\theta - c_n}\\
&=
-\frac{2\sqrt{n}}{\sqrt{2\pi}} \exp\left( -\left((n-1)\sin^2\theta + c_n\right)\ln\left(1 + \frac{c_n}{(n-1)\sin^2\theta} \right) -\left((n-1)\cos^2\theta - c_n\right)\ln\left(1 - \frac{c_n}{(n-1)\cos^2\theta} \right) \right)\\
&= 
-\frac{2\sqrt{n}}{\sqrt{2\pi}} \exp\left(- \frac{2c_n^2}{(n-1)\sin^2(2\theta)} + O\left(\frac{c_n^3}{n^2}\right)\right)
\simeq
-\frac{2\sqrt{n}}{\sqrt{2\pi}}, 
\end{align}}
\hspace{-0.085in} where we use Stirling's formula~\cite{spencer2014asymptopia} in the first step, use $c_n:= -(n-1)\sin^2\theta + \lfloor n\sin^2\theta_0\rfloor$ in the second step, use the Taylor expansion $\ln(1+x) = x - x^2/2 + O(x^3)$ and $c_n = o(\sqrt{n})$ in the last step. 

After resetting the discard qubits to be $\ket{0}^{\otimes n-1}$ and performing the C(NOT)$^{n-1}$ gate, we have 
\begin{equation}
    \rho_\theta^{(n),\mathrm{final}} = p_\theta\ket{0}^{\otimes n}\bra{0}^{\otimes n} + (1-p_\theta) \ket{1}^{\otimes n}\bra{1}^{\otimes n}. 
\end{equation}
Using the majority voting method, we have 
\begin{gather}
    \prob(0)_{\mj} = \trace(\rho_\theta^{(n),{\rm final}} M_{0,\mj}) 
    = p_\theta \prob\left(Y_n \leq \frac{\lfloor n/2\rfloor}{n}\right) + (1-p_\theta) \prob\left(Y_n > 1 - \frac{\lfloor n/2\rfloor}{n}\right),\\
    \prob(1)_{\mj} = \trace(\rho_\theta^{(n),{\rm final}} M_{1,\mj}) 
    = p_\theta \prob\left(Y_n > \frac{\lfloor n/2\rfloor}{n}\right) + (1-p_\theta) \prob\left(Y_n \leq 1 - \frac{\lfloor n/2\rfloor}{n}\right). 
\end{gather}
For a sufficiently large $n$, we have $\lfloor n/2\rfloor/n - m = \Omega(1)$, and the achievable FI is 
\begin{align}
    F(\rho_\theta^{(n),{\rm final}},\{M_{0,\mj},M_{1,\mj}\}) 
    &= \frac{(\partial_\theta \prob(0)_{\mj})^2}{\prob(0)_{\mj}(1-\prob(0)_{\mj})}\\
    &= \frac{(\partial_\theta p_\theta + e^{-\Omega(n)})^2}{(p_\theta + e^{-\Omega(n)})(1-p_\theta + e^{-\Omega(n)})}, 
\end{align}
Therefore we have, 
\begin{equation}
    F(\rho_\theta^{(n),{\rm final}},\{M_{0,\mj},M_{1,\mj}\}) \xrightarrow{n\rightarrow \infty} \frac{8}{\pi}n. 
\end{equation}
Note that the reason that the QFI is not achieved here lies in the second step where all but one probe qubit are discarded. If, instead, we can perfectly encode the entire $(n+1)$-level system that we obtain in the first step to a $n$-qubit state that is immune to measurement errors, the QFI will be achieved.


\subsection{Gate complexity}

Finally, we discuss the gate complexities of implementing C(NOT)$^{n-1}$, $DS$, and $\mS_{\rm sorting}$. We will show that, assuming arbitrary two-qubit gates and all-to-all connectivity, the C(NOT)$^{n-1}$ and the desymmetrization gate $DS$ can be implemented using $O(n)$ gates in depth $O(\log n)$. The $\mS_{\rm sorting}$ channel can be implemented using $O(n\log^2 n)$ gates and $O(n\log^2 n)$ ancillary qubits in depth $O(\log^2 n)$. 

We first investigate the implementation of C(NOT)$^{n-1}$ gates and $DS$ gates. Note that in experimental platforms like Rydberg atoms where long-range interactions are available, the C(NOT)$^{n-1}$ gate can be implemented in a single step~\cite{muller2009mesoscopic}. However, we focus on the standard quantum circuit model here where only two-qubit gates are allowed. 

Ref.~\cite{cruz2019efficient} included detailed quantum circuits for C(NOT)$^{n-1}$ gates and $DS$ gates using $O(n)$ gates in depth $O(\log n)$. For completeness, we briefly discuss these circuits here, in the case where $n = 2^k$. 

To implement C(NOT)$^{n-1}$, one starts with a Hadamard gate on the first qubit, and then implement C$_1$(NOT)$_2$ (which means a CNOT gate where the control qubit is the first qubit and the target qubit is the second) in the first step, C$_1$(NOT)$_3$ and C$_2$(NOT)$_4$ in the second step, and so on. The circuit continues in the same way. In the final step, i.e., the $k$-th step, C$_l$(NOT)$_{2^{k-1}+l}$ for $l = 1,2,\ldots,2^{k-1}$ are implemented. One can verifies the above $O(\log n)$-depth circuit implements a C(NOT)$^{n-1}$ gate using $O(n)$ single- or two-qubit gates. 

To implement $DS$, one can equivalently consider the circuit implementation of $DS^\dagger$ and then conjugate and reverse the orders of each gate. $DS^\dagger$ is a gate that prepares $W$ states, where  
\begin{equation}
    \ket{10\cdots 0} \mapsto \frac{1}{\sqrt{n}}\big(\ket{10\cdots 0} + \ket{010\cdots 0} + \cdots + \ket{0\cdots 01}\big),\quad \ket{00\cdots 0} \mapsto \ket{00\cdots 0}. 
\end{equation}
To implement $DS^\dagger$, one starts with a Pauli-X gate on the first qubit, then performs a two-qubit gate that is a composition of a C$_1$H$_2$ (controlled-Hadamard) gate and then a C$_2$NOT$_1$ gate (again, we use subscripts $l$ to denote the $l$-th qubit) in the first step. C$_1$H$_3$+C$_3$NOT$_1$ and C$_2$H$_4$+C$_4$NOT$_2$ in the second step, and so on. The circuit continues in the same way. In the final step, i.e., the $k$-th step, C$_l$H$_{2^{k-1}+l}$+C$_{2^{k-1}+l}$NOT$_l$ for $l = 1,2,\ldots,2^{k-1}$ are implemented. One can verifies the above $O(\log n)$-depth circuit implements a $DS^\dagger$ gate using $O(n)$ single- or two-qubit gates. 

Finally, we discuss the implementation of $\mS_{\rm{sorting}}$ which can be decomposed into a sorting network that implements 
\begin{equation}
\label{eq:sorting}
    \ket{b_1b_2\cdots b_n} \mapsto \ket{1^{\abs{\vb}} 0^{n-\abs{\vb}}},
\end{equation}
for $\vb \in \{0,1\}^n$ and a SWAP gate that swaps the first qubit with the $\lfloor n \sin^2\theta_0 \rfloor$-th qubit. Now we discuss the implementation of the sorting network (\eqref{eq:sorting}), which directly follows from a classical sorting network because our input state is a classically mixed state and the sorting channel is incoherent. To be specific, we define a comparator to be a two-qubit quantum channel such that 
\begin{equation}
    \ket{ij} \mapsto 
    \begin{cases}
    \ket{ij} & \text{when }i \geq j, \\
    \ket{ji} & \text{when }j > i. 
    \end{cases}
\end{equation}
It can be implemented using a unitary gate acting on two probe qubits and one ancillary qubit such that 
\begin{equation}
    \ket{ij}\ket{0} \mapsto 
    \begin{cases}
    \ket{ij}\ket{0} & \text{when }i \geq j, \\
    \ket{ji}\ket{1} & \text{when }j > i, 
    \end{cases}
\end{equation}
and discarding the ancillary qubit afterwards. Our sorting network (\eqref{eq:sorting}) then follows from a classical sorting network, replacing all its classical comparators with the two-qubit sorting channels described above. 

Here we use a classical sorting network called a bitonic sorter~\cite{parhami2006introduction} that uses $O(n \log^2 n)$ comparators in depth $O(\log^2 n)$. Note that it is also possible to construct sorting networks of depth $O(\log n)$ (and size $O(n \log^2 n)$)~\cite{ajtai1983O}, although the linear constant is large, making it impractical. We briefly summarize the bitonic sort algorithm in the following pseudocode. Note that here we assume $n = 2^k$ (we can always add more qubits in prepared in $\ket{0}$ to make $n$ a power of $2$).  
\begin{algorithm}[H]
\caption{Bitonic Sort Algorithm}
\begin{algorithmic}[1]
\Require $arr$: the array to be sort. 
\Ensure $arr$: the sorted array in descending order. 
\State $arr$ = \textsc{BitonicSort}($arr$,``descending'')
\Function{BitonicSort}{$arr$, direction}
  \State $n$ = the length of $arr$
  \If{$n > 1$}
    \State \textsc{BitonicSort}($arr$[1:$n/2$],``ascending'')
    \State \textsc{BitonicSort}($arr$[$n/2$+1:$n$],``descending'')
    \State \textsc{Merge}($arr$, direction)
  \EndIf
\EndFunction
\Function{Merge}{$arr$, direction}
  \State $n$ = the length of $arr$
  \If{$n > 1$}
    \For{$i=1,2,\ldots,n/2$}
    \State Exchange $arr$[i] and $arr$[i+n/2] if they are not in the right order
    \EndFor
    \State \textsc{Merge}($arr$[1:$n/2$], direction)
    \State \textsc{Merge}($arr$[$n/2$+1:$n$], direction)
  \EndIf
\EndFunction
\end{algorithmic}
\end{algorithm}
The bitonic sort algorithm can be divided into two steps: (1) forming a bitonic sequence; (2) sorting a bitonic sequence. A bitonic sequence of length $n$ is defined to be a sequence $\vb$ where there is an index $i$ such that $(b_1,\ldots,b_i)$ is monotonically non-decreasing, and $(b_{i+1},\ldots,b_n)$ is monotonically non-increasing, or a sequence that can be cyclically shifted into the above sequence. \textsc{Merge}($arr$, direction) sorts a bitonic sequence $arr$ in the required order. It can be proven that the \textsc{Merge} function contains $O(\log n)$ parallel computing steps; and the  \textsc{BitonicSort} function contains $O(\log^2 n)$ parallel computing steps. As a result, one can see that the bitonic sorter uses $O(n \log^2 n)$ comparators in depth $O(\log^2 n)$.

\end{document}